\documentclass[a4paper, 11pt]{article}

\usepackage{standalone, lmodern, xcolor, enumitem}
\usepackage{amsmath, amssymb, amsthm, physics, bbm, mathtools, stmaryrd, leftidx}
\usepackage{thmtools}
\usepackage{hyperref}
\usepackage{comment}

\usepackage{pgfplots}
\pgfplotsset{compat=1.16}

\usetikzlibrary{arrows, decorations.pathreplacing, cd, patterns, calc, decorations.markings, backgrounds}

\usepackage{quiver}

\usepackage[new]{old-arrows}
\usepackage{changepage}

\usepackage[bottom]{footmisc}

\theoremstyle{plain}
\newtheorem{theorem}{Theorem}[section]
\newtheorem{proposition}[theorem]{Proposition}
\newtheorem{lemma}[theorem]{Lemma}
\newtheorem{corollary}[theorem]{Corollary}

\numberwithin{equation}{section}
\theoremstyle{definition}
\newtheorem{definition}[theorem]{Definition}

\theoremstyle{remark}
\newtheorem{remark}[theorem]{Remark}

\numberwithin{figure}{section}

\setlist[enumerate]{topsep = 1ex, leftmargin=.7cm, itemsep= -2pt}

\allowdisplaybreaks

\definecolor{color1}{HTML}{1b9e77}
\definecolor{color2}{HTML}{F1A340}
\definecolor{color3}{HTML}{5e3c99}

\hypersetup{
  colorlinks = true, 
  urlcolor = color3, 
  linkcolor = color1, 
  citecolor = color3 
}

\newcommand{\R}{\mathbb{R}}
\newcommand{\Rp}{\mathbb{R}_{+}}
\newcommand{\Z}{\mathbb{Z}}
\newcommand{\C}{\mathbb{C}}
\newcommand{\vect}{\mathfrak{X}}
\DeclareMathOperator{\id}{{\mathbbm{1}}}
\DeclareMathOperator{\inv}{inv}
\DeclareMathOperator{\Lie}{Lie}
\newcommand{\vol}{\mathrm{vol}}
\newcommand{\dz}{\dd z \dd \bar z}
\newcommand{\detz}[1]{{\det}_\zeta \, {#1}}

\renewcommand{\tilde}[1]{\widetilde{#1}}
\newcommand{\blank}{\:\cdot\:}
\newcommand{\setsuchthat}[2]{\left\{ {#1} \:\middle|\: {#2} \right\}}
\newcommand{\setsuchthatinline}[2]{\{ {#1} \; |\; {#2} \}}
\newcommand{\restrict}[2]{{#1}\vert_{{#2}}}

\DeclareMathOperator{\Diffpan}{{Diff}_+^{\an}(S^1)}
\DeclareMathOperator{\DiffC}{{Def}_{\C}(S^1)}
\DeclareMathOperator{\Diffext}{\red{{Det}^\charge_{+}}}
\DeclareMathOperator{\Vir}{\mathfrak{Vir}_\charge}
\DeclareMathOperator{\VectR}{\vect^{\an}_{\R}(S^1)}
\DeclareMathOperator{\Witt}{\vect^{\an}_{\C}(S^1)}
\newcommand{\Detrc}{\Det_\R^\charge}
\newcommand{\Detrpc}{\Det_{\Rp}^{\charge}}

\newcommand{\moduli}[2]{\mathcal{M}_{{#1}, {#2}}}
\newcommand{\surfaces}[2]{\mathcal{B}_{{#1}, {#2}}}
\newcommand{\cylinders}{\surfaces{0}{2}}
\newcommand{\genus}{\mathsf{g}}
\newcommand{\boundaries}{\mathsf{b}}
\newcommand{\univteichm}{T(1)}

\newcommand{\intsets}{\mathcal{U}_{\R \times {S^1}}}

\newcommand{\an}{\mathrm{an}}
\newcommand{\charge}{\mathbf{c}}
\DeclareMathOperator{\Arg}{Arg}
\DeclareMathOperator{\Ln}{Ln}
\DeclareMathOperator{\Hom}{Hom}
\DeclareMathOperator{\Det}{Det}
\DeclareMathOperator{\Aut}{Aut}
\DeclareMathOperator{\kernel}{Ker}
\DeclareMathOperator{\coker}{Coker}
\DeclareMathOperator{\conf}{Conf}
\DeclareMathOperator{\cconf}{Adm}
\DeclareMathOperator{\End}{End}
\DeclareMathOperator{\vspan}{span}
\DeclareMathOperator{\Exp}{Exp}
\newcommand{\D}{\mathbb{D}}
\newcommand{\A}{\mathbb{A}}
\newcommand{\F}{\mathbb{F}}
\newcommand{\sH}{\mathscr{H}}
\newcommand{\boldone}{\mathbbm{1}}
\newcommand{\sew}[2]{\:\leftidx{{_#1}}{{\infty}}{{_#2}}\:}
\newcommand{\usew}[2]{\:\leftidx{{_#1}}{{\#}}{{_#2}}\:}
\newcommand{\sewx}[3]{\:\leftidx{{_{#2}}}{{\overset{{#1}}{\infty}}}{{_{#3}}}\:}
\newcommand{\doubletilde}[1]{\accentset{\approx}{#1}}
\DeclareMathOperator{\psl}{PSL(2, \R)}
\newcommand{\liou}[2]{\,\llbracket{\, #1}\, , {\,#2\,}\rrbracket\,}
\DeclareMathOperator{\lfunct}{S_L^0}
\newcommand{\lioub}[2]{\lfunct({#1}, {#2})}
\newcommand{\dersq}[1]{F_{#1}}

\newcommand{\Ua}{{U}}
\newcommand{\Uac}{{U^c}}
\newcommand{\Ub}{{V_t}}
\newcommand{\Ubc}{{V^c_t}}
\newcommand{\isom}[2]{{\operatorname{I}_{#2}^{#1}}}

\newcommand{\intera}{{\chi_{2}}}
\newcommand{\interb}{{\varrho_t}}
\newcommand{\interab}{{\chi_{1}}}

\newcommand{\interaPrime}{{\chi'_2}}
\newcommand{\interaPrimePrime}{{\chi''_2}}
\newcommand{\interbPrime}{{\varrho'_t}}
\newcommand{\interabPrime}{{\chi'_{1}}}
\newcommand{\interabPrimePrime}{{\chi''_{1}}}

\newcommand{\detmult}[2]{\operatorname{m}_{#1, #2}}
\newcommand{\detproj}[1]{{\operatorname{p}_{{#1}}}}
\newcommand{\detprojinv}[1]{{\operatorname{p}^{-1}_{{#1}}}}
\newcommand{\natisophi}[3]{\operatorname{I}^{#2}_{#3, #1}}
\newcommand{\globalsection}{\mu^\charge}
\newcommand{\globalsectionphi}{\hat{\mu}^\charge}
\newcommand{\globalsectionzeta}{\mu^\charge_\zeta}
\newcommand{\algcocycle}{\gamma_{\charge}}
\newcommand{\gelfandfuks}{\omega_{\charge}}
\newcommand{\calpha}{\Gamma_\charge}
\newcommand{\rconv}{R}
\newcommand{\rconvunif}{R}
\newcommand{\deteval}{\mathsf{ev}_\gamma}
\newcommand{\detevalcircle}{\mathsf{ev}_{S^1}}

\newcommand{\sewiso}[2]{\operatorname{S}_{#1, #2}}

\newcommand{\exchange}{\mathsf{flip}}
\newcommand{\evaluation}{\mathsf{ev}}
\newcommand{\applyat}[3]{#1_{#2, #3}}

\newcommand{\paramcyl}[3]{
\mathchoice
    {\big(\: #1, \quad #2, \quad #3 \:\big)}
    {(#1, #2, #3)}
    {(#1, #2, #3)}
    {(#1, #2, #3)}
}

\newcommand{\sid}[1]{{\color{orange}{#1}}}

\newsavebox\CBox
\newcommand\hcancel[2][0.5pt]{%
  \ifmmode\sbox\CBox{$#2$}\else\sbox\CBox{#2}\fi%
  \makebox[0pt][l]{\usebox\CBox}%
  \rule[0.5\ht\CBox-#1/2]{\wd\CBox}{#1}}
  
\newcommand{\sidremove}[1]{{\color{green} $\hcancel{\text{#1}}$}}

\newcommand{\red}{\textcolor{red}}
\newcommand{\blue}{\textcolor{blue}}
\newcommand{\green}{\textcolor{green}}

\newcommand{\ii}{\mathrm{i}\,}
\newcommand{\cancyl}[1]{\mathcal{U}{#1}}
\newcommand{\unif}[1]{\mathcal{F}_{#1}}
\newcommand{\unifsq}[1]{U_{#1}}

\newcommand{\bdryint}[1]{\mathsf{N}({#1})}

\newcommand{\Znn}{\mathbb{Z}_{\geq0}}
\newcommand{\Zpos}{\mathbb{Z}_{>0}}

\makeatletter
\newcommand{\uset}[3][-0.35ex]{%
  \mathrel{\mathop{#3}\limits_{
    \vbox to#1{\kern-7\ex@
    \hbox{$\scriptscriptstyle#2$}\vss}}}}
\makeatother

\newcommand{\diffActingInline}[1]{\,
    \smash{\uset{{#1}}{*}} 
\, }
\newcommand{\diffActing}[1]{\diffActingInline{#1}}

\newcommand{\rotnumber}{\rho}

\DeclareRobustCommand\longtwoheadrightarrow
 {\relbar\joinrel\twoheadrightarrow}

\title{From the Conformal Anomaly to the Virasoro Algebra}  
\date{}

\author{Sid Maibach\thanks{Institute for Applied Mathematics, University of Bonn, Endenicher Allee 60, 53115 Bonn, Germany.  
\protect\url{maibach@uni-bonn.de}} \;
and \,
Eveliina Peltola\thanks{Department of Mathematics and Systems Analysis, Aalto University, Otakaari 1, 02150 Espoo, Finland, and Institute for Applied Mathematics, University of Bonn, Endenicher Allee 60, 53115 Bonn, Germany.  
\protect\url{eveliina.peltola@aalto.fi}}}

\begin{document}
\maketitle

\begin{center}
\begin{minipage}{0.95\textwidth}
\abstract{
The conformal anomaly and the Virasoro algebra are fundamental aspects of 2D conformal field theory and conformally covariant models in planar random geometry. 
In this article, we explicitly derive the Virasoro algebra from an axiomatization of the conformal anomaly in terms of real determinant lines, 
one-dimensional vector spaces associated to Riemann surfaces with analytically parametrized boundary components.
Here, analytical orientation-preserving diffeomorphisms and deformations of the circle naturally act on the boundary components. 
We introduce a sewing operation on the real determinant lines over the semigroup of annuli, 
which then induces central extensions of the diffeomorphism group, as well as of the complex deformations.

Our main theorem shows that on the one hand, the cocycle associated to the central extension of diffeomorphisms is trivial, while on the other hand, 
the Lie algebra cocycle associated to the central extension of complex deformations is nontrivial, yielding the imaginary part of the Gel'fand-Fuks cocycle.
We thus answer a question, partly negatively and partly affirmatively, discussed by Andr\'e Henriques and Dylan Thurston in 2011.
The proof uses concrete computations, which we aim to be accessible to a wide audience.

We also show an explicit relation to loop Loewner energy, anticipating the real determinant lines to be pertinent to 
locally conformally covariant (Malliavin--Kontsevich--Suhov) measures on curves and loops, as well as to K\"ahler geometry and geometric quantization of moduli spaces of Riemann surfaces.
Inherently, the conformal anomaly and real determinant line bundles are expected to be universal, following a classification of modular functors.
}

\bigskip{}

\noindent\textbf{Keywords:} 
bordered surfaces, conformal field theory, Loewner energy, semigroup of annuli, Virasoro algebra \\ 

\noindent\textbf{MSC:}
Primary: 17B68, 81T40;  
Secondary: 30F45, 60D05, 81T50
\end{minipage}
\end{center}

\newpage

\tableofcontents

\newpage

\section{Introduction}

This work is inspired by the quest of gaining more concrete understanding of 
the emergent breaking of conformal symmetry 
in (euclidean) 2D conformal field theories (CFT),
as manifested by the conformal anomaly and central charge
(cf.~\cite{Gawedzki:CFT_lectures, Schottenloher:Mathematical_introduction_to_CFT}).
The same conformal anomaly is immanent to the conformal restriction property, which characterizes natural measures in models of planar random geometry
(cf.~\cite{LSW:Conformal_restriction_the_chordal_case,
Kontsevich-Suhov:On_Malliavin_measures_SLE_and_CFT, Werner:The_conformally_invariant_measure_on_self-avoiding_loops,
Chavez-Pickrell:Werners_measure_on_self-avoiding_loops_and_welding}).

The purpose of this article is to provide a concrete geometric description of the conformal anomaly in terms of a real determinant line bundle, 
and endow it with an algebraic structure\footnote{While geometrically, the real determinant line bundle is a trivial line bundle, the sewing operation is algebraically nontrivial. This is the key fact that gives rise to the the Virasoro algebra.}.
Our main theorem states that an infinitesimal description of the real determinant line bundle leads to the ubiquitous Virasoro algebra (Theorem~\ref{thm:main}).
Our methods involve elementary computations and concrete constructions, 
geometric \`{a}~la \cite{Friedan-Shenker:The_analytic_geometry_of_two-dimensional_conformal_field_theory, 
Kontsevich:Virasoro_and_Teichmuller_spaces, 
Segal:Definition_of_CFT}.
With future applications in mind, we
also highlight the universality of the real determinant line bundle
and its fundamental role in the geometry of critical interfaces~\cite{Malliavin:Canonic_diffusion_above_the_diffeomorphism_group_of_the_circle, Kontsevich:CFT_SLE_and_phase_boundaries},
large deviations of random curves~\cite{Wang:Equivalent_descriptions_of_the_Loewner_energy, Peltola-Wang:LDP},
K\"ahler geometry~\cite{Bowick-Rajeev:String_theory_as_the_Kahler_geometry_of_loop_space,
Takhtajan-Teo:Weil-Petersson_metric_on_the_universal_Teichmuller_space, 
Alekseev-Meinrenken:Symplectic_geometry_of_Teichmuller_spaces_for_surfaces_with_ideal_boundary},
and geometric quantization~\cite{Teschner-Vartanov:SUSY_gauge_theories_quantization_of_M_flat_and_CFT,
AST:Berezin_quantization_conformal_welding_and_the_Bott-Virasoro_group}.

Our results concern axioms of CFT as such, and thus apply to any 2D CFT.
Namely, Weyl covariance --- see Equation~\eqref{eq:weyl_partition} or~\cite{Gawedzki:CFT_lectures} --- axiomatizes the conformal anomaly and therefore establishes that CFTs can be defined on Riemann surfaces.
This perspective is the central aspect of 
Segal's axioms~\cite{Segal:Definition_of_CFT},
where sewing (or gluing) of Riemann surfaces with parametrized boundaries enables 
description of CFT content from just a few building blocks: disks, cylinders (annuli), and pairs of pants.
In this approach, the required disks and pants are characterized by a finite number of parameters.
However, due to the incorporation of boundary parametrizations, the moduli space of cylinders is an infinite-dimensional space.
Hence, pertaining to revealing the underlying infinite-dimensional Virasoro algebra by a direct computation, 
the primary interest to the present work
is the semigroup formed by cylinders and the sewing operation ---
the semigroup of annuli~\cite{Segal:Definition_of_CFT, Neretin:Holomorphic_extensions_of_representations_of_the_group_of_diffeomorphisms_of_the_circle, Neretin:book, Radnell-Schippers:The_semigroup_of_rigged_annuli_and_the_Teichmueller_space_of_the_annulus, 
BGKR:Semigroup_of_annuli_in_Liouville_CFT}.
(See also~\cite{GKRV:Segals_axioms_for_Liouville_theory, BGKRV:Virasoro_structure_and_the_scattering_matrix_for_Liouville_CFT} and references therein for a recent probabilistic approach in the special case of Liouville CFT.)

In random geometry, most works deal with simply connected planar domains.
In order to extend important probabilistic objects, 
such as Schramm--Loewner evolution (SLE) random 
curves~\cite{Lawler:Partition_functions_loop_measure_and_versions_of_SLE,
Dubedat:SLE_and_Virasoro_representations_localization, Zhan:SLE_loop_measures}, 
or probabilistic formulations of Liouville quantum gravity (LQG) surfaces~\cite{Polyakov:Quantum_geometry_of_bosonic_strings, 
Duplantier-Sheffield:LQG_and_KPZ, GRV:Polyakovs_formulation_of_2D_bosonic_string_theory}, 
to multiply connected domains and Riemann surfaces, 
one has to deal with the effect of the conformal moduli (see also the recent~\cite{ARS:The_moduli_of_annuli_in_random_conformal_geometry} for an interesting perspective). 
In this context, the conformal anomaly --- the starting point of the present article --- appears in the formulation of natural measures on curves and loops on Riemann surfaces~\cite{Malliavin:Canonic_diffusion_above_the_diffeomorphism_group_of_the_circle, 
Kontsevich:CFT_SLE_and_phase_boundaries, 
Kontsevich-Suhov:On_Malliavin_measures_SLE_and_CFT}.
The construction of these Malliavin--Kontsevich--Suhov (MKS) loop measures
has gained significant interest lately~\cite{Friedrich:On_connections_of_CFT_and_SLE,
Werner:The_conformally_invariant_measure_on_self-avoiding_loops, Dubedat:SLE_and_Virasoro_representations_localization, Benoist-Dubedat:SLE2_loop_measure, 
Zhan:SLE_loop_measures, AHS:SLE_loop_via_conformal_welding_of_quantum_disks, 
Carfagnini-Wang:OM_functional_for_SLE_loop_measures}. 
However, the key conjecture that the conformal restriction covariance property, structurally expressed by real determinant lines, \emph{uniquely} determines the MKS loop measures~\cite{Kontsevich-Suhov:On_Malliavin_measures_SLE_and_CFT} is still open\footnote{Let us mention that after our work, the uniqueness of the MKS loop measure on the Riemann sphere was proven by Baverez~\&~Jego~\cite{Baverez-Jego:The_CFT_of_SLE_loop_measures_and_the_Kontsevich-Suhov_conjecture} by using an infinitesimal approach to SLE loop measures.}
--- except for the special case without conformal anomaly~\cite{Werner:The_conformally_invariant_measure_on_self-avoiding_loops, 
Chavez-Pickrell:Werners_measure_on_self-avoiding_loops_and_welding}.
Establishing the uniqueness
would also have important applications to constructions involving
welding of LQG surfaces and their decoration by SLE.

In this article, we provide a detailed construction of the real determinant line bundles and by explicit computation show the emergence of the nontrivial central extension of the classical conformal symmetry: the Virasoro algebra (see Theorem~\ref{thm:main}). 
It is not only a fundamental aspect of the algebraic content of CFT, but also 
has geometric significance 
in the spirit of (geometric) quantization, where the Virasoro algebra is simultaneously expected to provide a symplectic form on the moduli spaces of (bordered) Riemann surfaces and the curvature of an anticipated connection on the real (and complex) determinant line bundles.
This is realized in the diffeomorphism group of the unit circle~\cite{Bowick-Rajeev:The_holomorphic_geometry_of_closed_bosonic_string_theory_and_DiffS1modS1, 
Bowick-Rajeev:String_theory_as_the_Kahler_geometry_of_loop_space,
Gordina-Lescot:Riemannian_geometry_of_DiffS1, Gordina:Riemannian_geometry_of_DiffS1_revisited},
or alternatively, the moduli space of disks covered by the universal Teichm\"uller space $\univteichm$~\cite{Nag-Verjovsky:DiffS1_and_Teichmuller_spaces, Nag-Sullivan:Teichmuller_theory_and_the_universal_period_mapping_via_quantum_calculus_and_H_half_space_on_the_circle, Takhtajan-Teo:Weil-Petersson_metric_on_the_universal_Teichmuller_space}.
In particular, $\univteichm$ admits is a K\"ahler potential, which is given by the ``universal Liouville action''~\cite{Schiffer-Hawley:Connections_and_conformal_mapping, Takhtajan-Teo:Weil-Petersson_metric_on_the_universal_Teichmuller_space}.
This K\"ahler potential also coincides with loop Loewner energy \cite{Wang:Equivalent_descriptions_of_the_Loewner_energy}, 
which is the Onsager-Machlup functional for SLE loops \cite{Carfagnini-Wang:OM_functional_for_SLE_loop_measures} and
their anticipated rate function for large deviations in the semiclassical limit, where the central charge approaches 
negative infinity~\cite{Peltola-Wang:LDP}. 
(See also Theorem~\ref{thm:loewner}, where we give an explicit relation to real determinant lines.)

\subsection{Conformal anomaly and central extensions}

Let $\Sigma$ be a compact connected Riemann surface and $\conf(\Sigma)$ the conformal class of metrics on $\Sigma$.  
\emph{Weyl transformations}
refer to
the action of functions
${\sigma \in C^\infty(\Sigma, \R)}$ on
$\conf(\Sigma)$ by locally rescaling a metric $g \in \conf(\Sigma)$ to $e^{2\sigma}g$.
Denote by $\nabla_g$, $R_g$, and $\vol_g$ respectively the divergence, 
Gaussian curvature,
and volume form on $\Sigma$ in the metric $g$, 
and by $k_g$ and $\tilde \vol_g$ the boundary curvature and the volume form on $\partial \Sigma$ induced by $g$.
The \emph{conformal anomaly} of such a Weyl transformation is defined by the functional
\begin{align}
\lfunct(\sigma, g) \coloneqq \frac{1}{12 \pi} \iint_\Sigma \bigg(
\frac{1}{2} |\nabla_g \sigma|_g^2 + R_g \sigma
\bigg) \vol_g
+ \frac{1}{12 \pi} \int_{\partial \Sigma} k_g \sigma \, \tilde \vol_g
.
\label{eq:weyl_liouville_action}
\end{align}

In the literature, $\lfunct(\sigma, g)$ is sometimes referred to as the ``Liouville'' or ``linear dilaton'' action, making a direct connection to quantum gravity and string theory\footnote{
The action functional~\eqref{eq:weyl_liouville_action} coincides with that of a linear dilaton CFT.
In Liouville CFT, the action functional~\eqref{eq:weyl_liouville_action}
is modified by an additional factor $Q$ in the curvature term
plus an interaction term of the form $\mu e^{\gamma \sigma}$, where $\gamma \in (0,2)$ is a parameter determining the central charge, 
and the coupling constant $\mu > 0$ is called the cosmological constant.
Since by taking $\mu = 0$ and $Q = 1$, we obtain the action~\eqref{eq:weyl_liouville_action}, we will denote it by $\lfunct$.
}. 
However, the conformal anomaly is common to all conformal field theories: 
it is postulated~\cite{Gawedzki:CFT_lectures} that the partition function $Z_g$ of any 2D CFT on $\Sigma$ is a function of the metric $g$, which only essentially depends on the conformal class.
(This is a key feature that distinguishes CFTs from other two-dimensional quantum field theories.) 
Specifically, $Z_g$ is diffeomorphism invariant, and \emph{Weyl covariant}: 
\begin{align}
\label{eq:weyl_partition}
Z_{e^{2\sigma} g} = e^{\charge \, \lfunct(\sigma, g)} Z_g ,
\end{align}
where $\charge \in \R$ is the \emph{central charge} of the CFT.

In their formulation of MKS loop
measures
on $\Sigma$, 
Kontsevich~\&~Suhov~\cite{Kontsevich-Suhov:On_Malliavin_measures_SLE_and_CFT} proposed a (local) conformal restriction covariance property for the loop measure, 
involving a reformulation of the conformal anomaly as a pairing of metrics, 
\begin{align} \label{eq:pairing}
\begin{split}
\liou{\blank}{\blank} & \colon \conf(\Sigma) \times \conf(\Sigma) \longrightarrow \R, \\
\liou{g_1}{g_2} & \coloneqq \frac{1}{48\pi \ii} \iint_\Sigma (f_1 - f_2) \partial \bar \partial (f_1 + f_2),
\end{split}
\end{align}
where locally $g_i = e^{f_i} \dz$ for $f_i \in \C^\infty(\Sigma, \R)$ and $i=1,2$.
For surfaces without boundary and metrics $g_1 = g$ and $g_2 = e^{2\sigma}g$,~\eqref{eq:weyl_liouville_action} and~\eqref{eq:pairing}
are equivalent: $\liou{g_1}{g_2} = -\lfunct(\sigma, g)$.
For surfaces $\Sigma$ with boundary, 
Equations~\eqref{eq:weyl_liouville_action}
and~\eqref{eq:pairing} still agree up to a sign and a boundary term,
which vanishes if we restrict the conformal class $\conf(\Sigma)$ to \emph{admissible} metrics $\cconf(\Sigma)$
which near the boundary are the pushforwards of the flat metric $\dz$ on 
the cylinder $S^1 \times \R \cong \C / 2\pi \Z$ along boundary parametrizations 
(see Sections~\ref{section:background_cylinders_2}--\ref{section:properties}).

The antisymmetry and cocycle properties of the pairing~\eqref{eq:pairing} enable 
the definition of a real determinant line $\Detrc(\Sigma)$ of the surface $\Sigma$ (see Section~\ref{section:setup_detlines}),
forming a real line bundle also mentioned in the seminal 
work~\cite{Friedan-Shenker:The_analytic_geometry_of_two-dimensional_conformal_field_theory} of Friedan~\&~Shenker. 
The line $\Detrc(\Sigma)$ is a one-dimensional real vector space of formal multiples of metrics $\lambda [g]$ for $\lambda \in \R$ and $g \in \cconf(\Sigma)$, subject to the relation\footnote{In principle, one can define the relation~\eqref{eq:detrelation} over the full conformal class, $\conf(\Sigma)$, using the conformal anomaly $\lfunct(\sigma, g)$ (see Appendix~\ref{appendix:cocycle_property_boundary}). 
However, it turns out that the pairing~\eqref{eq:pairing} on admissible metrics is particularly suitable for computations involving the sewing operation~\eqref{eq:intro_sewing}.}
\begin{align}
    \label{eq:detrelation}
[g_1] = e^{\charge \liou{g_2}{g_1}} [g_2] , \qquad g_1, g_2 \in \cconf(\Sigma).
\end{align}
One may think of these determinant lines as the notion which turns a CFT partition function~\eqref{eq:weyl_partition} into a \emph{Weyl invariant} object, $Z \coloneqq Z_g [g] \in \Detrc(\Sigma)$. 

In Segal's axiomatic approach to CFT~\cite{Segal:Definition_of_CFT}, 
quantities like the partition function are first computed locally and then composed by means of a sewing operation. 
To make the sewing unambiguous and the result a Riemann surface, one should pick reasonably regular (e.g.,  analytical~\cite{Segal:Definition_of_CFT, Huang:2D_Conformal_geometry_and_VOAs}, or quasisymmetric or Weil--Petersson~\cite{RSS:Quasiconformal_Teichmuller_theory_as_analytical_foundation_for_CFT}) parametrizations of the boundary components. 
As the present work is concerned with the nature of the conformal anomaly rather than questions about the various choices of regularity as such, 
we will assume that the boundary components of $\Sigma$ have analytical parametrizations (Section~\ref{section:preliminaries}).
The surfaces for which we define the real determinant lines $\Detrc(\Sigma)$ form moduli spaces
\begin{align}
\label{eq:all_surfaces}
\moduli{\genus}{\boundaries} = 
\left\{ \textnormal{\parbox{0.58 \textwidth}{\centering connected compact genus $\genus$ Riemann surfaces $\Sigma$
with $\boundaries$ enumerated and analytically parametrized boundary components $\partial_1 \Sigma, \ldots, \partial_\boundaries \Sigma$}} \right\}_{/ \; \mathrm{isom.}} 
\end{align}
for all genera $\genus \in \Znn$ and number $\boundaries \in \Znn$ of boundary components.
The key constructions in the present article only involve 
cylinders (or annuli), 
comprising
the special case $\moduli{0}{2}$
(see Sections~\ref{section:modular_functors} and~\ref{section:genus} for remarks about the general case).
In fact, we shall only be concerned with the algebraic structure of the moduli space, given by the \emph{sewing operation},
defined by identifying the $2$nd boundary of a cylinder $A$ with the $1$st boundary of a cylinder $B$ using the respective boundary parametrizations,
\begin{align}
A \sew{2}{1} B
\coloneqq (A \sqcup B) /_{\partial_2 A = \partial_1 B} .
\label{eq:sewing_operation}
\end{align}
This turns $\moduli{0}{2}$ into a semigroup.
For concreteness, we will write $A, B \in \cylinders$ to denote any explicit representatives of $[A], [B] \in \moduli{0}{2}$, i.e., $\cylinders$ is the proper class of cylinders (see Remark~\ref{remark:proper_class}).

Importantly, admissible metrics line up smoothly across the seam in $A \sew{2}{1} B$.
Hence, we can introduce an extension of the sewing operation~\eqref{eq:sewing_operation} to the real determinant lines by the bilinear maps
\begin{equation}
    \label{eq:intro_sewing}
\begin{aligned}
\Detrc(A) \otimes \Detrc(B) &\xlongrightarrow{\sim} \Detrc(A \sew{2}{1} B) , \\
\lambda_1 [g] \otimes \lambda_2 [h] &\longmapsto \lambda_1 \lambda_2 [g \cup h].
\end{aligned}
\end{equation}
These isomorphisms extend the semigroup structure on $\moduli{0}{2}$ to the determinant lines.

The main object of the present work is the generalization of the real determinant lines to the set~\eqref{eq:complex_deformations} of \emph{complex deformations} $\DiffC$ of the unit circle $S^1$ inside the cylinder $S^1 \times (-1, 1)$.
For some $\varepsilon > 0$, an element $\phi \in \DiffC$ extends to a complex-analytic map
\begin{align*}
\phi \colon S^1 \times (-\varepsilon, \varepsilon) \longrightarrow S^1 \times (-1, 1),
\end{align*}
where we endow the cylinders with the complex structure given by the coordinate $z = \theta + \ii x$ for $\theta \in S^1$ and $x$ in the interval (see Section~\ref{section:background_Diff}). In particular, $\DiffC$ includes
\begin{align*}
\Diffpan = \left\{ \textnormal{\parbox{0.45 \textwidth}{
\centering
real-analytic, orientation-preserving diffeomorphisms of the unit circle $S^1$
}} \right\} .
\end{align*}
The generalized real determinant lines form a central extension over $\DiffC$ and, in particular, a central extension of $\Diffpan$ in the following manner.

The complex deformations act on cylinders by deformation of one of the boundaries: 
\begin{align*}
A \diffActing{i} \phi , \qquad \textnormal{for} \ A \in \cylinders, \ \phi \in \DiffC, \ i=1,2.
\end{align*}
When $\phi \in \Diffpan$, this action is just a reparametrization of the $i$th boundary component.
Detailed constructions are given in Section~\ref{section:background_cylinders_2}.
Given a complex deformation $\phi \in \DiffC$ and a cylinder $A \in \cylinders$, we define 
\begin{align*}
\Detrc(\phi, A) \coloneqq \Detrc(A \diffActing{1} \phi) \otimes (\Detrc(A))^\vee
\end{align*}
to be the real determinant line of $A \diffActingInline{1} \phi$ tensored with the dual space of the real determinant line of $A$, 
that is, $(\Detrc(A))^\vee \coloneqq \Hom_\R(\Detrc(A), \R)$. 
Note that, as the real determinant lines are one-dimensional, 
the resulting space $\Detrc(\phi, A)$ is also one-dimensional --- 
it is the real determinant line of the complex deformation ${\phi \in \DiffC}$ with respect to the cylinder $A \in \cylinders$~\cite{Segal:Definition_of_CFT, 
Huang:2D_Conformal_geometry_and_VOAs}. 
As a set, the central extension of $\DiffC$ is defined as
\begin{align}
\label{eq:def_extension}
\Detrpc(\DiffC) \coloneqq \big\{ (\phi, \lambda \globalsectionphi(\phi)) \: \big| \: \phi \in \DiffC, \ \lambda > 0\big\} ,
\end{align}
where $\globalsectionphi(\phi) \neq 0$ is a canonical element of $\Detrc(\phi, \A)$ defined in Section~\ref{section:central_extension}, and $\A$ is a standard cylinder.
Importantly, the real determinant lines are independent of the choice of the cylinder in the sense that for $A, B, C \in \cylinders$, there are canonical isomorphisms
\begin{align*}
\natisophi{\phi}{A}{B} \colon \Detrc(\phi, A) \longrightarrow \Detrc(\phi, B),
\end{align*}
such that
$\natisophi{\phi}{A}{A} = \id_{\Detrc(\phi, A)}$
and $\natisophi{\phi}{B}{C} \circ \natisophi{\phi}{A}{B}
= \natisophi{\phi}{A}{C}$.
The multiplication in the central extension $\Detrpc(\DiffC)$ 
is given by 
composition in $\DiffC$ for the first component, 
and bilinear maps for composable $\phi_1, \phi_2 \in \DiffC$, $A \in \cylinders$ for the second component,
\begin{align*}
\detmult{\phi_1}{\phi_2} \colon \Detrc(\phi_1, A) \otimes \Detrc(\phi_2, A) \longrightarrow \Detrc(\phi_1\phi_2, A),
\end{align*}
which satisfy an associativity axiom and are independent of the choice of the annulus~$A$. (See Section~\ref{section:central_extension} for details.)
The sewing isomorphisms $\detmult{\phi_1}{\phi_2}$ were introduced by Huang in the context of complex determinant lines~\cite[Appendix~D]{Huang:2D_Conformal_geometry_and_VOAs}, whose precise relation to $\Detrc(\Sigma)$ 
remains unclear to us at the moment. (See Section~\ref{section:modular_functors} for some comments.)

\subsection{Main results: Identification of the cocycle}

For concrete computations, it is usually much more convenient to consider the Lie algebra central extensions induced by Lie group central extensions by taking differentials.
In the case of present interest, 
the Lie algebra of $\Diffpan$ is $\VectR$, the Lie algebra of real-analytic vector fields on $S^1$. 
The complexification $\Witt = \VectR \otimes\: \C$ is known as the \emph{Witt algebra}, and it is the Lie algebra of the complex deformations $\DiffC$ in the sense that flows of complex vector fields yield complex deformations (see also Section~\ref{section:background_Diff}).

Central extensions of Lie algebras are characterized by two-cocycles in the Lie algebra cohomology.
In the case of $\Witt$, all two-cocycles with coefficients in $\C$ are cohomologous to the Gel'fand--Fuks cocycles $\gelfandfuks$ for some $\charge \in \C$~\cite{Bott:On_the_characteristic_classes_of_groups_of_diffeomorphisms, Guieu-Roger:Virasoro_book},
\begin{align}
\label{eq:cocycle_virasoro}
\gelfandfuks(v, w)
= \frac{\charge}{24\pi}
\int_0^{2\pi}
v'(\theta)
w''(\theta)
\; \dd \theta,
\qquad v, w \in \Witt.
\end{align}
These induce the central extension known as the \emph{Virasoro algebra} $\Vir$ of central charge $\charge \in \C$. (We include the central charge in the notation to highlight a correspondence at the level of cocycles. See Remark~\ref{rem:Vir_exact_sequence}.)

A cocycle is obtained from the abstract central extension by picking a section (e.g., $\Witt \to \Vir$),
which is linear but not a Lie algebra homomorphism, unless the central extension is trivial. Different sections give cocycles differing by coboundaries.
Group cocycles are also obtained from corresponding sections 
by differentiation,  
and for the central extension $\Detrpc(\DiffC)$ defined in~\eqref{eq:def_extension},  
there is a convenient section $\globalsectionphi(\phi)$ defined in terms of 
uniformized representatives of cylinders ---  
see Equation~\eqref{eq:diffglobalsection} in Section~\ref{section:central_extension}. 
Even though $\DiffC$ is not a Lie group, 
its cocycle may still be differentiated via flows of vector fields --- see Equation~\eqref{eq:algebra_cocycle_general}.
This yields the corresponding cocycle $\algcocycle$ of the Lie algebra central extension of $\Witt$ with coefficients in $\R$, 
as specified in Theorem~\ref{thm:main}.

The main result of this article is the explicit computation of this Lie algebra two-cocycle 
$[\algcocycle] \in H^2(\Witt, \R)$ 
of the real central extension $\Detrpc(\DiffC)$. 
Our result gives a direct proof that the central extension of real vector fields $\VectR$
induced by real determinant lines of central charge $\charge \in \R$ is trivial. Nonetheless, the central extension of complex vector fields $\Witt$ is nontrivial, and gives the imaginary part of the Gel'fand--Fuks cocyle with the same central charge. 

\begin{theorem}
\label{thm:main}
The Lie algebra of the central extension $\Detrpc(\DiffC)$, $\charge \in \R$, with respect to the section $\globalsectionphi(\phi)$ defined in~\eqref{eq:diffglobalsection}, is given by the cocycle
\begin{align*} 
\algcocycle(v ,w)
= \frac{\charge}{24\pi}
\Im
\int_0^{2\pi}
v'(\theta)
w''(\theta)
\; \dd \theta,
\qquad v, w \in \Witt.
\end{align*}
It vanishes for $v, w \in \VectR$ and equals the imaginary part of the Gel'fand-Fuks cocycle $\gelfandfuks$ with the same central charge.
\end{theorem}
The proof of this theorem is the content of Section~\ref{section:computation}.
Our result answers a question discussed by Andr\'e Henriques and Dylan Thurston in 2011 in mathoverflow\footnote{See~\url{https://mathoverflow.net/questions/61601}.}. 

\begin{remark}
The agreement of the central charge of $\gelfandfuks$ and $\algcocycle$ depends on the choices of several conventions throughout the article, notably in Equations~(\ref{eq:weyl_liouville_action},~\ref{eq:pairing},~\ref{eq:cocycle_virasoro},~\ref{eq:diff_cocycle},~\ref{eq:algebra_cocycle_general}).
However, these conventions change the results by factors of $2$ and signs only.
In~\textnormal{\cite{Huang:2D_Conformal_geometry_and_VOAs}} a disagreement of the central charge of the complex determinant line bundle and the cocycle~\eqref{eq:cocycle_virasoro} by a factor of $2$ was observed \textnormal{(}with different convention for~\eqref{eq:diff_cocycle}\textnormal{)}.
\end{remark}

\begin{remark}
\label{rmk:diffeomorphisms_triviality}
The following idea, pointed out to us by Andr\'e Henriques, can be used to directly argue that the cocycle $\algcocycle$ vanishes on real vector fields and moreover that the group level cocycle $\calpha(\phi, \psi)$ (Definition~\ref{def:cocycle_deformation}) is a coboundary for $\phi, \psi \in \Diffpan$. 
Diffeomorphisms in $\Diffpan$ may be approximated by ``thin'' annuli, where the boundary components of the annuli overlap~\cite{Henriques:Conformal_field_theory_lectures}.
Thin annuli do not admit admissible metrics, which is why 
it is useful to consider the conformal anomaly with boundary term~\eqref{eq:weyl_liouville_action}.
A~diffeomorphism corresponds to a ``completely thin'' annulus where both boundary components are given by $S^1$, one of which is parametrized by the identity $\id_{S^1}$ and the other by the diffeomorphism.
Now, as such annuli have empty interior, it follows that 
the surface integrals in~\eqref{eq:weyl_liouville_action} vanish. 
Since the boundary components overlap, yet with opposite normal vectors, 
the boundary integrals cancel out --- 
rendering the conformal anomaly trivial. 
Thus, the determinant line of a diffeomorphism is canonically isomorphic to $\R$.
We turn this idea into a detailed proof within the context of 
the present article in Appendix~\ref{appendix:diffeomorphisms_triviality}.
However, this argument does not generalize to complex deformations.
\end{remark}

\begin{remark}
\label{remark:obstruction}
Theorem~\ref{thm:main} shows that the central extension of $\Diffpan$ induced by the real determinant lines is trivial at the level of Lie algebras.
To lift this result to $\Diffpan$, the only possible obstruction is the real two-cocycle on $\Diffpan$,
\begin{align*}
(\phi_1, \phi_2) &\longmapsto \Arg(\rotnumber(\phi_1 \circ \phi_2) - \rotnumber(\phi_1) - \rotnumber(\phi_2)),
\end{align*}
where $\rotnumber \colon \Diffpan \to S^1$ is the (Poincar\'e) rotation number.
This is a nontrivial cocycle in $H^2(\Diffpan, \R)$, yet the associated Lie algebra cocycle is trivial (see~\cite{Guieu-Roger:Virasoro_book}).
We show in Proposition~\ref{prop:diffeomorphisms_triviality} in Appendix~\ref{appendix:diffeomorphisms_triviality} 
that $\Detrpc(\Diffpan)$ does not contain this cocycle.
\end{remark}

\subsection{Remarks on modular functors and universality}
\label{section:modular_functors}

As will be explained in Section~\ref{section:genus}, 
the real determinant lines may be defined for all 
Riemann surfaces $\Sigma$, where $[\Sigma] \in \moduli{\genus}{\boundaries}$ for any genus $\genus \in \Znn$ and nonzero number $\boundaries \in \Zpos$ of analytically parametrized boundary components. 
By employing the Polyakov--Alvarez anomaly formula for the zeta-regularized determinant of the Laplacian on $\Sigma$ (see Equation~\eqref{eq:polyakov_alvarez}), 
we obtain a global trivialization\footnote{Note that the determinant of the Laplacian is not very explicit, nor it is amenable to computations. 
Thus, we use another trivialization in the computation of the cocycle in Theorem~\ref{thm:main}.} 
$\globalsectionzeta$ for such a line bundle in Proposition~\ref{proposition:detz}.
This turns the collection $\Detrc$ of determinant lines $\Detrc(\Sigma)$ into a line bundle over the moduli space $\moduli{\genus}{\boundaries}$.
The combination of these line bundles with the sewing isomorphisms~\eqref{eq:sewing_iso} yields a real one-dimensional modular functor, 
as discussed in~\cite{Friedrich:On_connections_of_CFT_and_SLE, Maibach:Master_thesis}.

The notion of a \emph{complex} modular functor
was introduced by Segal~\cite{Segal:Definition_of_CFT} 
and more precisely by Huang~\cite{Huang:2D_Conformal_geometry_and_VOAs,
Huang:Genus-zero_modular_functors_and_intertwining_operator_algebras}.
For the case of one-dimensional complex modular functors, 
the Mumford--Segal theorem 
(see~\cite[Appendix~D]{Huang:2D_Conformal_geometry_and_VOAs}), 
shows that in genus zero, 
all of them are isomorphic to complex determinant line bundles, and the only free parameter is the central charge $\charge \in \C$.
For $\charge = 2$, a fiber of the complex determinant line bundle above a given surface $\Sigma$, such that $[\Sigma] \in \moduli{\genus}{\boundaries}$ with $\boundaries \in \Zpos$, is the one-dimensional complex vector space
\begin{align*}
\Det_\C^2(\Sigma) := \Big(\bigwedge_{\mathclap{\dim \kernel(\pi_\Sigma)}} \kernel(\pi_\Sigma)\Big)^\vee \; \otimes \; \bigwedge_{\mathclap{\dim \coker(\pi_\Sigma)}} \coker(\pi_\Sigma),
\end{align*}
which is the determinant line of an operator $\pi_\Sigma$ acting on holomorphic functions on $\Sigma$, mapping them to the positive Fourier modes of the restriction of the function to the boundary $\partial \Sigma$.
This complex determinant line bundle has been studied in several works  
\cite{Segal:Definition_of_CFT, 
Huang:2D_Conformal_geometry_and_VOAs, 
Radnell:PhD, 
Frenkel-Ben-Zvi:Vertex_Algebras_and_Algebraic_Curves, 
RSS:Quasiconformal_Teichmuller_theory_as_analytical_foundation_for_CFT,
RSSS:Schiffer_operators_and_calculation_of_a_determinant_line_in_conformal_field_theory}, including the construction of an associated central extension of $\Witt$ by $\C$ through the Gel'fand--Fuks cocycle~\eqref{eq:cocycle_virasoro}.
In contrast to our Theorem~\ref{thm:main}, this cocycle does not vanish on real vector fields.
The relation to the real determinant line bundle remains unclear to us at the moment, but we hope to report on this in subsequent work.

Also, we expect that an analogue to the Mumford--Segal theorem holds in the real case --- but, to our knowledge, no rigorous proof is available.
In particular, establishing a real version of the Mumford--Segal theorem would show that $\Detrc$ is the universal real one-dimensional modular functor.
Conceptually, this would emphasize \emph{universality} of the conformal anomaly~(\ref{eq:weyl_liouville_action},~\ref{eq:pairing}),
in the sense that it is
the only possible anomaly
arising from conformal symmetry in two dimensions.
In the context of this program, we obtain the following result from Theorem~\ref{thm:main}.

\begin{corollary}
$\Detrc$ is nontrivial as a real one-dimensional modular functor.
\end{corollary}

\subsection*{Acknowledgments}

We would like to thank Masha Gordina, Yi-Zhi Huang, Eric Schippers, and Yilin Wang for interesting discussions,
and Klimcik Ctirad and Peter Kristel for useful comments on the first version of this manuscript. 
We thank Andr\'e Henriques for his insight towards the triviality of the group cocycle on diffeomorphisms (see Remark~\ref{rmk:diffeomorphisms_triviality}).
We are grateful to the anonymous referees for their careful comments, which helped to improve this manuscript. 

S.M.~is supported by the Deutsche Forschungsgemeinschaft (DFG, German Research Foundation) under Germany's Excellence Strategy EXC-2047/1-390685813.

This material is part of a project that has received funding from the  European Research Council (ERC) under the European Union's Horizon 2020 research and innovation programme (101042460): 
ERC Starting grant ``Interplay of structures in conformal and universal random geometry'' (ISCoURaGe) 
and from the Academy of Finland grant number 340461 ``Conformal invariance in planar random geometry.''

E.P.~is also supported by 
the Academy of Finland Centre of Excellence Programme grant number 346315 ``Finnish centre of excellence in Randomness and STructures (FiRST)'' 
and by the Deutsche Forschungsgemeinschaft (DFG, German Research Foundation) under Germany's Excellence Strategy EXC-2047/1-390685813, 
as well as the DFG collaborative research centre ``The mathematics of emerging effects'' CRC-1060/211504053.

\section{Preliminaries}
\label{section:preliminaries}

\subsection{Diffeomorphisms and complex deformations of the circle}
\label{section:background_Diff}

We define the unit circle as $S^1 = \R/(2\pi \Z)$ equipped with the coordinate ${\theta \in [0, 2\pi)}$.
We identify $S^1$ with the subset $S^1 \times \{0\}$ of the infinite cylinder $S^1 \times \R$, with complex coordinate
\begin{align} \label{eq:zcoordinate}
\begin{split}
S^1 \times \R &\longrightarrow \{\theta + \ii x \in \C \colon \theta \in [0, 2\pi), \ x \in \R\} , \\
(\theta, x) &\longmapsto \theta + \ii x.
\end{split}
\end{align}
Let $\Diffpan$ denote the group of real-analytic, orientation-preserving diffeomorphisms of $S^1$. 
It is a Fr\'echet-Lie group modeled on the Fr\'echet-Lie algebra of real-analytic vector fields on $S^1$,
\begin{align*}
\Lie(\Diffpan) = \VectR.
\end{align*}
See~\cite{Neeb:Infinite-Dimensional_Lie_Groups} for details on the analytic structure. In particular, the exponential map is 
\begin{equation}
\label{eq:exponential_diff}
\begin{aligned}
\Exp \colon \VectR &\longrightarrow \Diffpan , \\
v &\longmapsto \phi_1,
\end{aligned}
\end{equation}
where $\phi_t \in \Diffpan$ for $t \in \R$ is the flow of the vector field $v(\theta)$ defined by 
\begin{align}
\phi_0(\theta) = \theta 
\qquad \textnormal{and} \qquad
\pdv{\phi_t(\theta)}{t} = v(\phi_t(\theta)).
\label{eq:floweqs}
\end{align}

By Equation~\eqref{eq:floweqs}, we also define the flow $\phi_t \colon S^1 \to S^1 \times \R$ of complex-valued vector fields $v \in \Witt = \VectR \otimes\:  \C$ for small $t$, where we use that $v$ has a complex-analytic extension to a small cylinder $S^1 \times (-\varepsilon, \varepsilon)$ for some $\varepsilon > 0$.
For example, $v(\theta) = \ii$ generates the translation of the unit circle along the imaginary axis, $\phi_t(\theta) = \theta + t \ii$.
More generally, consider the set of \emph{complex deformations}
\begin{align}
\label{eq:complex_deformations}
\DiffC = \left\{ \textnormal{\parbox{0.6 \textwidth}{
\centering
$\phi \colon S^1 \to S^1 \times (-1, 1)$ extending complex-analytically to $S^1 \times (-\varepsilon, \varepsilon)$ for some $\varepsilon > 0$ such that the image of the extension contains $S^1 \times \{0\}$.
}} \right\},
\end{align}
which contains $\Diffpan \subset \DiffC$, but also includes the flows of complex vector fields.
Note that the set $\DiffC$ is not a group. Nevertheless, it is closed under taking inverses, and, if $\phi, \psi \in \DiffC$ are such that $\phi$ extends to $\psi(S^1)$ and $\phi(\psi(S^1)) \subset S^1 \times (-1, 1)$, they are \emph{composable} in the sense that $\phi \psi \in \DiffC$. 
The benefit of uniformly bounding the imaginary part of complex deformations by $1$ becomes apparent in Section~\ref{section:central_extension}.

We do not need any differentiable structure on $\DiffC$ in this work. However, we will use the generalization of the exponential map~\eqref{eq:exponential_diff} to interpret $\Witt$ as the Lie algebra of $\DiffC$.

\subsection{Cocycles and the Virasoro algebra}
\label{section:background_vir}

Next, recall that a central extension $\mathsf{H}$ of a Lie group $\mathsf{G}$ with unit $1$ 
by an Abelian Lie group $\mathsf{A}$ is a short exact sequence of Lie groups
\begin{align*}
\begin{alignedat}{2}
1 \longrightarrow \mathsf{A} &\longhookrightarrow \mathsf{H} &&\longtwoheadrightarrow \mathsf{G} \longrightarrow 1 
\end{alignedat}
\end{align*}
such that the image of $\mathsf{A}$ commutes with all of $\mathsf{H}$. 
A central extension $\mathfrak{h}$ of a Lie algebra $\mathfrak{g}$ 
by an Abelian Lie algebra $\mathfrak{a}$ is a short exact sequence 
\begin{align}
\label{eq:lie_algebra_exact}
\begin{alignedat}{2}
0 \longrightarrow \mathfrak{a} &\longhookrightarrow \mathfrak{h} &&\longtwoheadrightarrow \mathfrak{g} \longrightarrow 0 .
\end{alignedat}
\end{align}
These central extensions are classified by the second Lie algebra cohomology~\cite{Guieu-Roger:Virasoro_book}, that is, the quotient
of two-cocycles $Z^2(\mathfrak{g}, \mathfrak{a})$ by coboundaries $B^2(\mathfrak{g}, \mathfrak{a})$,
\begin{align*}
H^2(\mathfrak{g}, \mathfrak{a}) = Z^2(\mathfrak{g}, \mathfrak{a}) / B^2(\mathfrak{g}, \mathfrak{a}) .
\end{align*}
The case of interest to the present work is the Lie algebra $\mathfrak{g} = \Witt$ and coefficients $\mathfrak{a} = \R$. 
The cohomology in the case of $\mathfrak{a} = \C$ is well-known.
\begin{proposition}
\label{prop:virasoro_cocycle_classification}
Every complex-valued two-cocycle of the Witt algebra in $Z^2(\Witt, \C)$ is cohomologous to the Gel'fand--Fuks cocycle $\gelfandfuks$ defined in~\eqref{eq:cocycle_virasoro} for some $\charge \in \C$.
\end{proposition}
\begin{proof}
See, for instance,~\cite{Azcarraga-Izquierdo:Lie_groups_Lie_algebras_cohomology_and_some_applications_in_physics, Guieu-Roger:Virasoro_book, Schottenloher:Mathematical_introduction_to_CFT,
Khesin-Wendt:The_geometry_of_infinite-dimensional_groups}.
\end{proof}
Thus, the Virasoro algebra is the only nontrivial central extension of $\Witt$ by $\C$.

\begin{remark} \label{rem:Vir_exact_sequence}
An isomorphism of Lie algebra central extensions is an isomorphism of short exact sequences~\eqref{eq:lie_algebra_exact},
i.e., a Lie algebra isomorphism 
$f \colon \mathfrak{h} \to \mathfrak{h}'$ such that the following diagram commutes:
\begin{displaymath}
\begin{tikzcd}
0 & \mathfrak{a} & \mathfrak{h} & \mathfrak{g} & 0 \\
0 & \mathfrak{a} & \mathfrak{h}' & \mathfrak{g} & 0
\arrow[from=1-1, to=1-2]
\arrow[from=1-2, to=1-3]
\arrow[from=1-3, to=1-4]
\arrow[from=1-4, to=1-5]
\arrow[from=2-1, to=2-2]
\arrow[from=2-2, to=2-3]
\arrow[from=2-3, to=2-4]
\arrow[from=2-4, to=2-5]
\arrow[from=1-1, to=2-1]
\arrow["{\id}"', from=1-2, to=2-2]
\arrow["{f}"', from=1-3, to=2-3]
\arrow["{\id}"', from=1-4, to=2-4]
\arrow[from=1-5, to=2-5]
\end{tikzcd}
\end{displaymath}
For different values of $\charge \in \C \setminus \{0\}$, the Lie algebras $\Vir$ are isomorphic. However, this is not the case for the associated exact sequences of Lie algebras,
\begin{align*}
\begin{alignedat}{2}
0 \longrightarrow \C &\longhookrightarrow \Vir &&\longtwoheadrightarrow \Witt \longrightarrow 0  \\
\lambda &\longmapsto (0, \lambda) &&\longmapsto 0 \\
&\phantom{\longmapsto .} (v, \lambda) &&\longmapsto v,
\end{alignedat}
\end{align*}
where $\Vir = \Witt \oplus \, \C$ with the Lie bracket $[(v, \lambda), (w, \mu)] = ([v, w], \lambda + \mu + \gelfandfuks(v, w))$.
Therefore, we will explicitly keep the value of $\charge$ in the notation $\Vir$.
\end{remark}

Let $G$ be a Fr\'echet-Lie group with Lie algebra $\mathfrak{g} = \Lie(G)$.
In general, one obtains a real-valued Lie algebra two-cocycle $[\omega] \in H^2(\mathfrak{g}, \R)$ from a real-valued smooth group two-cocycle $[\Omega] \in H^2(G, \R)$ by differentiating the group cocycle $\Omega$
along the exponentials:
\begin{align}
\label{eq:diff_cocycle}
\omega(v, w) = 
\frac{1}{2}
\pdv{}{t}{s} 
\Big( \Omega(\Exp(tv), \Exp(sw)) - \Omega(\Exp(sw), \Exp(tv)) \Big)\Big|_{t = s = 0}  
,
\end{align}
for $v, w \in \mathfrak{g}$. 
In the case of the two-cocycle $\Omega \in H^2(\Diffpan, \R)$, we have $\Exp(tv) = \phi_t$ and $\Exp(sw) = \psi_s$, 
where $\phi_t$ and $\psi_s$ are respectively the flows~\eqref{eq:floweqs} of $v, w \in \VectR$.
(See~\cite[Proposition~3.14]{Khesin-Wendt:The_geometry_of_infinite-dimensional_groups}.)
The Lie algebra cocycle can thus be computed as
\begin{align}
\omega(v, w) = 
\frac{1}{2} \pdv{}{t}{s} \Big(\Omega(\phi_t, \psi_s) - \Omega(\psi_s, \phi_t) \Big)\Big|_{t = s = 0} .
\label{eq:algebra_cocycle_general}
\end{align}
The proof of Theorem~\ref{thm:main} in Section~\ref{section:computation} relies on the computation of the Lie algebra cocycle $\algcocycle$ of the central extension $\Detrpc(\DiffC)$ from the formula~\eqref{eq:algebra_cocycle_general} generalized to $\Witt$.

\subsection{Cylinders with analytical boundary parametrizations}
\label{section:background_cylinders_1}

In this section, we introduce cylinders with analytical boundary parametrizations, which comprise the particular case of the surfaces in the moduli spaces $\moduli{\genus}{\boundaries}$, defined in Equation~\eqref{eq:all_surfaces}, with genus $\genus = 0$ and $\boundaries = 2$.
A \emph{cylinder} is a compact genus-zero Riemann surface $A$ with two boundary components labeled $\partial_1 A$ and $\partial_2 A$. 
We endow $A$ with
a choice of analytical boundary parametrization.
\begin{definition}
\label{def:boundary_param}
An \emph{analytical boundary parametrization} of the cylinder $A$ is a pair 
$(\zeta_1, \zeta_2)$ of smooth maps
\begin{align*}
\zeta_i \colon S^1 \longrightarrow \partial_i A,
\quad \textnormal{for $i = 1,2$,}
\end{align*}
which respectively extend complex-analytically to 
the cylinders $S^1 \times [0, \varepsilon)$ and $S^1 \times (-\varepsilon, 0]$, for some $\varepsilon > 0$, 
sending them to neighborhoods of the two boundary components of $A$.
\end{definition}

\begin{remark}
\label{remark:proper_class}
Hence, cylinders with analytical boundary parametrization are triples
\begin{align*}
\paramcyl{A}{\zeta_1}{\zeta_2}
= A
\; \in \;  \cylinders,
\end{align*}
where we omit the parametrizations if they are clear from context.
By $A \in \cylinders$ we denote any representative of a cylinder $[A] \in \moduli{0}{2}$.
Note that while $\moduli{0}{2}$ is an infinite-dimensional Banach manifold, $\cylinders$ is a proper class and should thus not be considered to have any geometric structure.
The reason to consider $\cylinders$ is that we want to give definitions involving particular representatives and only later show that they are natural over the moduli space.
\end{remark}

The \emph{standard cylinder} is defined as 
\begin{align}
\label{eq:standard_cylinder}
\A \coloneqq  \paramcyl{S^1 \times [0, 1]}{\theta}{\theta + \ii} \; \in \; \cylinders .
\end{align}
We later endow $\A$ with the flat metric $\dz$, where $z = \theta + \ii x$.

An \emph{isomorphism} of two cylinders $\paramcyl{A}{\zeta_1}{\zeta_2}, \paramcyl{B}{\xi_1}{\xi_2} \in \cylinders$ with analytical boundary parametrizations is a biholomorphism $\isom{A}{B} \colon A \to B$ 
which is compatible with the boundary parametrizations and preserves their order, 
i.e., such that the following diagrams commute:
\begin{displaymath}
\begin{tikzcd}
A && {B} \\
& {S^1}
\arrow["{\xi_i}"', from=2-2, to=1-3]
\arrow["{\zeta_i}", from=2-2, to=1-1]
\arrow["\isom{A}{B}", from=1-1, to=1-3]
\end{tikzcd}
\qquad \textnormal{for $i = 1,2$.}
\label{diag:isomorphisms}
\end{displaymath}
Applying the identity theorem in an open neighborhood of $\partial_1 A$ where $\isom{A}{B} = \xi_1 \circ \zeta_1^{-1}$ is well-defined, we see that these isomorphisms are unique if they exist (hence the notation $\isom{A}{B}$),
and they have the following composition property 
(if all three isomorphisms exist): 
\begin{align*}
\isom{B}{C} \circ \isom{A}{B} = \isom{A}{C},
\qquad A, B, C \in \cylinders .
\end{align*}

\subsection{The semigroup of cylinders and uniformization}
\label{section:background_cylinders_2}

We now provide the definition of the sewing operation~\eqref{eq:sewing_operation} in the case of two cylinders.
Note that under this operation, $\cylinders$ forms a semigroup, since sewing two cylinders again results in a cylinder.
For $\paramcyl{A}{\zeta_1}{\zeta_2}, \paramcyl{B}{\xi_1}{\xi_2} \in \cylinders$, we define the Riemann surface
\begin{align}
\label{eq:sewing_def}
A \sew{2}{1} B \coloneqq (A \sqcup B ) \, /\sim ,
\end{align}
where the equivalence relation $\sim$ generated by
$\zeta_2(\theta) \sim \xi_1(\theta)$ for all $\theta \in S^1$
identifies the boundaries $\partial_2 A$ and $\partial_1 B$ using their respective parametrizations.
The validity of this definition is spelled out in the next basic lemma.

\begin{lemma}
\leavevmode
\makeatletter
\@nobreaktrue
\makeatother
\label{lem:conformal_atlas}
\begin{enumerate}
\item
\label{item:conformal_atlas1}
For $\paramcyl{A}{\zeta_1}{\zeta_2}, \paramcyl{B}{\xi_1}{\xi_2} \in \cylinders$,
$A \sew{2}{1} B$ is a Riemann surface and $\paramcyl{A \sew{2}{1} B}{\zeta_1}{\xi_2}$ is a cylinder with parametrized boundaries in $\cylinders$.
\item
\label{item:conformal_atlas2}
For $A, B, C, D \in \cylinders$ such that $A$ and $B$ are isomorphic to $C$ and $D$ respectively, we have
the isomorphism
\begin{align*}
\isom{A \sew{2}{1} B}{C \sew{2}{1} D} = \isom{A}{C} \cup \isom{B}{D}.
\end{align*}
\end{enumerate}
\end{lemma}

\begin{figure}
\centering
\begin{tikzpicture}
    \definecolor{color2}{HTML}{F1A340}
    \definecolor{color1}{HTML}{998EC3}

    \node[inner sep=1.2cm, outer sep = 0.1cm]
		(center) at (0, 0) {};
    \begin{scope}[shift=(center)]
        \draw[color=color1, fill]
        (-1, 0) -- (1, 0) -- (1, 0.3) -- (-1, 0.3);
        \draw[color=color2, fill]
        (-1, 0) -- (1, 0) -- (1, -0.3) -- (-1, -0.3);
        \draw [] (-1, 0) -- (1, 0);
        \draw [] (-1, -0.3) -- (-1, 0.3);
        \draw [] (1, -0.3) -- (1, 0.3);
        \draw [densely dashed] (1, 0.3) -- (-1, 0.3);
        \draw [densely dashed] (1, -0.3) -- (-1, -0.3);
        \node[] at (0, -0.8) {$S^1 \times (-\varepsilon, \varepsilon)$};
        \node[circle, fill, inner sep=1pt, outer sep=0.2cm, label=above:{$p$}]
            (pointp) at (-0.4,0) {};
        \draw[densely dashed]
            (-0.4, 0) ellipse (0.25 and 0.25); 
        \node[]
            (cplxplane) at (1, 0.8) {$\C$};
        \draw[->] (pointp) edge[above, bend left] node {$\Psi_p$} (cplxplane);
    \end{scope}

    \node[inner sep=1.2cm, outer sep = 0.1cm]
		(leftup) at (-5, 1.5) {$B$};
    \begin{scope}[shift=(leftup)]
        \draw[] (0,1) ellipse (1 and 0.2);
        \draw[] (-1,-1) arc (180:360:1 and 0.2);
        \draw[postaction={decorate}, decoration={markings,
            mark=at position 0.2 with {\coordinate (leftupA);},
            mark=at position 0.8 with {\coordinate (leftupB);}
        }]
            (-1, 1)
            .. controls(-0.8, 0) ..
            (-1, -1);
        \draw[postaction={decorate}, decoration={markings,
            mark=at position 0.2 with {\coordinate (leftupC);},
            mark=at position 0.8 with {\coordinate (leftupD);}
        }]
            (1, 1)
            .. controls (0.8, 0) .. 
            (1, -1);
        \draw[densely dashed,
            postaction={decorate}, decoration={markings,
            mark=at position 0.3 with {
                \node[circle, fill, inner sep=1pt, outer sep=0.2cm] {};
            }
        }]
            (leftupB)
            .. controls (0, -0.8) ..
            (leftupD);
        \begin{scope}[on background layer]
        \begin{scope}
        \clip
            (-1.1, -1.5)
            --
            (-1.1, 0)
            -- (leftupB)
            .. controls (0, -0.8) ..
            (leftupD)
            -- (1.1, 0)
            -- (1.1, -1.5)
            -- cycle;
        \draw[color=color1, fill]
            (-1,-1)
            arc (180:360:1 and 0.2)
            .. controls (0.8, 0) .. 
            (1, 1)
            arc (180:360:-1 and 0.2)
            .. controls(-0.8, 0) ..
            (-1, -1);
        \end{scope}
        \end{scope}
        \node[inner sep=1cm, outer sep = 0.1cm]
        (leftup2) at (0, -1) {};
    \end{scope}

    \node[inner sep=1.2cm, outer sep = 0.1cm]
    (leftdown) at (-5, -1.5) {$A$};
    \begin{scope}[shift=(leftdown)]
        \draw[] (0,1) ellipse (1 and 0.2);
        \draw[] (-1.2,-1) arc (180:360:1.2 and 0.2);
        \draw[postaction={decorate}, decoration={markings,
            mark=at position 0.15 with {\coordinate (leftdownA);},
            mark=at position 0.8 with {\coordinate (leftdownB);}
        }]
            (-1, 1)
            .. controls(-0.3, -0.8) ..
            (-1.2, -1);
        \draw[postaction={decorate}, decoration={markings,
            mark=at position 0.15 with {\coordinate (leftdownC);},
            mark=at position 0.8 with {\coordinate (leftdownD);}
        }]
            (1, 1)
            .. controls (0.3, -0.8) .. 
            (1.2, -1);
        \draw[densely dashed,
            postaction={decorate}, decoration={markings,
            mark=at position 0.8 with {
                \node[circle, fill, inner sep=1pt, outer sep=0.2cm] {};
            }
        }]
            (leftdownA)
            .. controls (0, 0.45) ..
            (leftdownC);
        \begin{scope}[on background layer]
        \begin{scope}
        \clip
            (-1.1, 1.3)
            --
            (-1.1, 0.8)
            -- (leftdownA)
            .. controls (0, 0.45) ..
            (leftdownC)
            -- (1.1, 0.8)
            -- (1.1, 1.3)
            -- cycle;
        \draw[color=color2, fill]
            (-1,1)
            arc (180:360:1 and -0.2)
            .. controls (0.3, -0.8) .. 
            (1.2, -1)
            arc (180:360:-1.2 and 0.2)
            .. controls(-0.3, -0.8) ..
            (-1, 1);
        \end{scope}
        \node[inner sep=1cm, outer sep = 0.1cm]
        (leftdown1) at (0, 1) {};
    \end{scope}
    \end{scope}

    \node[inner sep=1.2cm, outer sep = 0.1cm]
    (rightup) at (5, 1.5) {$D$};
    \begin{scope}[shift=(rightup)]
        \draw[] (0.1,1) ellipse (0.9 and 0.2);
        \draw[] (-1,-1) arc (180:360:1 and 0.2);
        \draw[postaction={decorate}, decoration={markings,
            mark=at position 0.2 with {\coordinate (rightupA);},
            mark=at position 0.8 with {\coordinate (rightupB);}
        }]
            (-0.8, 1)
            .. controls(-0.8, 0.5) ..
            (-1, 0)
            .. controls(-1.2, -0.5) ..
            (-1, -1);
        \draw[postaction={decorate}, decoration={markings,
            mark=at position 0.2 with {\coordinate (rightupC);},
            mark=at position 0.8 with {\coordinate (rightupD);}
        }]
            (1, 1)
            .. controls (1.2, 0) .. 
            (1, -1);
        \draw[densely dashed,
            postaction={decorate}, decoration={markings,
            mark=at position 0.5 with {
                \node[circle, fill, inner sep=1pt, outer sep=0.2cm] {};
            }
        }]
            (rightupB)
            .. controls (0, -0.8) ..
            (rightupD);
        \begin{scope}[on background layer]
        \begin{scope}
        \clip
            (-1.1, -1.5)
            -- (rightupB)
            .. controls (0, -0.8) ..
            (rightupD)
            -- (1.1, -1.5)
            -- cycle;
        \draw[color=color1, fill]
            (-1,-1)
            arc (180:360:1 and 0.2)
            .. controls (1.2, 0) .. 
            (1, 1)
            arc (180:360:-1 and 0.2)
            .. controls(-0.8, 0.5) ..
            (-1, 0)
            .. controls(-1.2, -0.5) ..
            (-1, -1);
        \end{scope}
        \end{scope}
        \node[inner sep=1cm, outer sep = 0.1cm]
        (rightup2) at (0, -1) {};
    \end{scope}

    \node[inner sep=1.2cm, outer sep = 0.1cm]
    (rightdown) at (5, -1.5) {$C$};
    \begin{scope}[shift=(rightdown)]
        \draw[] (0,1) ellipse (1 and 0.2);
        \draw[] (-1,-1) arc (180:360:1 and 0.2);
        \draw[postaction={decorate}, decoration={markings,
            mark=at position 0.35 with {\coordinate (rightdownA);},
            mark=at position 0.8 with {\coordinate (rightdownB);}
        }]
            (-1, 1)
            .. controls (-1.1, 0.75) .. 
            (-1, 0.5)
            .. controls (-0.9, 0.25) .. 
            (-1, 0)
            .. controls (-1.1, -0.25) .. 
            (-1, -0.5)
            .. controls (-0.9, -0.75) .. 
            (-1, -1);
        \draw[postaction={decorate}, decoration={markings,
            mark=at position 0.25 with {\coordinate (rightdownC);},
            mark=at position 0.8 with {\coordinate (rightdownD);}
        }]
            (1, 1)
            .. controls (0.9, 0.75) .. 
            (1, 0.5)
            .. controls (1.1, 0.25) .. 
            (1, 0)
            .. controls (0.9, -0.25) .. 
            (1, -0.5)
            .. controls (1.1, -0.75) .. 
            (1, -1);
        \draw[densely dashed,
            postaction={decorate}, decoration={markings,
            mark=at position 0.6 with {
                \node[circle, fill, inner sep=1pt, outer sep=0.2cm] {};
            }
        }]
            (rightdownA)
            .. controls (0, 0.3) ..
            (rightdownC);
        \begin{scope}[on background layer]
        \begin{scope}
        \clip
            (-1.1, 1.3)
            --
            (-1.1, 0.4)
            -- (rightdownA)
            .. controls (0, 0.3) ..
            (rightdownC)
            -- (1.1, 0.5)
            -- (1.1, 1.3)
            -- cycle;
        \draw[color=color2, fill]
            (-1,1)
            arc (180:360:1 and -0.2)
            .. controls (0.9, 0.75) .. 
            (1, 0.5)
            .. controls (1.1, 0.25) .. 
            (1, 0)
            .. controls (0.9, -0.25) .. 
            (1, -0.5)
            .. controls (1.1, -0.75) .. 
            (1, -1)
            arc (180:360:-1 and 0.2)
            .. controls (-0.9, -0.75) .. 
            (-1, -0.5)
            .. controls (-1.1, -0.25) .. 
            (-1, 0)
            .. controls (-0.9, 0.25) .. 
            (-1, 0.5)
            .. controls (-1.1, 0.75) .. 
            (-1, 1);
        \end{scope}
        \node[inner sep=1cm, outer sep = 0.1cm]
        (rightdown1) at (0, 1) {};
    \end{scope}
    \end{scope}

    \draw[<-] (leftup2) edge[above] node {$\xi_{B, 1}$} (center);
    \draw[<-] (leftdown1) edge[below] node {$\zeta_{A, 2}$} (center);
    \draw[<-] (rightup2) edge[above] node {$\xi_{D, 1}$} (center);
    \draw[<-] (rightdown1) edge[below] node {$\zeta_{C, 2}$} (center);
    \draw[->] (leftup) edge[above] node {$\isom{B}{D}$} (rightup);
    \draw[->] (leftdown) edge[below] node {$\isom{A}{C}$} (rightdown);
\end{tikzpicture}
\caption{
This illustrated commutative diagram shows, on the one hand,  the compatibility requirement~\eqref{diag:isomorphisms} that the isomorphism $\isom{A}{C}$ of cylinders satisfies.
On the other hand, it depicts the sewing operation, see Lemma~\ref{lem:conformal_atlas}.
}
\label{fig:conformal_atlas}
\end{figure}

\begin{proof}
The construction of a holomorphic atlas on $A \sew{2}{1} B$ is detailed in~\cite[Paragraph~II.3C]{Ahlfors-Sario:Riemann_surfaces}. 
We illustrate it in Figure~\ref{fig:conformal_atlas}.
The chart $\Psi_p \circ (\zeta_2 \cup \xi_1)^{-1}$ on the seam is defined via the composition of the map $(\zeta_2 \cup \xi_1)^{-1}$ from $A \sew{2}{1} B$ to the open cylinder $S^1 \times (-\varepsilon, \varepsilon)$ 
(which is well-defined because $\zeta_2(\theta)$ extends to some $S^1 \times [0, \varepsilon)$ and $\xi_1(\theta)$ to some $S^1 \times (-\varepsilon, 0]$ by Definition~\ref{eq:sewing_def}) and a holomorphic chart $\Psi_p$ of the standard cylinder~\eqref{eq:standard_cylinder} around $\zeta_2^{-1}(p) = \xi_1^{-1}(p)$. This proves item~\ref{item:conformal_atlas1}. 
To prove item~\ref{item:conformal_atlas2}, note that as the isomorphisms $\isom{A}{C}$ and $\isom{B}{D}$ are compatible with the boundary parametrizations, 
we have
\begin{align*}
\isom{A}{C}(\zeta_{A,2}(\theta)) = \zeta_{C,2}(\theta) \sim \xi_{D,1}(\theta) = \isom{B}{D}(\xi_{B,1}(\theta)), \qquad \theta \in S^1.
\end{align*}
Thus, $\isom{A}{C} \cup \isom{B}{D} \colon A \sew{2}{1} B \to C \sew{2}{1} D$ is well-defined and compatible with the boundary parametrizations. 
It is holomorphic on the seam, since for the charts $\Psi_p \circ (\zeta_{A,2} \cup \xi_{B,1})^{-1}$ 
and $\Psi_p \circ (\zeta_{C,2} \cup \xi_{D,1})^{-1}$, 
for $p \in S^1$, 
the map
\begin{align*}
( \Psi_p \circ  (\zeta_{C,2} \cup \xi_{D,1})^{-1} \circ \isom{A}{C} ) \cup ( \isom{B}{D} \circ (\zeta_{A,2} \cup \xi_{B,1})^{-1} \circ \Psi_p^{-1} ) 
= \Psi_p \circ \Psi_p^{-1}
\end{align*}  
 is holomorphic. Hence, $\isom{A}{C} \cup \isom{B}{D} \colon A \sew{2}{1} B \to C \sew{2}{1} D$ 
 is an isomorphism of cylinders.
\end{proof}

A convenient representative for $\paramcyl{A}{\zeta_1}{\zeta_2}$ up to isomorphisms was introduced by Neretin for annuli instead of cylinders~\cite[Theorem 7.4.2]{Neretin:book}. We will use it in computations and call it the \emph{uniformized representative}.

\begin{proposition}
\label{prop:neretin_form}
Every cylinder $\paramcyl{A}{\zeta_1}{\zeta_2} \in \cylinders$ with
modulus $\tau_A > 0$
is isomorphic to a unique cylinder $\cancyl{A} \coloneqq \paramcyl{S^1 \times [0, \tau_A]}{\xi_1}{\xi_2} \in \cylinders$, 
where $\xi_2(0) = (0, \tau_A)$.
\end{proposition}

\begin{proof}
The uniformizing map $f \colon A \to S^1 \times [0, \tau_A]$ is unique by requiring $f(\zeta_2(0)) = (0, \tau_A)$.
It gives an isomorphism from $\paramcyl{A}{\zeta_1}{\zeta_2}$ to 
$\cancyl{A}$ by setting
$\xi_1 \coloneqq f \circ \zeta_1$ and $\xi_2 \coloneqq f \circ \zeta_2$, 
which also clearly satisfies $\xi_2(0) = f(\zeta_2(0)) = (0, \tau_A)$.
\end{proof}

The analytical circle diffeomorphisms $\phi \in \Diffpan$ act on $\cylinders$ from the right by reparametrizing either one of the boundary components,
\begin{align*}
A \diffActing{1} \phi \coloneqq \paramcyl{A}{\zeta_1}{\zeta_2} & \diffActing{1} \phi
\coloneqq \paramcyl{A}{\zeta_1 \circ \phi}{\zeta_2} , \\
A \diffActing{2} \phi \coloneqq \paramcyl{A}{\zeta_1}{\zeta_2} & \diffActing{2} \phi
\coloneqq \paramcyl{A}{\zeta_1}{\zeta_2 \circ \phi}.
\end{align*}
Note that the uniformized representative is preserved by the action of $\Diffpan$ on the first boundary component in the sense that $(\cancyl{A}) \diffActing{1} \phi = \cancyl{(A \diffActing{1} \phi)}$ for $\phi \in \Diffpan$.

These actions extend to complex deformations defined in~\eqref{eq:complex_deformations}.
Consider the uniformized representative $\cancyl{A} = \paramcyl{S^1 \times [0, \tau_A]}{\xi_1}{\xi_2}$ and embed it into the infinite cylinder $S^1 \times \R$.
Indeed, if $\xi_1$ or $\xi_2$ respectively extends to $S^1 \times (-\varepsilon, \varepsilon) \to S^1 \times \R$ such that it is composable with $\phi \in \DiffC$, we can define $A \diffActing{1} \phi$ or $A \diffActing{2} \phi$ respectively by taking said composition as the new boundary parametrization and finding the cylinder bounded by the parametrizations. Moreover, it is necessary that the deformed boundary does not overlap with the other boundary.
In the case of flows of complex vector fields~\eqref{eq:floweqs}, both composability and non-overlapping conditions can be guaranteed taking $t$ sufficiently small.

The uniformized representative $\cancyl{A}$ provides a canonical flat metric $\dz|_{S^1 \times [0, \tau_A]}$ on $A$ by pullback. For the definition of real determinant lines, however, a different type of metric is needed.
The complex structure of a Riemann surface $\Sigma$ is defined by its conformal class $\conf(\Sigma)$, which is the set of all Riemannian metrics $g$ on $\Sigma$ such that in any (holomorphic) coordinate chart $\Psi$, the pushforward metric has the form $\Psi_*g = e^f \dz$ for some $f \in C^\infty(\Sigma, \R)$. If $g_1, g_2 \in \conf(\Sigma)$ are in the same conformal class, they are uniquely related by $g_2 = e^{2\sigma} g_1$ for some function $\sigma \in C^\infty(\Sigma, \R)$. Thus, given $g \in \conf(\Sigma)$, we have
\begin{align*}
\conf(\Sigma) = \big\{ e^{2 \sigma} g \: \big| \: \sigma \in C^\infty(\Sigma, \R) \big\} .
\end{align*}
\emph{Admissible metrics} are compatible with the sewing operation~\eqref{eq:sewing_def}, in the sense that the disjoint union of metrics on the left and right-hand sides form a smooth metrics across that seam.
We use the following definition of admissibility, similar to the one in~\cite{GKRV:Segals_axioms_for_Liouville_theory}.

\begin{definition}
\label{def:admissible}
A metric $g \in \conf(A)$ on a cylinder $\paramcyl{A}{\zeta_1}{\zeta_2} \in \cylinders$ is admissible 
if
\begin{align*}
\zeta_1^* g &= \dz \quad \textnormal{on $S^1 \times [0, \varepsilon)$},\\
\zeta_2^* g &= \dz \quad \textnormal{on $S^1 \times (-\varepsilon, 0]$}, \qquad \textnormal{for some $\varepsilon > 0$.}
\end{align*}
We denote the set of admissible metrics by $\cconf(A, \zeta_1, \zeta_2)$ or,
if the parametrizations are clear from context simply by $\cconf(A)$.
\end{definition}
Note that for admissible metrics, the boundary components are geodesics of length $2\pi$.

\subsection{Properties of the conformal anomaly}
\label{section:properties}

In this section, we gather key properties of the conformal anomaly defined by the pairing~\eqref{eq:pairing}, 
which are essential for the definition of the real determinant lines $\Detrc(A)$.
First, we show equivalence of the conformal anomaly defined as a pairing by Kontsevich~\&~Suhov~\cite{Kontsevich-Suhov:On_Malliavin_measures_SLE_and_CFT} 
with the form~\eqref{eq:weyl_liouville_action}. 
However, for $\partial A \neq \emptyset$, this relation only holds for admissible metrics $g \in \cconf(A) \subseteq \conf(A)$ in the sense of Definition~\ref{def:admissible}.

\begin{proposition}
\label{prop:actionsagree}
For $A \in \cylinders$ 
and $g_1, g_2 \in \cconf(A)$, let $\sigma, f_1, f_2 \in C^\infty(A, \R)$ 
be such that $g_2 = e^{2\sigma} g_1$ and locally $g_i = e^{f_i} \dz$, $i = 1, 2$. Then, the pairing~\eqref{eq:pairing},
\begin{align*}
\liou{g_1}{g_2} \coloneqq \frac{1}{48\pi \ii} \iint_A (f_1 - f_2) \partial \bar \partial (f_1 + f_2),
\end{align*}
and the conformal anomaly~\eqref{eq:weyl_liouville_action},
\begin{align*}
\lfunct(\sigma, g_1)
\coloneqq \frac{1}{12 \pi} \iint_A \bigg(
\frac{1}{2} |\nabla_{g_1} \sigma|_{g_1}^2 + R_{g_1} \sigma
\bigg) \vol_{g_1}
+ \frac{1}{12 \pi} \int_{\partial \Sigma} k_g \sigma \, \tilde \vol_g
,
\end{align*}
are related by $\liou{g_1}{g_2} = - \lfunct(\sigma, g_1)$.
\end{proposition}

\begin{proof}
The local expression of the (positive) Laplacian $\Delta_0 = -\nabla_0^2$ in the flat metric $\dz = \dd x^2 + \dd y^2$ in coordinates $z = x + \ii y$ with volume form $\vol_0 = \dd x \dd y$ is related to the complex differentials by
\begin{align*}
\partial \bar \partial f 
= \frac{\ii}{2} \Delta_0 f \: \vol_0 .
\end{align*}
Combined with the relation $f_2 = 2 \sigma + f_1$, we obtain
\begin{align}
\label{eq:pairingincoords}
\liou{g_1}{g_2} = \frac{1}{48 \pi \ii} \iint_A (-2 \sigma) \partial \bar \partial (2 \sigma + 2 f_1) 
= - \frac{1}{24 \pi} \iint_A \sigma \Delta_0 (\sigma + f_1) \vol_0.
\end{align}
On the other hand, using the conformal change of 
the Gaussian curvature
and the Laplacian
\begin{align*}
R_{g_1} = \frac{1}{2} e^{-f_1} \Delta_0 f_1,
\qquad
\Delta_{g_1} \sigma = e^{-f_1} \Delta_0 \sigma ,
\end{align*}
and the local expression for the volume form,
$\vol_{g_1} = \sqrt{|\det g_1|} \ \dd x \dd y = e^{f_1} \vol_0$, we have
\begin{align} 
\nonumber
-\lfunct(\sigma, g_1) 
&= -\frac{1}{24 \pi} \iint_A \Big( \langle \nabla_{g_1} \sigma, \nabla_{g_1} \sigma \rangle_{g_1} - \sigma \, \Delta_{g_1} \sigma
+ \sigma \, \Delta_{g_1} \sigma + 2 R_{g_1} \sigma \Big) \vol_{g_1} \\
\nonumber
&= \underbrace{-\frac{1}{24 \pi} \int_{\partial A} \sigma \, \langle  \nabla_{g_1} \sigma, N_{g_1} \rangle_{g_1} \tilde \vol_{g_1}}_{= \, 0} 
\; - \;  \frac{1}{24 \pi} \iint_A \sigma ( \Delta_{g_1} \sigma + 2 R_{g_1} ) \vol_{g_1}
\\
&= - \frac{1}{24 \pi} \iint_A \sigma ( \Delta_0 \sigma + \Delta_0 f_1 ) \, \vol_0 ,
\label{eq:applygreens}
\end{align}
where the first line follows by Green's first identity
\begin{align}
\label{eq:greens}
\iint_A \big(
\langle \nabla_{g_1} u, \nabla_{g_1} v \rangle_{g_1}
- u \, \Delta_{g_1} v 
\big) \: \vol_{g_1} = \int_{\partial A} u \, \langle \nabla_{g_1} v, N_{g_1} \rangle_{g_1} \tilde \vol_{g_1} , 
\end{align}
for $u, v \in C^\infty(A, \R)$, with 
$N_{g_1}$ being the outward pointing normal vector field on $\partial A$ with respect to $g_1$. 
In Equation~\eqref{eq:applygreens}, the normal derivative 
and the earlier term involving $k_g$ vanish 
because the admissible metrics $g_1$ and $g_2$ have to agree on a neighborhood of the boundary (i.e., $\sigma \equiv 0$ there). 
It follows that~\eqref{eq:applygreens} equals~\eqref{eq:pairingincoords}.
\end{proof}

The following basic properties facilitate the definition of the real determinant lines $\Detrc(A)$ in the next section.

\begin{proposition}
\leavevmode
\makeatletter
\@nobreaktrue
\makeatother
\label{prop:liou_properties}
\begin{enumerate}
\item \textnormal{(Diffeomorphism invariance).}
\label{item:diffeo_invariance}
Let $A, B \in \cylinders$ and let $f \colon A \to B$ be an isomorphism. Then, the pairing~\eqref{eq:pairing} is invariant under $f$, that is, 
\begin{align*}
\liou{g_1}{g_2} = \liou{f_* g_1}{f_* g_2}, \qquad \textnormal{for all} \ g_1, g_2 \in \conf(A) .
\end{align*}

\item \textnormal{(Cocycle property).}
\label{item:liou_cocycle}
Let $A \in \cylinders$ and let $g_1, g_2, g_3 \in \conf(A)$ such that at least two out of three metrics are in $\cconf(A)$. Then, the pairing~\eqref{eq:pairing} satisfies the cocycle property
\begin{align*}
\liou{g_1}{g_2} + \liou{g_2}{g_3} = \liou{g_1}{g_3}.
\end{align*}

\item \textnormal{(Antisymmetry).}
\label{item:antisymmetry}
Let $A \in \cylinders$.
The pairing~\eqref{eq:pairing} is antisymmetric, that is, 
\begin{align*}
\liou{g_1}{g_2} = -\liou{g_2}{g_1}, \qquad \textnormal{for all} \ g_1, g_2 \in \conf(A) .
\end{align*}
\end{enumerate}
\end{proposition}
Item~\ref{item:diffeo_invariance} holds in particular for isomorphisms $f \colon A \to B$ of cylinders with analytically parametrized boundaries. 

\begin{proof} 
Diffeomorphism invariance is immediate, as~\eqref{eq:pairing} is defined in local coordinates.
To prove the cocycle identity, we again use the ability to carry out integration by parts on admissible metrics like in Equation~\eqref{eq:applygreens}.
Let $g_1$ be given by $e^f \dz$ locally and $g_2 = e^{2 \sigma_2} g_1$, $g_3 = e^{2 \sigma_3} g_1$ globally. Then, we compute
\begin{align*} 
\; & \liou{g_1}{g_2} + \liou{g_2}{g_3} - \liou{g_1}{g_3} \\
= \; & \phantom{{}+{}} \, \frac{1}{48 \pi \ii} \iint_A
\big(f - (f + 2 \sigma_2)\big) \partial \bar \partial \big(f + (f + 2 \sigma_2)\big) \\
\; & + \frac{1}{48 \pi \ii} \iint_A
\big((f + 2 \sigma_2) - (f + 2 \sigma_3)\big) \partial \bar \partial
\big((f + 2 \sigma_2) + (f + 2 \sigma_3)\big)\\
\; & - \frac{1}{48 \pi \ii} \iint_A
\big(f - (f + 2 \sigma_3)\big) \partial \bar \partial
\big(f + (f + 2 \sigma_3)\big)
\\
= \; & \phantom{{}+{}} \frac{1}{12 \pi \ii} \iint_A \big(
- \sigma_2 \partial \bar \partial \sigma_2
+
(\sigma_2 - \sigma_3) \partial \bar \partial
(\sigma_2 + \sigma_3)
+ \sigma_3 \partial \bar \partial
\sigma_3
\big)\\
= \; & \phantom{{}+{}} \frac{1}{12 \pi \ii} \iint_A \big(
\sigma_2 \partial \bar \partial \sigma_3
-
\sigma_3 \partial \bar \partial
\sigma_2
\big)\\
= \; & - \frac{1}{24 \pi} \int_{\partial A} \big(
\sigma_2 N_{g_1} \sigma_3 - \sigma_3 N_{g_1} \sigma_2
\big) \, \tilde \vol_{g_1} ,
\end{align*}
which equals zero if $g_1$ and either $g_2$ or $g_3$ are admissible: 
in that case, $\sigma_2$ or $\sigma_3$ vanishes on a neighborhood of $\partial A$.
If $g_2$ and $g_3$ are admissible, but $g_1$ is not, the functions $\sigma_2$ and $\sigma_3$ agree on a neighborhood of $\partial A$ and hence
$\sigma_2 N_{g_1} \sigma_3 - \sigma_3 N_{g_1} \sigma_2 = 0$.

Lastly, antisymmetry follows directly from the Definition~\ref{eq:pairing}.
\end{proof}

\section{Real determinant lines}

\subsection{Definitions and sewing isomorphisms}
\label{section:setup_detlines}

In this section, we define the real determinant lines $\Detrc(A)$ for cylinders $A \in \cylinders$, following~\cite{Kontsevich:CFT_SLE_and_phase_boundaries, Friedrich:On_connections_of_CFT_and_SLE, 
Kontsevich-Suhov:On_Malliavin_measures_SLE_and_CFT, Dubedat:SLE_and_Virasoro_representations_localization, Benoist-Dubedat:SLE2_loop_measure}.
We introduce convenient choices of global trivializations (Proposition~\ref{prop:global_section} and Proposition~\ref{proposition:detz}) 
and the sewing operation, Equation~\eqref{eq:sewing_iso}.

Even though for concreteness we specialize to cylinders in the present work, 
these definitions generalize immediately to surfaces of any genus $\genus$ and with any number $\boundaries$ of analytically parametrized boundary components 
(thus, to the moduli spaces~\eqref{eq:all_surfaces}).
However, to obtain a real determinant line bundle over the moduli space of such surfaces of higher genus, the global trivializations based on zeta-regularized determinants of the Laplacian introduced in Section~\ref{section:genus} are needed.

\begin{definition}
\label{def:detrc}
Fix a cylinder $A \in \cylinders$ with analytical boundary parametrizations and a central charge $\charge \in \R$.
The \emph{real determinant line} of $A$ is the set 
\begin{align*}
\Detrc(A) \coloneqq (\R \times \cconf(A))/_\sim
\end{align*}
defined via the equivalence relation
\begin{align}
\label{eq:relation}
(\lambda_1, g_1) \sim (\lambda_2, g_2) 
\qquad \iff \qquad 
\lambda_1 = e^{\charge \liou{g_1}{g_2}} \lambda_2.
\end{align}
Equivalence classes in $\Detrc(A)$ are denoted $\lambda [g]$. 
The space $\Detrc(A)$ is endowed with the following real vector space structure:
\begin{itemize}
\item
scalar multiplication by $\mu \in \R$ defined as $\mu \cdot \lambda [g] = (\mu \lambda) [g]$, and

\item
addition of $\lambda_1 [g_1], \lambda_2 [g_2] \in \Detrc(A)$ defined as  
\begin{align*}
\lambda_1 [g_1] + \lambda_2 [g_2] = (\lambda_1 + e^{\charge \liou{g_1}{g_2}} \lambda_2) [g_1] = (\lambda_1 e^{\charge \liou{g_2}{g_1}} + \lambda_2) [g_2].
\end{align*}
\end{itemize}
\end{definition}

\begin{proposition}
\label{prop:well_defined}
The relation $\sim$ defined in~\eqref{eq:relation} is indeed an equivalence relation and the addition on $\Detrc(A)$ is associative and commutative. $\Detrc(A)$ is one-dimensional.
\end{proposition}

\begin{proof}
Reflexivity of $\sim$ holds since $\charge \liou{g}{g} = 0$ and symmetry follows by antisymmetry of the pairing (item~\ref{item:antisymmetry} of Proposition~\ref{prop:liou_properties}):
\begin{align*}
\lambda_1 = e^{\charge \liou{g_1}{g_2}} \lambda_2
\quad \implies \quad \lambda_2 = e^{-\charge \liou{g_1}{g_2}} \lambda_1 = e^{\charge \liou{g_2}{g_1}} \lambda_1.
\end{align*}
Transitivity follows by the cocycle property (item~\ref{item:liou_cocycle} of Proposition~\ref{prop:liou_properties}):
\begin{align*}
\begin{cases}
\lambda_1 = e^{\charge \liou{g_1}{g_2}} \lambda_2 , \\ 
\lambda_2 = e^{\charge \liou{g_2}{g_3}} \lambda_3 
\end{cases}
\quad \implies \quad \lambda_1 = e^{\charge \liou{g_1}{g_2} + \charge \liou{g_2}{g_3}} \lambda_3 = e^{\charge \liou{g_1}{g_3}} \lambda_3.
\end{align*}
The commutativity of addition again follows from antisymmetry:
\begin{align*}
\begin{aligned}
\; & \lambda_1[g_1] + \lambda_2[g_2]
\; = \; (\lambda_1 + e^{\charge \liou{g_1}{g_2}} \lambda_2) [g_1] 
\; = \; (\lambda_1 + e^{\charge \liou{g_1}{g_2}} \lambda_2) e^{\charge \liou{g_2}{g_1}} [g_2] \\
\; = \; &  (\lambda_1 e^{\charge \liou{g_2}{g_1}} + e^{\charge \liou{g_1}{g_2} + \charge \liou{g_2}{g_1}}\lambda_2) [g_2] 
\; = \; (\lambda_2 + e^{\charge \liou{g_2}{g_1}} \lambda_1) [g_2] 
\; = \; \lambda_2[g_2] + \lambda_1 [g_1] ,
\end{aligned}
\end{align*}
and the associativity uses the coycle property:
\begin{align*}
\; & (\lambda_1 [g_1] + \lambda_2 [g_2] ) + \lambda_3 [g_3]
\; = \; (\lambda_1 + e^{\charge \liou{g_1}{g_2}} \lambda_2) [g_1] + \lambda_3 [g_3] \\
\; = \; &  (\lambda_1 + e^{\charge \liou{g_1}{g_2}} \lambda_2 + e^{\charge \liou{g_1}{g_3}} \lambda_3) [g_1] 
\; = \; (\lambda_1 + e^{\charge \liou{g_1}{g_2}} \lambda_2 + e^{\charge \liou{g_1}{g_2} + \charge \liou{g_2}{g_3}} \lambda_3) [g_1] \\
\; = \; &  (\lambda_1 + e^{\charge \liou{g_1}{g_2}} (\lambda_2 + e^{\charge \liou{g_2}{g_3}} \lambda_3)) [g_1] 
\; = \; \lambda_1 [g_1] + (\lambda_2 + e^{\charge \liou{g_2}{g_3}} \lambda_3) [g_2] \\
\; = \; &  \lambda_1 [g_1] + (\lambda_2 [g_2] + \lambda_3 [g_3]).
\end{align*}
Since any two admissble metrics are conformally equivalent,
$\Detrc(A)$ is one-dimensional.
\end{proof}

Let $A, B \in \cylinders$ be isomorphic cylinders. 
The pullback of metrics induces an isomorphism of determinant lines, denoted
\begin{align*}
(\isom{A}{B})^* \colon \Detrc(B) &\longrightarrow \Detrc(A) , \\
\lambda [g] &\longmapsto \lambda [(\isom{A}{B})^*g] ,
\end{align*}
which is well-defined by the diffeomorphism invariance 
(item~\ref{item:diffeo_invariance} of Proposition~\ref{prop:liou_properties}).
Consequently, we can represent elements of any determinant line $\Detrc(A)$ in $\Detrc(\cancyl{A})$, where $\cancyl{A}$ is the uniformized representative from Proposition~\ref{prop:neretin_form}.
The cylinder $\cancyl{A}$ comes with a canonical flat metric $\dz$, 
which however it is not admissible.
The following result 
leverages item~\ref{item:liou_cocycle} of Proposition~\ref{prop:liou_properties} to still obtain an element of $\Detrc(\cancyl{A})$ that only depends on $\dz$ (as shown in Equation~\eqref{eq:global_section} below).
One may think of the collection of these elements as a global trivialization for the line bundle comprising real determinant lines over the moduli space $\moduli{0}{2}$.

\begin{proposition}
\label{prop:global_section}
Let $A \in \cylinders$ and let $\cancyl{A} = \paramcyl{S^1 \times [0, \tau_A]}{\xi_1}{\xi_2}$ be its corresponding uniformized representative from Proposition~\ref{prop:neretin_form}. 
Let $\dz$ be the flat metric on $\cancyl{A}$. 
Then,
\begin{align}
\label{eq:global_section}
\globalsection(A) \coloneqq (\isom{A}{\cancyl{A}})^*(e^{-\charge \liou{\dz}{g}} [g]) \in \Detrc(A)
\end{align} 
is nonzero and independent of $g \in \cconf(\cancyl{A})$.
\end{proposition}

\begin{proof}
The flat metric $\dz$ may not be admissible, yet, 
we still have
\begin{align*}
e^{-\charge \liou{\dz}{g_1}} [g_1] = e^{-\charge \liou{\dz}{g_2} - \charge \liou{g_2}{g_1}} [g_1] = e^{-\charge \liou{\dz}{g_2}} [g_2] , 
\end{align*}
for all $g_1, g_2 \in \cconf(\cancyl{A})$,
thanks to item~\ref{item:liou_cocycle} of Proposition~\ref{prop:liou_properties}.
\end{proof}

We now introduce an extension of the sewing operation~\eqref{eq:sewing_operation} on cylinders to the real determinant lines.
Due to the categorical formulation of Segal's axioms, a natural extension of the sewing $A \sew{2}{1} B$ is a bilinear function from the real determinant lines of $A$ and $B$ to that of $A \sew{2}{1} B$.
From Definition~\ref{def:admissible}, we see that given $g_1 \in \cconf(A)$ and $g_2 \in \cconf(B)$, the union of metric $g_1 \cup g_2$ is well-defined across the seam and admissible on $A \sew{2}{1} B$. 

\begin{definition}
\label{def:sewing_iso}
The \emph{sewing isomorphism} of real determinant lines for two cylinders is
\begin{align}
\label{eq:sewing_iso}
\begin{split}
\sewiso{A}{B} \colon \Detrc(A) \otimes \Detrc(B) &\longrightarrow \Detrc(A \sew{2}{1} B) 
, \qquad A, B \in \cylinders ,
\\
\lambda_1 [g_1] \otimes \lambda_2 [g_2] &\longmapsto \lambda_1 \lambda_2 [g_1 \cup g_2] .
\end{split}
\end{align}
\end{definition}

By locality of the conformal anomaly, we have
\begin{align}
\liou{g_1 \cup g_2}{g} = \liou{g_1}{\restrict{g}{A}} + \liou{g_2}{\restrict{g}{B}}.
\label{eq:liou_locality}
\end{align}
for $g \in \conf(A \sew{2}{1} B)$. This shows that the definition of the sewing isomorphisms is independent of the choice of metric.
Furthermore, as taking unions of metrics is compatible with pullbacks, it follows that a natural compatibility property 
for $A,B,C,D \in \cylinders$ holds,
namely, the following diagram commutes:
\begin{displaymath}
\begin{tikzcd}
{\Detrc(A) \otimes \Detrc(B)} & {\Detrc(A \sew{2}{1} B)} \\
{\Detrc(C) \otimes \Detrc(D)} & {\Detrc(C \sew{2}{1} D)}
\arrow["{\sewiso{A}{B}}", from=1-1, to=1-2]
\arrow["{(\isom{A}{C})^* \otimes (\isom{B}{D})^*}", from=2-1, to=1-1]
\arrow["{(\isom{A \sew{2}{1} B}{C \sew{2}{1} D})^*}"', from=2-2, to=1-2]
\arrow["{\sewiso{C}{D}}"', from=2-1, to=2-2]
\end{tikzcd}
\end{displaymath}

\begin{remark}
\label{remark:cocycle_cylinders}
Using the global trivialization $\globalsection$ from Proposition~\ref{prop:global_section}, 
we can define a two-cocycle $\calpha \colon \cylinders \times \cylinders \to \R \setminus \{ 0 \}$
such that
\begin{align}
\begin{split}
\sewiso{A}{B}(\globalsection(A) \otimes \globalsection(B)) 
= \; & \calpha(A, B) \; \globalsection(A \sew{2}{1} B), \\
\calpha(A, B)
= \; & e^{\charge \liou{g_0(A)}{g_A}} \, 
e^{\charge \liou{g_0(B)}{g_B}} \,
e^{- \charge \liou{g_0(A \sew{2}{1} B)}{g_A \cup g_B}} ,
\end{split}
\end{align}
where $g_A \in \cconf(A)$ and $g_B \in \cconf(B)$ are any admissible metrics 
\textnormal{(}the cocycle is independent of this choice\textnormal{)}, 
and $g_0(\blank)$ denotes the pullback of the flat metric 
$\dz$ from the uniformized representative $\cancyl{\blank}$ in Proposition~\ref{prop:neretin_form}.
The cocycle property
\begin{align*}
\calpha(A \sew{2}{1} B, C) \; \calpha(A, B) = \calpha(A, B \sew{2}{1} C) \; \calpha(B, C) , \qquad A, B, C \in \cylinders ,
\end{align*}
follows from the associativity of the sewing.
Note that the locality~\eqref{eq:liou_locality} does not make the cocycle trivial since in general, we have $g_0(A) \cup g_0(B) \neq g_0(A \sew{2}{1} B)$. 
In Section~\ref{section:computation}, 
we will extend this cocycle to $\DiffC$  (Definition~\ref{def:cocycle_deformation}).
In particular, the corresponding Lie algebra cocycle of $\calpha$ on $\DiffC$ is the sought cocycle $\algcocycle$ in Theorem~\ref{thm:main}.
\end{remark}

\subsection{Comments on general moduli spaces and Laplacian determinants}
\label{section:genus}

The properties of the conformal anomaly in Section~\ref{section:properties} readily generalize to Riemann surfaces of higher genus and with any finite number of analytically parametrized boundary components.
Therefore, Definition~\ref{def:detrc} of the real determinant line (as well as Proposition~\ref{prop:well_defined}) 
also immediately extends to this larger class of Riemann surfaces.

However, our choice of the global trivialization $\globalsection$ in Proposition~\ref{prop:global_section} 
relies on the existence of the uniformized representative of cylinders in Proposition~\ref{prop:neretin_form}, 
which in turn uses the fact that we are working with cylinders.
While we could perform a similar uniformization by embedding surfaces of genus zero into the Riemann sphere (cf.~\cite{Huang:2D_Conformal_geometry_and_VOAs}), or, more generally, 
use canonical flat metrics with geodesic boundaries 
(like in~\cite{OPS:Extremals_of_determinants_of_Laplacians}), 
we focus on another method that works in any genus and with at least one boundary component.
Indeed, one can define a global trivialization on the real determinant line bundle using the zeta-regularized determinant of the Laplacian operator, as also used in~\cite{Dubedat:SLE_and_Virasoro_representations_localization, Benoist-Dubedat:SLE2_loop_measure, GRV:Polyakovs_formulation_of_2D_bosonic_string_theory, GKRV:Segals_axioms_for_Liouville_theory}. 

We consider the \emph{positive Laplacian} (Laplace--Beltrami operator)
on a Riemann surface, $\Sigma$, with $[\Sigma] \in \moduli{\genus}{\boundaries}$, $\boundaries \in \Zpos$,
defined by 
\begin{align*}
\Delta_g \coloneqq - \frac{1}{\sqrt{\det(g)}} \sum_{i,j= 1}^2 \partial_i \sqrt{\det (g)}g^{ij} \partial_j , \qquad g \in \conf(\Sigma) ,
\end{align*}
with Dirichlet boundary conditions, 
so $\Delta_g$ has a discrete positive spectrum.  
The zeta-regularized determinant $\detz{\Delta_g}$ can be defined using its spectral zeta function~\cite{Ray-Singer:R-Torsion_and_the_Laplacian_on_Riemannian_manifolds} 
and analytic continuation (see also~\cite[Section~6~\&~Appendix~B]{Peltola-Wang:LDP}).
The change of $\detz{\Delta_g}$ under a Weyl transformation by $\sigma \in C^\infty(\Sigma, \R)$ of a metric $g \in \conf(\Sigma)$ is given by the \emph{Polyakov-Alvarez anomaly formula}~\cite{Polyakov:Quantum_geometry_of_bosonic_strings, Alvarez:Theory_of_strings_with_boundaries, 
OPS:Extremals_of_determinants_of_Laplacians},
\begin{align}
\label{eq:polyakov_alvarez}
\begin{split}
\frac{\detz{\Delta_{e^{2\sigma} g}}}{\detz{\Delta_{g}}}
= \; & \exp\bigg({-}\frac{1}{6 \pi} 
\iint_\Sigma \bigg(
\frac{1}{2} |\nabla_g \sigma|_g^2 
+ R_g \sigma
\bigg) \vol_g  \\
& \phantom{\exp\bigg(}{-}\frac{1}{6 \pi} \int_{\partial \Sigma} \Big( k_g \sigma + \frac{3}{2} \, \langle \nabla_{g} \sigma, N_{g} \rangle_{g}
\Big) \tilde \vol_g
\bigg),
\end{split}
\end{align}
where $k_g$, $\tilde \vol_{g}$ and $N_g$ are respectively the boundary curvature, the induced volume form on $\partial \Sigma$ and the outward pointing normal vector field on $\partial \Sigma$, all with respect to $g$. 

\begin{proposition}
\label{proposition:detz}
Let $\Sigma \in \moduli{\genus}{\boundaries}$ and $g \in \cconf(\Sigma)$.
The following element of the real determinant line over $\Sigma$ is independent of $g$:
\begin{align*}
\globalsectionzeta(\Sigma) \coloneqq 
(\detz{\Delta_g})^{-\charge/2} \, [g] \; \in \; \Detrc(\Sigma).
\end{align*}
It defines a global trivialization of the real determinant line bundle over $\moduli{\genus}{\boundaries}$.
\end{proposition}

\begin{proof} 
As $\Sigma$ comes with analytical boundary parametrizations, the boundary terms in~\eqref{eq:polyakov_alvarez} vanish when both $g \in \cconf(\Sigma)$ 
and $e^{2\sigma} g \in \cconf(\Sigma)$.
Proposition~\ref{prop:actionsagree} therefore implies that
\begin{align*}
\frac{\detz{\Delta_{e^{2\sigma} g}}}{\detz{\Delta_{g}}}
= e^{-2 \lfunct(\sigma, g) }
= e^{2 \liou{g}{e^{2\sigma} g} } .
\end{align*}
Taking this to the power of $-\charge/2$ cancels the factor from the equivalence relation~\eqref{eq:relation} defining the determinant lines. 
Since $(\detz{\Delta_g})^{\charge/2} > 0$, it defines a global trivialization.
\end{proof}

As the main concern of the present article are 
cylinders, it is more convenient to use the trivialization $\globalsection$ from Proposition~\ref{prop:global_section}. On cylinders, the global trivializations $\globalsection$ 
and $\globalsectionzeta$ from Proposition~\ref{proposition:detz} are explicitly related as follows.

\begin{proposition} \label{prop:relation of global sections}
For a cylinder $A = \paramcyl{S^1 \times [0, \tau]}{\zeta_1}{\zeta_2}$ in the uniformized form of Proposition~\ref{prop:neretin_form}, 
and for an admissible metric $g = e^{2\sigma} \dz \in \cconf(A)$, we have
\begin{align*}
\globalsectionzeta(A)
= \frac{\exp \Big(
\frac{\charge}{8\pi} \int_{\partial A}
(\sigma + 3)
\langle \nabla_{0} \sigma, N_{0} \rangle_{0}
\: \tilde \vol_0
\Big)}{(\detz{\Delta_{0}})^{\charge/2}} \, 
\globalsection(A) ,
\end{align*}
where the subscript \textnormal{``$0$''} refers to the flat metric $\dz = \dd \theta^2 + \dd x^2$ on $A$ in the coordinate $z = \theta + \ii x$.
\end{proposition}

\begin{proof}
In the flat metric, we have vanishing curvature $R_0 = 0$ and since the boundaries are geodesic, vanishing boundary curvature $k_0 = 0$.
Thus, the Polyakov--Alvarez anomaly formula~\eqref{eq:polyakov_alvarez} gives
\begin{align*}
\globalsectionzeta(A)
= \frac{\exp \Big(
\frac{\charge}{12 \pi}
\iint_A
\frac{1}{2} |\nabla_0 \sigma|_0^2 
\: \vol_0
+ \frac{\charge}{8\pi} \int_{\partial A}
\langle \nabla_{0} \sigma, N_{0} \rangle_{0}
\: \tilde \vol_0 
\Big)}{(\detz{\Delta_{0}})^{\charge/2}} \, e^{\charge \liou{\dz}{g}} \, \globalsection(A).
\end{align*}
Using Green's first identity~\eqref{eq:greens}, the exponents become
\begin{align} \label{eq:exponents1}
\begin{split}
\; & \frac{\charge}{12 \pi}
\iint_A
\frac{1}{2} |\nabla_0 \sigma|_0^2 
\: \vol_0
+ \frac{\charge}{8\pi} \int_{\partial A}
\langle \nabla_{0} \sigma, N_{0} \rangle_{0}
\: \tilde \vol_0 \\
= \; & \frac{\charge}{24 \pi}
\iint_A ( \sigma \Delta_0 \sigma ) \vol_0  
+ \frac{\charge}{24\pi} \int_{\partial A}
(\sigma + 3) 
\langle \nabla_{0} \sigma, N_{0} \rangle_{0}
\: \tilde \vol_0
\end{split}
\end{align}
and like in~\eqref{eq:pairingincoords} with $f_1 = 0$,
\begin{align} \label{eq:exponents2}
\charge \liou{\dz}{g}
={-} \frac{\charge}{24 \pi}
\iint_A ( \sigma \Delta_0 \sigma ) \vol_0.
\end{align}
Putting~(\ref{eq:exponents1},~\ref{eq:exponents2})  together yields the asserted identity.
\end{proof}

In Theorem~\ref{thm:loewner}, we shall compare the above global trivializations with the loop Loewner energy, which can be written in terms of ratios of zeta-regularized determinants~\cite{Wang:Equivalent_descriptions_of_the_Loewner_energy} 
(see also~\cite{Takhtajan-Teo:Weil-Petersson_metric_on_the_universal_Teichmuller_space, BFKW:On_the_functional_logdet_and_related_flows}).
To this end, we will first define the real determinant line of an analytic Jordan curve $\gamma$ on $\Sigma \in \moduli{\genus}{\boundaries}$, 
following Kontsevich~\&~Suhov~\cite{Kontsevich-Suhov:On_Malliavin_measures_SLE_and_CFT} (see also~\cite[Section~2.5.4]{Benoist-Dubedat:SLE2_loop_measure}).
Fix a real-analytic parametrization of $\gamma$ and assume that $\gamma$ separates $\Sigma$ into two connected components, whose closures we denote $L_\gamma$ and $R_\gamma$.
Then, the parametrization of $\gamma$ yields analytic boundary parametrizations for the seam in $L_\gamma$ and $R_\gamma$. 
We define
\begin{align*}
\Detrc(\gamma) \coloneqq \Detrc(\Sigma) \otimes 
(\Detrc(L_\gamma))^\vee \otimes (\Detrc(R_\gamma))^\vee.
\end{align*}
The global trivialization $\globalsectionzeta$ extends to $\Detrc(\gamma)$ by defining
\begin{align*}
\globalsectionzeta(\gamma) \coloneqq \globalsectionzeta(\Sigma) \otimes 
(\globalsectionzeta(L_\gamma))^\vee
\otimes
(\globalsectionzeta(R_\gamma))^\vee.
\end{align*}
To find the connection with Loewner energy, we employ the sewing isomorphisms~\eqref{eq:sewing_iso} to evaluate elements of $\Detrc(\gamma)$ to real numbers. By sewing the connected components along $\gamma$ and subsequently evaluating $\Detrc(\Sigma)$ with its dual, we obtain
\begin{align*}
\deteval \colon \Detrc(\gamma) \xlongrightarrow{\id \otimes (\sewiso{L_\gamma}{R_\gamma})^\vee} \Detrc(\Sigma) \otimes (\Detrc(\Sigma))^\vee \xlongrightarrow{\evaluation} \R,
\end{align*}
where $\evaluation$ denotes the canonical pairing $v \otimes v^\vee \mapsto v^\vee(v) := 1$ of a vector space and its dual.

In the special case of a loop on the Riemann sphere\footnote{Technically, we have not defined $\globalsectionzeta$ for $\hat \C \in \moduli{0}{0}$. The definition in Proposition~\ref{proposition:detz} does not generalize because the Polyakov-Alvarez anomaly formula for $\boundaries = 0$ has an additional term involving the volume. However, the precise choice of $\globalsectionzeta(\hat \C)$ is immaterial, since the moduli space $\moduli{\genus}{\boundaries}$ is a point.}
$\hat \C$, we obtain the following interpretation of $\globalsectionzeta$ by comparing $\gamma$ to the unit circle $S^1$. 
Here, $L_\gamma$ and $R_\gamma$ are the two disks separated by $\gamma$, and $\D = L_{S^1}$ and $\D^* = R_{S^1}$ are the unit disk and its complement, respectively.
The loop Loewner energy of $\gamma$ was discussed, e.g., in~\cite{Wang:Equivalent_descriptions_of_the_Loewner_energy}.
We will not use the definition in the present work, but nevertheless point out the following connection, which may be of independent interest.

\begin{theorem}
\label{thm:loewner}
For any metric $g \in \conf(\hat \C)$, we have 
\begin{align}
\label{eq:detzloewner}
\log \left( \frac{\deteval(\globalsectionzeta(\gamma))}{\detevalcircle(\globalsectionzeta(S^1))} \right)
= \frac{\charge}{2} \, 
\log \left( \frac{\detz{\Delta_{\restrict{g}{L_\gamma}}}}{\detz{\Delta_{\restrict{g}{\D}}}} \, 
\frac{\detz{\Delta_{\restrict{g}{R_\gamma}}}
}{\detz{\Delta_{\restrict{g}{\D^*}}}} \right) .
\end{align}
This expression is independent of $g$, 
and \textnormal{(}for $\charge \neq 0$\textnormal{)} 
proportional to the universal Liouville action~\textnormal{\cite{Takhtajan-Teo:Weil-Petersson_metric_on_the_universal_Teichmuller_space}}, 
or equivalently, to the loop Loewner energy of $\gamma$~\textnormal{\cite[Theorem~7.3]{Wang:Equivalent_descriptions_of_the_Loewner_energy}}. 
\end{theorem}

\begin{proof}
Let $h \in \cconf(\hat \C)$ be a metric which admits
admissible restrictions to the two disks, 
\begin{align*}
\restrict{h}{L_\gamma} \in \cconf(L_\gamma)
\qquad \textnormal{and} \qquad 
\restrict{h}{R_\gamma} \in \cconf(R_\gamma) .
\end{align*}
Since $\restrict{g}{L_\gamma} \cup \restrict{g}{R_\gamma} = g$, the evaluation of $\globalsectionzeta(\gamma)$ reads
\begin{align*}
\deteval(\globalsectionzeta(\gamma))
= 
\frac{
(\detz{\Delta_h})^{-\charge/2}
}{
(\detz{\Delta_{\restrict{h}{L_\gamma}}})^{-\charge/2} \, 
(\detz{\Delta_{\restrict{h}{R_\gamma}}})^{-\charge/2}
}.
\end{align*}
This expression is independent of $h$ because the boundary terms in the Polyakov-Alvarez anomaly formula~\eqref{eq:polyakov_alvarez} cancel on $\gamma$.
We replace $h$ by the metric $g$ and specialize $\gamma$ to $S^1$ to obtain the asserted identity~\eqref{eq:detzloewner}.
Comparing the expression with the loop Loewner energy $I^L(\gamma)$ in~\cite[Theorem~7.3]{Wang:Equivalent_descriptions_of_the_Loewner_energy}, we see that~\eqref{eq:detzloewner} equals  
$I^L(\gamma)$ when $\charge = -24$.
Moreover, since by~\cite[Theorem~1.4]{Wang:Equivalent_descriptions_of_the_Loewner_energy} 
the loop Loewner energy $I^L(\gamma)$ is proportional to 
the universal Liouville action of~\cite{Takhtajan-Teo:Weil-Petersson_metric_on_the_universal_Teichmuller_space},
so is~\eqref{eq:detzloewner}. 
\end{proof}

\subsection{The central extension}
\label{section:central_extension}

We now give a detailed construction of the generalization of $\Detrc$ to complex deformations $\phi \in \DiffC$ and the associated central extension $\Detrpc(\DiffC)$ in Equation~\eqref{eq:def_extension}.

\begin{definition}
\label{def:det_line}
Given a complex deformation $\phi \in \DiffC$ and a cylinder $A \in \cylinders$ such that $A \diffActing{1} \phi$ exists, the one-dimensional vector space
\begin{align} \label{eq:det_line}
\Detrc(\phi, A) \coloneqq \Detrc(A \diffActing{1} \phi) \otimes (\Detrc(A))^\vee
\end{align}
is the \emph{determinant line of the complex deformation $\phi$} with respect to $A$.
\end{definition}

This definition is also used by Y.-Z.~Huang in his book~\cite[Appendix~D]{Huang:2D_Conformal_geometry_and_VOAs} 
to define a central extension of $\Diffpan$ from complex determinant lines. The idea is that, although $\DiffC$ does not embed into $\moduli{0}{2}$ (which would allow a definition by pullback), 
one can take any cylinder and
deform a boundary component.
Then, if one were able to ``divide out'' the original cylinder, only the deformation would be left. 
For the determinant lines, this division is achieved by tensoring with the dual of the original determinant line, implemented by the definition in Equation~\eqref{eq:det_line}.
The following result shows that $\Detrc(\phi, A)$ indeed only depends on $A$ up to a canonical isomorphism.

\begin{theorem}
\label{thm:natisophi}
For all $\phi \in \DiffC$ and pairs $A, B \in \cylinders$ of cylinders such that $A \diffActing{1} \phi$ and $B \diffActing{1} \phi$ exist, there are isomorphisms
\begin{align*}
\natisophi{\phi}{A}{B} \colon \Detrc(\phi, A) \longrightarrow \Detrc(\phi, B),
\end{align*}
which are canonical
in the sense that for $A, B, C \in \cylinders$, we have
\begin{align}
\label{eq:natisophi_properties}
\natisophi{\phi}{A}{A} = \id_{\Detrc(\phi, A)} 
\qquad \textnormal{and} \qquad
\natisophi{\phi}{B}{C} \circ \natisophi{\phi}{A}{B}
= \natisophi{\phi}{A}{C}.
\end{align}
\end{theorem}

We prove Theorem~\ref{thm:natisophi} at the end of this section.

Using these isomorphisms, we define a determinant line incorporating all choices of cylinders by imposing 
the equivalence relation 
\begin{align*} 
a \sim b \quad  \iff \quad  \natisophi{\phi}{A}{B}(a) = b, \qquad  
\textnormal{ for } a \in \Detrc(\phi, A),  \; b \in \Detrc(\phi, B)
\end{align*}  
for each $\phi \in \DiffC$, we set
\begin{align}
\label{eq:det_equivalence_cylinders}
\Detrc(\phi) \coloneqq \bigg(\bigsqcup_{\substack{A \in \cylinders \\ \text{$A \diffActing{1} \phi$ exists}}} \Detrc(\phi, A) \bigg)\bigg/\sim .
\end{align}
A generic element of $\Detrc(\phi)$ is given by
\begin{align*}
\lambda \: [g_\phi] \otimes [g]^\vee \; \in \; \Detrc(\phi) ,
\end{align*}
for some $\lambda \in \R$ and admissible metrics $g_\phi \in \cconf(A \diffActingInline{1} \phi)$ and $g \in \cconf(A)$ on some cylinder $A \in \cylinders$ such that $A \diffActing{1} \phi$ exists.

Before we prove Theorem~\ref{thm:natisophi}, we will introduce a lemma that decomposes any cylinder $A$ by identifying a standard cylinder of small height $r > 0$, given by
\begin{align*}
\A_r \coloneqq \paramcyl{S^1 \times [0, r]}{\theta}{\theta + \ii r},
\end{align*}
at the boundary component $\partial_1 A$.
See also Figure~\ref{fig:decompositions}. 
The possibility of choosing simultaneous decompositions uniformly in $r$ over a family of cylinders is needed for the computation of the cocycle in the next section. 
This in particular requires choosing $r$ smaller than the radius of convergence
\begin{align}
\label{eq:def_rconv}
\rconv(\zeta)
\coloneqq \sup \big\{ x > 0 \: \big| \: \textnormal{$\zeta(\theta \pm \ii x)$ converges for all $\theta \in S^1$}
\big\} 
\end{align}
of each of the boundary parametrizations $\zeta$, with the sign ``$\pm$'' chosen depending on the orientation of the boundary component.
The second part of the lemma gives a way of extending elements of the real determinant line of the small cylinder to the original cylinder. 
Conceptually, these extensions are similar to the ``neutral collections'' in~\cite{Kontsevich-Suhov:On_Malliavin_measures_SLE_and_CFT}.

\begin{lemma}
\label{lemma:decomposition}
\leavevmode
\makeatletter
\@nobreaktrue
\makeatother
\begin{enumerate}
\item 
\label{item:decomposition1}
Given a collection $\paramcyl{A_i}{\zeta_{1, i}}{\zeta_{2, i}} \in \cylinders$, $i \in I$, of cylinders 
such that there exists
\begin{align*}
0 < r < \inf_{i \in I} \rconv(\zeta_{1, i}),
\end{align*}
there are simultaneous decompositions $A_i = U_i \sew{2}{1} U_i^c$, where
\begin{align}
\label{eq:small_cylinder}
U_i \coloneqq \zeta_{1, i}(\A_r) = \paramcyl{\zeta_{1, i}(\A_r)}{\zeta_{1, i}(\theta)}{\zeta_{1, i}(\theta + \ii r)}
\end{align}
is isomorphic to $\A_r$ and the complements are given by
\begin{align}
\label{eq:complement_cylinder}
U_i^c \coloneqq \paramcyl{\overline{A_i \setminus U_i}}{\zeta_{1,i}(\theta + \ii r)}{\zeta_{2,i}(\theta)}.
\end{align}

\item
\label{item:decomposition2}
Furthermore, for any $\phi \in \DiffC$ and 
$g_\phi \in \cconf(\A_r \diffActingInline{1} \phi)$, 
the metrics representing
\begin{align*}
[g_\phi] \otimes [\restrict{\dz}{\A_r}]^\vee \in \Detrc(\phi, \A_r)
\end{align*}
can be simultaneously extended to metrics 
$g_{\phi,i} \in \cconf(A_i \diffActingInline{1} \phi)$ and 
$g_i \in \cconf(A_i)$, for all $i \in I$, 
so that with decompositions from item~\ref{item:decomposition1} for all $i \in I$, we have 
\begin{align}
\label{eq:extension_conditions}
\begin{split}
\restrict{g_{\phi,i}}{U_i} &= (\zeta_{1, i})_* g_\phi \in \cconf(U_i \diffActing{1} \phi), \\
\restrict{g_i}{U_i} &= (\zeta_{1, i})_* (\restrict{\dz}{\A_r}) \in \cconf(U_i), \\
\restrict{g_{\phi,i}}{U_i^c} &= \restrict{g_i}{U_i^c} \in \cconf(U_i^c),
\end{split}
\end{align}
and the vectors $[g_{\phi,i}] \otimes [g_i]^\vee \in \Detrc(\phi, A_i)$
do not depend on the choice of the extension. 
\end{enumerate}
\end{lemma}

Note that one may replace the flat metric $\dz|_{\A_r}$, which we have chosen to use here, by any other admissible metric on $\A_r$.

\begin{proof}
For item~\ref{item:decomposition1}, the decomposition is already explicitly defined by~\eqref{eq:small_cylinder} and~\eqref{eq:complement_cylinder}, 
so we only need to observe that the parametrizations $\zeta_{1,i}(\theta + \ii r)$ of $\partial_2 U_i$ and $\partial_1 U_i^c$ indeed agree. 
For item~\ref{item:decomposition2},
since $\partial_2 U_i$ and $\partial_2 U_i \diffActingInline{1} \phi$ have the same parametrization, the pushforward metrics 
$(\zeta_{1, i})_* g_\phi$ and $(\zeta_{1, i})_* (\restrict{\dz}{\A_r})$ agree near that boundary. 
Therefore, we can simultaneously extend them to $U_i^c$ in such a way that they are admissible with respect to $\zeta_{2, i}$.
We denote the extended metrics by $g_{\phi,i}$ and $g_i$, respectively.
By locality~\eqref{eq:liou_locality} of the conformal anomaly, the vector $[g_{\phi,i}] \otimes [g_i]^\vee$ is independent of the choice of 
$\restrict{g_{\phi,i}}{U_i^c} = \restrict{g_i}{U_i^c}$.
\end{proof}

\begin{proof}[Proof of Theorem~\ref{thm:natisophi}]
Consider two cylinders
\begin{align*}
A = \paramcyl{A}{\zeta_1}{\zeta_2} 
\qquad \textnormal{and} \qquad 
B = \paramcyl{B}{\xi_1}{\xi_2} .
\end{align*}
For $0 < r < \min\{\rconv(\zeta_1), \rconv(\xi_1)\}$, we take the decompositions as in Lemma~\ref{lemma:decomposition}, 
\begin{align} \label{eq:decomposition_r}
\begin{split}
U &= \zeta_1(\A_r) = \paramcyl{\zeta_1(\A_r)}{\zeta_1(\theta)}{\zeta_1(\theta + \ii r)} , \\
U^c &= \paramcyl{\overline{A \setminus U}}{\zeta_1(\theta + \ii r)}{\zeta_2(\theta)} , \\
A &= U^c \sew{2}{1} U , \\
V &= \xi_1(\A_r) = \paramcyl{\xi_1(\A_r)}{\xi_1(\theta)}{\xi_1(\theta + \ii r)} , \\
V^c &= \paramcyl{\overline{B \setminus V}}{\xi_1(\theta + \ii r)}{\xi_2(\theta)} , \\
B &= V^c \sew{2}{1} V.
\end{split}
\end{align}
Since the complements are sewn to $\partial_2 A$ and $\partial_2 B$, respectively, 
while the action of $\DiffC$ takes place on $\partial_1 A$ and $\partial_1 B$, we similarly have the decompositions 
\begin{align} \label{eq:decomposition_r_phi}
\begin{split}
A \diffActing{1} \phi &= (U \diffActing{1} \phi) \sew{2}{1} U^c = (U \sew{2}{1} U^c) \diffActing{1} \phi , \\
B \diffActing{1} \phi &= (V \diffActing{1} \phi) \sew{2}{1} V^c = (V \sew{2}{1} V^c) \diffActing{1} \phi.
\end{split}
\end{align}
Note also that $U$ and $V$ are isomorphic via $\isom{U}{V} = \xi_1 \circ \zeta_1^{-1}$.
The sought isomorphism $\natisophi{\phi}{A}{B}$ is obtained by using the sewing isomorphisms~\eqref{eq:sewing_iso} together with Lemma~\ref{lemma:decomposition}: 
\begin{align*}
\begin{aligned}
\natisophi{\phi}{A}{B} \; \colon \Detrc(\phi, A) 
= \; & \Detrc(A \diffActing{1} \phi) \otimes (\Detrc(A))^\vee 
&& \textnormal{[by~\eqref{eq:det_line}]}
\\
\xrightarrow{\displaystyle
\sewiso{U \diffActingInline{1} \phi}{U^c}^{-1}
\otimes \sewiso{U}{U^c}^{-1}
} \; & \Detrc(U \diffActing{1} \phi) \otimes \Detrc(U^c)
\otimes \big(
\Detrc(U) \otimes \Detrc(U^c)
\big)^\vee 
&& \textnormal{[by~\eqref{eq:sewing_iso}]}
\\
\xrightarrow{\displaystyle\applyat\evaluation{2}{4}} \;  
& \Detrc(U \diffActing{1} \phi) \otimes (\Detrc(U))^\vee 
\\
\xrightarrow{\displaystyle\isom{U \diffActingInline{1} \phi}{V \diffActingInline{1} \phi} \otimes (\isom{U}{V})^\vee} \; &
\Detrc(V \diffActing{1} \phi) \otimes (\Detrc(V))^\vee \\
\xrightarrow{\displaystyle\applyat\evaluation{2}{4}^{-1}} \; &
\Detrc(V \diffActing{1} \phi) \times \Detrc(V^c)
\otimes \big(
\Detrc(V) \otimes \Detrc(V^c)
\big)^\vee 
\\
\xrightarrow{\displaystyle\sewiso{V \diffActingInline{1} \phi}{V^c} \otimes \sewiso{V}{V^c}} \; &
\Detrc(B \diffActing{1} \phi) \otimes (\Detrc(B))^\vee ,
&& \textnormal{[by~\eqref{eq:sewing_iso}]}
\end{aligned}
\end{align*}
where $\evaluation$ denotes the canonical pairing $v \otimes v^\vee \mapsto 1$, and 
$\applyat\evaluation{2}{4}$ denotes the application of $\evaluation$ to the $2$nd and $4$th tensor components.

Next, to make sure that the isomorphism $\natisophi{\phi}{A}{B}$ does not depend on the choice of $r$, take $0 < s < r$ and decompositions
\begin{align*}
A \diffActing{1} \phi = (\tilde U \diffActing{1} \phi) \sew{2}{1} \tilde U^c 
\qquad \textnormal{and} \qquad 
B \diffActing{1} \phi = (\tilde V \diffActing{1} \phi) \sew{2}{1} \tilde V^c 
\end{align*}
with respect to $\A_s$, obtained as in~(\ref{eq:decomposition_r},~\ref{eq:decomposition_r_phi}) by replacing $r$ by $s$.
Then, exactly the same computation as above with this decomposition yields another isomorphism. 
Now, since all vector spaces here are one-dimensional, 
to prove the equality of these isomorphisms, 
it is sufficient to find an element $[h_1] \otimes [h_2]^\vee \in \Detrc(\phi, A)$ given by $h_1 \in \cconf(A \diffActingInline{1} \phi)$ and $h_2 \in \cconf(A)$ such that both isomorphisms send 
$[h_1] \otimes [h_2]^\vee$ to the same element of $\Detrc(\phi, B)$.

To this end, we begin by observing that $\restrict{\dz}{\A_r} \in \cconf(\A_r)$ and $\restrict{\dz}{\A_s} \in \cconf(\A_s)$. 
Pick a metric $g_s \in \cconf(\A_s \diffActingInline{1} \phi)$.
Then, as $s < r$, we have $\A_s \subset \A_r$ and 
the metric $g_s$ can be extended to $\A_r$ by
\begin{align*}
g \coloneqq \restrict{\dz}{\A_r \setminus \A_s} \, \cup \, g_s \; \in \; \cconf(\A_r \diffActing{1} \phi).
\end{align*}
These metrics may be pushed forward along $\zeta_1$ and $\xi_1$, respectively, to form the metrics
\begin{align*}
\restrict{h_1}{U} &\coloneqq (\zeta_1)_* (\restrict{\dz}{\A_r}) , \qquad \quad \,
\restrict{h_1}{\tilde U} = (\zeta_1)_* (\restrict{\dz}{\A_s}) , \\
\restrict{h_2}{U \diffActingInline{1} \phi} &\coloneqq (\zeta_1)_* (\restrict{g}{\A_r}) , \qquad \qquad
\restrict{h_2}{\tilde U \diffActingInline{1} \phi} = (\zeta_1)_* (\restrict{g}{\A_s}),
\end{align*}
all of which are admissible.
The metrics $h_1$ and $h_2$ may now be extended further from $U$ and $U \diffActingInline{1} \phi$ 
(resp.~$\tilde U$ and $\tilde U \diffActingInline{1} \phi$)
to $A$ and $A \diffActingInline{1} \phi$ by means of Lemma~\ref{lemma:decomposition}. 
Because these extensions agree, we obtain the same vector $[h_1] \otimes [h_2]^\vee \in \Detrc(\phi, A)$ from both $s$ and $r$.

Starting with the same $g \in \cconf(\A_s \diffActingInline{1} \phi)$, the construction above can be carried out replacing $A$ and $\zeta_1$ by $B$ and $\xi_1$, 
to obtain an element $[h_3] \otimes [h_4]^\vee \in \Detrc(\phi, B)$ such that
\begin{align}
\label{eq:nofactor}
\natisophi{\phi}{A}{B}([h_1] \otimes [h_2]^\vee) = [h_3] \otimes [h_4]^\vee
\end{align}
under the isomorphisms both for $s$ and $r$, by construction.
Thus, we conclude that $\natisophi{\phi}{A}{B}$ 
is independent of the choice of $r$. 
It now also follows immediately that
$\natisophi{\phi}{A}{A} = \id_{\Detrc(\phi, A)}$, 
because in that case $U = V$ and $\isom{U}{V} = \id$,
so that each operation in the first half of the construction of $\natisophi{\phi}{A}{A}$ is reversed in the second half.

To prove the property 
$\natisophi{\phi}{B}{C} \circ \natisophi{\phi}{A}{B}
= \natisophi{\phi}{A}{C}$, 
observe that the isomorphism $\isom{U \diffActingInline{1} \phi}{V \diffActingInline{1} \phi} \otimes (\isom{U}{V})^\vee$ factors through $\Detrc(\phi, \A_r)$ since
\begin{align*}
\isom{U \diffActingInline{1} \phi}{V \diffActingInline{1} \phi} \otimes (\isom{U}{V})^\vee
=
\big(\isom{U \diffActingInline{1} \phi}{\A_r \diffActingInline{1} \phi} \otimes (\isom{U}{\A_r})^\vee\big)
\circ
\big(\isom{\A_r \diffActingInline{1} \phi}{V \diffActingInline{1} \phi} \otimes (\isom{\A_r}{V})^\vee\big).
\end{align*}
By substituting this factorization into the middle step of $\natisophi{\phi}{A}{B}$, we find that $\natisophi{\phi}{A}{B} = \natisophi{\phi}{\A_r}{B} \circ \natisophi{\phi}{A}{\A_r}$. 
Thus, for three cylinders $A, B, C \in \cylinders$, 
we obtain the following commutative diagram for $r > 0$ small enough:
\begin{displaymath}
\begin{tikzcd}[row sep = 1cm]
{\Detrc(\phi, A)} &&&& {\Detrc(\phi, B)} \\
&& {\Detrc(\phi, \A_r)} \\
\\
&& {\Detrc(\phi, C)}
\arrow["{\natisophi{\phi}{A}{B}}", from=1-1, to=1-5]
\arrow["{\natisophi{\phi}{A}{C}}"', from=1-1, to=4-3]
\arrow["{\natisophi{\phi}{B}{C}}", from=1-5, to=4-3]
\arrow["{\natisophi{\phi}{A}{\A_r}}"{description}, from=1-1, to=2-3]
\arrow["{\natisophi{\phi}{B}{\A_r}}"{description}, from=1-5, to=2-3]
\arrow["{\natisophi{\phi}{C}{\A_r}}"{description}, from=4-3, to=2-3]
\end{tikzcd}
\end{displaymath}
This implies that $\natisophi{\phi}{B}{C} \circ \natisophi{\phi}{A}{B}
= \natisophi{\phi}{A}{C}$, 
and finishes the proof.
\end{proof}

\begin{remark}
\label{remark:nofactor}
From Equation~\eqref{eq:nofactor}, we see that if Equations~\eqref{eq:extension_conditions} are satisfied, then
\begin{align*}
\natisophi{\phi}{A_i}{A_j}([g_{\phi, i}] \otimes [g_i]^\vee)
= [g_{\phi, j}] \otimes [g_j]^\vee , \qquad i, j \in I .
\end{align*}
\end{remark}

There is a convenient basis element of $\Detrc(\phi)$, which is built from the global trivialization $\globalsection$, defined
in~\eqref{eq:global_section}, and the standard cylinder $\A$ in~\eqref{eq:standard_cylinder}.
For the latter, the flat metric $\restrict{\dz}{\A} \in \cconf(\A)$ is admissible, and 
by the uniform boundedness of the imaginary part of complex deformations in~\eqref{eq:complex_deformations}, $\A \diffActing{1} \phi$ is defined for all $\phi \in \DiffC$.
Define
\begin{align}
\unif{\phi} \colon \A \diffActing{1} \phi \to S^1 \times [0, \tau_{\A \diffActing{1} \phi}],
\qquad 
\unifsq{\phi}(z) = |\unif{\phi}'(z)|^2,
\label{eq:unifsq}
\end{align}
where $\unif{\phi}$ is the isomorphism to the uniformized representative of Proposition~\ref{prop:neretin_form}.
Then, the vector
\begin{align}
\label{eq:diffglobalsection}
\globalsectionphi(\phi) \coloneqq \globalsection(\A \diffActing{1} \phi) \otimes (\globalsection(\A))^\vee 
= e^{- \charge \liou{\unifsq{\phi} \dz}{g}}[g] \otimes [\restrict{\dz}{\A}]^\vee \; \in \; \Detrc(\phi) 
\end{align}
is defined for all $\phi \in \DiffC$ and independent of $g \in \cconf(\A \diffActingInline{1} \phi)$. 
The vectors $\globalsectionphi(\phi) \in \Detrc(\phi)$ provide a global trivialization for the bundle of real determinant lines over $\DiffC$,
\begin{align}
\label{eq:Diffext}
\Detrpc(\DiffC) = \bigsqcup_{\phi \in \DiffC} 
\setsuchthat{\lambda \globalsectionphi(\phi) \in \Detrc(\phi)}{\lambda > 0}.
\end{align}
The bundle fits into the following sequence of maps (see also Corollary~\ref{cor:extension_sequence}):
\begin{equation}
\label{eq:extension_sequence}
\begin{tikzcd}[row sep = 0]
0 & {\Rp} & \Detrpc(\DiffC) & {\DiffC} & 0 \\
& \lambda & {\lambda \: \globalsectionphi(\id_{S^1})} & {\id_{S^1}} \\
&& {v \in \Detrc(\phi)} & \phi
\arrow[maps to, from=3-3, to=3-4]
\arrow[maps to, from=2-3, to=2-4]
\arrow[maps to, from=2-2, to=2-3]
\arrow[from=1-1, to=1-2]
\arrow[from=1-2, to=1-3]
\arrow[from=1-3, to=1-4]
\arrow[from=1-4, to=1-5]
\end{tikzcd}
\end{equation}
where $\Rp := (0,\infty)$ is the multiplicative group of positive real numbers.
The non-positive vectors are excluded in~\eqref{eq:Diffext} because
with the sewing operation on $\Detrpc(\DiffC)$ in the next Theorem~\ref{thm:detmult}, this sequence becomes a central extension by $\Rp$.
For $v \in \Detrc(\phi)$ and $A \in \cylinders$ such that $\A \diffActing{1} \phi$ is defined, denote by
\begin{align} \label{eq:detproj}
\detproj{A}(v) \in \Detrc(\phi, A)
\end{align}
the vector in $\Detrc(\phi, A)$ which is obtained from the quotient~\eqref{eq:det_equivalence_cylinders} by choosing the specific cylinder $A \in \cylinders$ 
to represent the real determinant line. Note that for $A, B \in \cylinders$, we have $\detproj{B} \circ \detprojinv{A} = \natisophi{\phi}{A}{B}$.

Next, fix composable $\phi_1, \phi_2 \in \DiffC$ and $A \in \cylinders$ such that the cylinders 
$A \diffActing{1} \phi_1$, $A \diffActing{1} \phi_2$, and $(A \diffActing{1} \phi_1) \diffActing{1} \phi_2 = A \diffActing{1} \phi_1 \phi_2$ are well-defined. 
We then find the multiplication isomorphism $\detmult{\phi_1}{\phi_2}$
via the following composition of isomorphisms of one-dimensional vector spaces:
\begin{align*}
\detmult{\phi_1}{\phi_2} \colon & \; \Detrc(\phi_1) \otimes \Detrc(\phi_2) \\
\xrightarrow{\displaystyle\detproj{A} \otimes \detproj{A \diffActingInline{1} \phi_1}}
& \; \underbrace{\Detrc(\phi_1, A)}_{\Detrc(A \diffActing{1} \phi_1) \, \otimes \, (\Detrc(A))^\vee} \otimes \underbrace{\Detrc(\phi_2, A \diffActing{1} \phi_1)}_{\Detrc(A \diffActing{1} \phi_1 \phi_2)
\, \otimes \, (\Detrc(A \diffActing{1} \phi_1))^\vee } 
&& \textnormal{[by~(\ref{eq:detproj},~\ref{eq:det_line})]}
\\
\xrightarrow{\displaystyle\applyat\evaluation{1}{4}}
& \; (\Detrc(A))^\vee
\otimes \Detrc(A \diffActing{1} \phi_1 \phi_2) 
\\
\xrightarrow{\displaystyle\exchange}
& \; \Detrc(A \diffActing{1} \phi_1 \phi_2) 
\otimes (\Detrc(A))^\vee
\; = \;  \Detrc(\phi_1 \phi_2, A) 
\\
\xrightarrow{\displaystyle\detprojinv{A}} & \;\Detrc(\phi_1 \phi_2) ,
&& \textnormal{[by~\eqref{eq:detproj}]}
\end{align*}
where
$\applyat\evaluation{1}{4}$ denotes the application of $\evaluation$ to the $1$st and $4$th tensor components, and $\exchange \colon v \otimes w \mapsto w \otimes v$ denotes the isomorphism that exchanges the tensor components.

\begin{theorem}
\label{thm:detmult}
The multiplication isomorphism $\detmult{\phi_1}{\phi_2}$ associated to composable complex deformations $\phi_1, \phi_2 \in \DiffC$ is independent of the choice of $A \in \cylinders$.
Moreover, the multiplication is associative in the sense that for $\phi_1, \phi_2 , \phi_3 \in \DiffC$ composable,
\begin{align}
\label{eq:det_associativity}
\detmult{\phi_1 \phi_2}{\phi_3} \circ (\detmult{\phi_1}{\phi_2} \otimes \id) = \detmult{\phi_1}{\phi_2 \phi_3} \circ (\id \otimes \: \detmult{\phi_2}{\phi_3}).
\end{align}
\end{theorem}
\begin{proof}
For independence of $A$, we will show that for any other choice of cylinder $B \in \cylinders$, 
\begin{align}
\label{eq:associativity_sufficient}
\natisophi{\phi_1 \phi_2}{A}{B} \circ \exchange \circ
\applyat\evaluation{1}{4}
= 
\exchange\circ
\applyat\evaluation{1}{4} \circ
\big(\natisophi{\phi_1}{A}{B}
\otimes \natisophi{\phi_2}{A \diffActingInline{1} \phi_1}{B \diffActingInline{1} \phi_1} \big).
\end{align}
Take metrics 
$[g_{\phi_1,A}] \otimes [g_A]^\vee \in \Detrc(\phi_1, A)$
and 
$[g_{\phi_1,B}] \otimes [g_B]^\vee \in \Detrc(\phi_1, B)$ 
as given by Lemma~\ref{lemma:decomposition} for the diffeomorphism $\phi_1$ and surfaces $A$ and $B$.
By Remark~\ref{remark:nofactor}, we have
\begin{align*}
\natisophi{\phi_1}{A}{B}([g_{\phi_1,A}] \otimes [g_A]^\vee) = [g_{\phi_1,B}] \otimes [g_B]^\vee.
\end{align*}
Similarly, apply Lemma~\ref{lemma:decomposition} for $\phi_2$ and surfaces $A \diffActingInline{1} \phi_1$ and $B \diffActingInline{1} \phi_1$ to obtain metrics $g_{\phi_1\phi_2, A} \in \cconf(A \diffActingInline{1} \phi_1 \phi_2)$ and $g_{\phi_1\phi_2, B} \in \cconf(B \diffActingInline{1} \phi_1 \phi_2)$ such that 
\begin{align*}
\natisophi{\phi_2}{A \diffActingInline{1} \phi_1}{B \diffActingInline{1} \phi_1}
(
[g_{\phi_1\phi_2, A}] \otimes [g_{\phi_1, A}]^\vee
)
=
[g_{\phi_1\phi_2, B}] \otimes [g_{\phi_1, B}]^\vee 
\end{align*}
by Remark~\ref{remark:nofactor}. 
Hence, the right-hand side of~\eqref{eq:associativity_sufficient} reads
\begin{align*}
&\exchange \big(
\applyat\evaluation{1}{4}
\big(
(\natisophi{\phi_1}{A}{B}
\otimes \natisophi{\phi_2}{A \diffActingInline{1} \phi_1}{B \diffActingInline{1} \phi_1})
( [g_{\phi_1,A}] \otimes [g_A]^\vee \otimes [g_{\phi_1\phi_2,A}] \otimes [g_{\phi_1,A}]^\vee)
\big)
\big)
\\ 
= \; &
\exchange \big(
\applyat\evaluation{1}{4}
\big(
[g_{\phi_1,B}] \otimes [g_B]^\vee \otimes [g_{\phi_1 \phi_2,B}] \otimes [g_{\phi_1,B}]^\vee
\big) 
\big)
\\ 
= \; &
\exchange([g_B]^\vee \otimes [g_{\phi_1 \phi_2,B}]) 
\\ 
= \; &
[g_{\phi_1 \phi_2,B}] \otimes [g_B]^\vee.
\end{align*}
Now, note that Lemma~\ref{lemma:decomposition} was applied two times, so we have two decompositions
\begin{align*}
A = U_A \cup U_A^c, \qquad A \diffActing{1} \phi_1 = U_{A \diffActingInline{1} \phi_1} \cup U^c_{A \diffActingInline{1} \phi_1}.
\end{align*}
If we pick the radius of the second decomposition small enough, we have $U_A^c \subset U^c_{A \diffActingInline{1} \phi_1}$.
Now, the metrics $g_{\phi_1\phi_2, A}$ and $g_A$ satisfy the extension conditions~\eqref{eq:extension_conditions} of Lemma~\ref{lemma:decomposition} for the decomposition $A = U_A \cup U_A^c$. 
Since this also holds for the respective metrics on $B$, 
by Remark~\ref{remark:nofactor} the left-hand side of~\eqref{eq:associativity_sufficient} agrees with the right-hand side:
\begin{align*}
&\natisophi{\phi_1 \phi_2}{A}{B} 
\big( 
\exchange 
\big(
\applyat\evaluation{1}{4}(
[g_{\phi_1,A}] \otimes [g_A]^\vee \otimes [g_{\phi_1 \phi_2,A}] \otimes [g_{\phi_1,A}]^\vee
)
\big)
\big)
\\ 
= \; &
\natisophi{\phi_1 \phi_2}{A}{B} 
\big( \exchange (
[g_A]^\vee \otimes [g_{\phi_1 \phi_2,A}]
)
\big)
\\ 
= \; &
\natisophi{\phi_1 \phi_2}{A}{B} (
[g_{\phi_1 \phi_2,A}] \otimes [g_A]^\vee
))
\\ 
= \; &
[g_{\phi_1 \phi_2,B}] \otimes [g_B]^\vee
,
\end{align*}
which proves the asserted identity in Equation~\eqref{eq:associativity_sufficient}.

To verify associativity, we apply independence of $A \in \cylinders$. In~\eqref{eq:det_associativity} we use, on the one hand, any fixed $A$ for $\detmult{\phi_2}{\phi_3}$ and $\detmult{\phi_1}{\phi_2 \phi_3}$. On the other hand, for $\detmult {\phi_2}{\phi_3}$ we use $A \diffActingInline{1} \phi_1$ instead and for $\detmult{\phi_1 \phi_2}{\phi_3}$ we use $A$ again.
Then, associativity follows from the commutativity of the following diagram:
\begin{displaymath}
\hspace*{-1cm}
\begin{tikzcd}[
ampersand replacement=\&,
row sep = 1cm,
column sep = 2cm,
]
{
\begin{aligned}
\Detrc(\phi_1, A) \, \otimes & \Detrc(\phi_2, A \diffActing{1} \phi_1) 
\otimes \Detrc(\phi_3, A \diffActing{1} \phi_1\phi_2)
\end{aligned}
} \& {
\begin{aligned}
&\Detrc(\phi_1\phi_2, A) \otimes \Detrc(\phi_3, A \diffActing{1} \phi_1\phi_2)
\end{aligned}
} \\
{
\begin{aligned}
&\Detrc(\phi_1, A) \otimes \Detrc(\phi_2\phi_3, A \diffActing{1} \phi_1)
\end{aligned}
} \& {\Detrc(\phi_1 \phi_2 \phi_3, A)}
\arrow[""{name=2232, anchor=center, inner sep=0}, "{\applyat\exchange{3}{4} \circ \applyat\evaluation{3}{6}}", from=1-1, to=2-1]
\arrow[""{name=2333, anchor=center, inner sep=0}, "{\applyat\exchange{1}{2} \circ \applyat\evaluation{1}{4}}", from=1-2, to=2-2]
\arrow[""{name=2223, anchor=center, inner sep=0}, "{\applyat\exchange{1}{2} \circ \applyat\evaluation{1}{4}}", from=1-1, to=1-2]
\arrow[""{name=3233, anchor=center, inner sep=0}, "{\applyat\exchange{1}{2} \circ \applyat\evaluation{1}{4}}", from=2-1, to=2-2]
\end{tikzcd}
\end{displaymath}
where $\applyat\exchange{i}{j}$ and $\applyat\evaluation{i}{j}$ denote the maps $\exchange$ 
and $\evaluation$ applied to the $i$th and $j$th tensor components. 
This diagram readily commutes, because a generic element $a \otimes \hat b \otimes c \otimes \hat d \otimes e \otimes \hat f$ is sent to $\hat d (a) \hat f(c) \;  e \otimes \hat b$ by both compositions.
\end{proof}

Having introduced the multiplication isomorphisms, we can now define the cocycle appearing in our main result (Theorem~\ref{thm:main}):

\begin{definition} \label{def:cocycle_deformation}
The (``group'') cocycle of the multiplication in Theorem~\ref{thm:detmult}
with respect to the global trivialization $\globalsectionphi$, defined in~\eqref{eq:diffglobalsection},
is the factor $\calpha(\phi, \psi)$ in
\begin{align*} 
\detmult{\phi}{\psi}(\globalsectionphi(\phi), \globalsectionphi(\psi))
= \calpha(\phi, \psi) \: \globalsectionphi(\phi \psi),
\qquad \text{$\phi, \psi \in \DiffC$ composable}.
\end{align*}
\end{definition}
Since the cocycle on cylinders (see Remark~\ref{remark:cocycle_cylinders}) is strictly positive, it is reasonable to expect that $\calpha(\phi, \psi) > 0$ 
for any composable $\phi, \psi \in \DiffC$. This will become evident from the explicit expression for $\calpha(\phi, \psi)$ constructed in Equation~\eqref{eq:alpha_cocycle} in the next section.

\begin{corollary}\label{cor:extension_sequence}
The sequence~\eqref{eq:extension_sequence} is a central extension of $\DiffC$ by the multiplicative group $\Rp$ in the sense that it is exact, respects the respective multiplications, 
and the image of $\Rp$ commutes with all of $\Detrpc(\DiffC)$.
\end{corollary}

\begin{proof}
Note that $\globalsectionphi(\id_{S^1}) = [\restrict{\dz}{\A}] \otimes [\restrict{\dz}{\A}]^\vee$, 
which implies that the map $\detmult{\id_{S^1}}{\id_{S^1}}$ sends $\globalsectionphi(\id_{S^1}) \otimes \globalsectionphi(\id_{S^1})$ just to $\globalsectionphi(\id_{S^1})$, so that we have
\begin{align*}
\detmult{\id_{S^1}}{\id_{S^1}} \big( \lambda_1 \globalsectionphi(\id_{S^1}) \otimes \lambda_2 \globalsectionphi(\id_{S^1}) \big) 
= \lambda_1 \lambda_2 \, \globalsectionphi(\id_{S^1}).
\end{align*}
Thus, the sequence~\eqref{eq:extension_sequence} indeed preserves the multiplication. 
(In fact, it is a sequence of group homomorphisms when restricted to real determinant lines over $\Diffpan$.) 
Exactness follows since any element in the fiber of $\id_{S^1}$ has the form $\lambda \globalsectionphi(\id_{S^1})$ for some $\lambda \in \Rp$.

Lastly, take any $[g_\phi] \otimes [\dz]^\vee \in \Detrc(\phi)$, with $\phi \in \DiffC$ and with the choice $A = \A_r$ in Definition~\ref{def:det_line}, such that 
$r < \rconv(\phi)$, yet still $\A_r \diffActing{1} \phi$ exists.
We have
\begin{align*}
\natisophi{\id_{S^1}}{\A_r}{\A_r \diffActingInline{1} \phi}([\dz] \otimes [\dz]^\vee) = [g] \otimes [g]^\vee , 
\qquad g \in \cconf(\A_r \diffActing{1} \phi) ,
\end{align*}
so taking $g = g_\phi$ we obtain 
\begin{align*}
\detmult{\id_{S^1}}{\phi} \big( \globalsectionphi(\id_{S^1}) \otimes [g_\phi] \otimes [\dz]^\vee \big) 
= \; & \exchange \big( 
\applyat\evaluation{1}{4} ( [\dz] \otimes [\dz]^\vee \otimes [g_\phi] \otimes [\dz]^\vee)
\big) 
\\
= \; & [g_\phi] \otimes [\dz]^\vee 
\\
= \; & \exchange \big(
\applyat\evaluation{1}{4}([g_\phi] \otimes [\dz]^\vee \otimes [g_\phi] \otimes [g_\phi]^\vee)
\big) 
\\
= \; & \detmult{\phi}{\id_{S^1}}([g_\phi] \otimes [\dz]^\vee \otimes \globalsectionphi(\id_{S^1})) .
\end{align*}
Thus, the extension~\eqref{eq:extension_sequence} is central, as stated.
\end{proof}

\newpage

\section{Computation of the Lie algebra cocycle}
\label{section:computation}

With Definition~\ref{def:cocycle_deformation} of the cocycle $\calpha$ in place, we are  ready to turn towards proving the main Theorem~\ref{thm:main}.
Since for the global trivialization~\eqref{eq:diffglobalsection}, we have fixed the cylinder $\A = \A_1$ in Equation~\eqref{eq:standard_cylinder}, 
the only step in the multiplication isomorphism 
$\detmult{\phi}{\psi}$ appearing in Theorem~\ref{thm:detmult} which may introduce a factor with respect to the global trivialization is 
\begin{align}
\label{eq:wanted_isomorphism}
\detproj{\A \diffActingInline{1} \phi}(\globalsectionphi(\psi)) = 
\natisophi{\psi}{\A}{\A \diffActingInline{1} \phi}(\globalsectionphi(\psi))
= \lambda \: \globalsection(\A \diffActingInline{1} \phi\psi) \otimes \big(\globalsection(\A \diffActingInline{1} \phi)\big)^\vee,
\end{align}
for some $\lambda \in \R$.
By following the definition of $\detmult{\phi}{\psi}$, we find that $\lambda = \calpha(\phi, \psi)$:
\begin{align*}
\detmult{\phi}{\psi}(\globalsectionphi(\phi) \otimes \globalsectionphi(\psi))
= \; & \detprojinv{\A}\Big(
\exchange\Big(\applyat{\evaluation}{1}{4}\Big(
\big(\detproj{\A} \otimes \detproj{\A \diffActing{1} \phi}\big)\big(\globalsectionphi(\phi) \otimes \globalsectionphi(\psi)\big)
\Big)\Big)
\Big)
\\ 
= \; & \detprojinv{\A}\Big(
\exchange\Big(\applyat{\evaluation}{1}{4}\Big(
\globalsection(\A \diffActing{1} \phi) \otimes \big(\globalsection(\A)\big)^\vee \otimes \natisophi{\psi}{\A}{\A \diffActingInline{1} \phi}\big(\globalsectionphi(\psi)\big)
\Big)\Big)
\Big)
\\ 
= \; & \detprojinv{\A}\Big(
\exchange\Big(
\big(\globalsection(\A)\big)^\vee \otimes \lambda \: \globalsection(\A \diffActingInline{1} \phi\psi)
\Big)
\Big)
\\
= \; & \lambda \: \globalsectionphi(\phi \psi).
\end{align*}
Note that the way this isomorphism $\natisophi{\psi}{\A}{\A \diffActingInline{1} \phi}$ is given by the proof of Theorem~\ref{thm:natisophi} is implicit.
The strategy in this section is to first find a more explicit formula for Equation~\eqref{eq:wanted_isomorphism}. 
From this, we obtain a formula for the cocycle $\calpha$ that can then be differentiated. 
Finally, the differentiated $\calpha$ is the sought Lie algebra cocycle $\algcocycle$ on $\Witt$.

\subsection{Convenient choices of metrics}

Throughout, we let $(\phi_t)_{t \in \R}$ and $(\psi_s)_{s \in \R}$ be two one-parameter families of complex deformations in $\DiffC$ such that $\phi_0 = \psi_0 = \id_{S^1}$.
For instance,
anticipating the proof of Theorem~\ref{thm:main} in Section~\ref{subsec:mainthmproof},
they could be 
flows of given vector fields $v, w \in \Witt$, defined via the flow equations~\eqref{eq:floweqs}.
(Let us cautiously note, however, that not even every diffeomorphism $\phi \in \Diffpan$ is reachable by the flow of a real, time-independent vector field --- for a counterexample, see~\cite[Warning~1.6]{Milnor:Remarks_on_infinite_dimensional_Lie_groups}.)
We first gather some observations of technical nature, crucial in order to carry out the proof.

\begin{lemma}
\label{lemma:radius_of_convergence}
Let $(\phi_t)_{t \in \R}$ and $(\psi_s)_{s \in \R}$ be analytic one-parameter families of complex deformations in $\DiffC$ such that $\phi_0 = \psi_0 = \id_{S^1}$. 
Then, there exists $\varepsilon > 0$ such that
\begin{align*}
\rconvunif \coloneqq \inf \big\{1, \rconv(\phi^{-1}_t), \rconv((\phi_t \psi_s)^{-1}) \: \big| \: t, s \in (-\varepsilon, \varepsilon) \big\} > 0 ,
\end{align*}
where the radius of convergence $\rconv(\blank)$ is defined in~\eqref{eq:def_rconv}. 
Fix $0 < r < \inf\setsuchthatinline{\rconv(\phi_t)}{t \in (-\varepsilon, \varepsilon)}$ such that 
\begin{align*}
\phi_t(S^1 \times [-r, r]) \subseteq S^1 \times (-\infty, \rconvunif] ,
\qquad \textnormal{for all } \; 
t \in (-\varepsilon, \varepsilon) .
\end{align*}
Then, we have 
$r \leq \rconv(\psi_s^{-1})$, for all $s \in (-\varepsilon, \varepsilon)$.
\end{lemma}

\begin{proof}
By compactness of $S^1$ and because the images of the complex-analytic extensions of the complex deformations include $S^1$, the quantity $\rconvunif$ is positive when $\varepsilon > 0$ is small enough.
Also, for all $t, s \in (-\varepsilon, \varepsilon)$ and $x \in [0, r]$, 
we have $\Im(\phi_t(\theta + \ii r)) \leq \rconvunif \leq \rconv((\phi_t\psi_s)^{-1})$ by the choice of $r$,
which implies that 
$((\phi_t \psi_s)^{-1} \circ \phi_t)(\theta + \ii r) = \psi_s^{-1}(\theta + \ii r)$ converges for all $\theta \in S^1$. 
Hence, we see that $r \leq \rconv(\psi_s^{-1})$ by~\eqref{eq:def_rconv}, as claimed.
\end{proof}

\begin{figure}
\centering
\includestandalone[]{fig_decompositions}
\caption{
The decomposition of $\A$ in Equation~\eqref{eq:decompa}. 
}
\label{fig:decompositions}
\end{figure}

\begin{figure}
\centering
\includestandalone[]{fig_decompositions2}
\caption{
The decomposition of $\A \diffActing{1} \phi_t$ in Equation~\eqref{eq:decompa1}. 
}
\label{fig:decompositions2}
\end{figure}

We next decompose the standard cylinder $\A$ 
(defined in~\eqref{eq:standard_cylinder}) into the following parts:
\begin{align} \label{eq:decompa}
\begin{split}
\A &= \Ua \sew{2}{1} \Uac , \\
\Ua &= \paramcyl{S^1 \times [0, r]}{\theta}{\theta + \ii r} , \\
\Uac &= \paramcyl{S^1 \times [r, 1]}{\theta + \ii r}{\theta + \ii} ,
\end{split}
\end{align}
which we illustrate in Figure~\ref{fig:decompositions}.
On the other hand, from Lemma~\ref{lemma:decomposition}  we also obtain a decomposition of $\A \diffActingInline{1} \phi_t$ into the following parts,
which we illustrate in Figure~\ref{fig:decompositions2}: 
\begin{align} \label{eq:decompa1}
\begin{split}
\A \diffActing{1} \phi_t &= \Ub \sew{2}{1} \Ubc , \\
\Ub &= \paramcyl{
\setsuchthat{\phi_t(z)}{z \in U}}{
\phi_t(\theta)}{
\phi_t(\theta + \ii r)}
\big) , \\
\Ubc &= \paramcyl
{\overline{\A \setminus \Ub}}
{\phi_t(\theta + \ii r)}
{\theta + \ii}
\big).
\end{split}
\end{align}
Recall that on $\A$, we use the complex coordinate $z = \theta + \ii x$ as in~\eqref{eq:zcoordinate}, 
so that the flat metric is $g(\A) = \dz$.

We choose smooth cut-off functions 
$\interab, \intera, \interb \colon S^1 \times (-\infty, 1] \to [0, 1]$ 
which all equal the constant $1$ in a neighborhood of $S^1 \times (-\infty, 0]$, and the constant $0$ in a neighborhood of $S^1 \times [R, 1]$. 
We require that $\interab$ and $\intera$ are independent of the $\theta$-coordinate, 
and their $x$-derivatives $\interabPrime(x)$ and $\interaPrime(x)$
respectively have support\footnote{Such a cut-off function changes from $0$ to $1$ on the support of its derivative, and it is constant elsewhere.} in 
$\Ub \diffActing{1} \psi_s$ and $\Ubc$ for any $t, s \in (-\varepsilon, \varepsilon)$. 
This setup is depicted in Figure~\ref{fig:interpolation}.
To make this possible, one has to decrease $\varepsilon > 0$ even further such that there exists $0 < \delta < \min\{R - r, \frac{r}{2}\}$ such that $\phi_t(\psi_s(S^1 \times \{0\})) \subset S^1 \times (-\infty, \delta)$ and $\phi_t(S^1 \times \{r\}) \subset S^1 \times (r - \delta, r + \delta)$ for any $s, t \in (-\varepsilon, \varepsilon)$.
Then, we can concretely require for the cut-off functions that
\begin{align}
\label{eq:support_bounds}
\begin{split}
\setsuchthat{x \in (-\infty, 1]}{\interabPrime(x) \neq 0} &\subset (\delta, r - \delta), \\
\setsuchthat{x \in (-\infty, 1]}{\interaPrime(x) \neq 0} &\subset (r + \delta, R).
\end{split}
\end{align}
We then define $\interb$ as  
\begin{align}\label{eq:interbab}
\interb(z) 
\coloneqq
\begin{cases}
0,  & \text{$x \geq r$, i.e.\ above $\Ua$,} \\ 
\interab(\phi_t(x)) , & z = \theta + \ii x \in \Ua , \\
1,  & \text{$x \leq 0$, i.e.\ below $\Ua$.} 
\end{cases}
\end{align}
Now, since $\isom{\Ub}{\Ua} = \phi_t^{-1}$, it follows that the derivative of $\interb$ has support in $\Ua$ ---   it might, however, not be independent of the $\theta$-coordinate. See Figure~\ref{fig:interpolation} for an illustration. 

\begin{figure}
\centering
\includestandalone[]{fig_interpolation}
\caption{
The hatched areas show where the cut-off functions $\intera(x)$, $\interb(z)$, and $\interab(x)$ are transitioning from value $1$ (below the area) to value $0$ (above the area).
The isomorphism between $\Ua$ and $\Ub$ is the map that relates $\interb$ to $\interab$ in Equation~\eqref{eq:interbab}.
The dashed lines are the bounds for the hatched areas as in Equation~\eqref{eq:support_bounds}. Note that the bounds are chosen such that even when $\A \diffActing{1} \phi_t$ and $\A$ are respectively deformed into $\A \diffActing{1} \phi_t \psi_s$ and $\A \diffActing{1} \psi_s$, the hatched areas do not intersect the boundaries of $\Ub \diffActing{1} \psi_s$ and $U \diffActing{1} \psi_s$ respectively.
}
\label{fig:interpolation}
\end{figure}

To shorten notation, in analogy of~\eqref{eq:unifsq}, we write
\begin{align} \label{eq:def_F_phi}
\dersq{\phi}(z) \coloneqq
\big| (\phi^{-1})'(z) \big|^2
\end{align}
for any $\phi \in \DiffC$ and $z$ in the domain of $\phi^{-1}$. 
This is the conformal factor of a pushforward of a metric along a deformation $\phi$ of a boundary component. 
Now, the following metrics are admissible:
\begin{align} \label{eq:metrics_list}
\begin{split}
g(\A) &= \dz , \\
g(\A \diffActing{1} \phi_t) &= 
\big(
\dersq{\phi_t}(z) \, 
\intera(x) + 1 - \intera(x)
\big)  \, \dz , \\
g(\A \diffActing{1} \psi_s) &= 
\big(
\dersq{\psi_s}(z)
\interb(z) + 1 - \interb(z)
\big)  \, \dz ,
\\
g(\A \diffActing{1} \phi_t \psi_s) &= 
\big(
\dersq{\phi_t \psi_s}(z) \, 
\interab(x)
+ \dersq{\phi_t}(z) \,  
\big(\intera(x) - \interab(x)\big)
+ 1 - \intera(x)
\big)  \, \dz ,
\end{split}
\end{align}
and $g(\A \diffActingInline{1} \phi_t \psi_s)$ is compatible with the parametrization of $\partial_2 \Ub = \partial_1 \Ubc$. 
Note that the functions $\dersq{\phi_t}$, $\dersq{\psi_s}$, and $\dersq{\phi_t \psi_s}$ are defined on the support of the respective cut-off functions $\intera$, $\interb$, and $\interab$ since by Lemma~\ref{lemma:radius_of_convergence}, 
the latter are bounded respectively by $R$, $r$, and $R$ in the $x$-direction.
Restrictions of the metrics~\eqref{eq:metrics_list} to the sub-surfaces from decompositions~(\ref{eq:decompa},~\ref{eq:decompa1}) are also admissible:
\begin{align} \label{eq:metrics_restrictions}
\begin{aligned}
&g(\Ua) &&= \; \restrict{g(\A)}{\Ua} &&= \; \dz , \\
&g(\Uac) &&= \; \restrict{g(\A)}{\Uac} &&= \; \dz , \\
&g(\Ua \diffActing{1} \psi_s) &&= \restrict{g(\A \diffActing{1} \psi_s)}{\Ua \diffActing{1} \psi_s} &&=
\big(
\dersq{\psi_s}(z) \, 
\interb(z) + 1 - \interb(z)
\big) \, \dz , \\
&g(\Ub) &&= \; \restrict{g(\A \diffActing{1} \phi_t)}{\Ub} &&= \;
\dersq{\phi_t}(z) \,  
\dz , \\
&g(\Ubc) &&= \; \restrict{g(\A \diffActing{1} \phi_t)}{\Ubc} &&= \;
\big(
\dersq{\phi_t}(z) \, 
\intera(x) + 1 - \intera(x)
\big)  \, \dz , \\
&g(\Ub \diffActing{1} \psi_s) &&= \; \restrict{g(\A \diffActing{1} \phi_t \psi_s)}{\Ub \diffActing{1} \psi_s} &&= \;
\big(
\dersq{\phi_t \psi_s}(z) \, 
\interab(x)
+ \dersq{\phi_t}(z)
\big(1 - \interab(x)\big)
\big) \, \dz .
\end{aligned}
\end{align}

\begin{lemma}
\label{lemma:nofactor_metrics}
The metrics in~\eqref{eq:metrics_list} with $\interb$ defined via Equation~\eqref{eq:interbab} satisfy
\begin{align} \label{eq:pullbacks}
(\phi_t^{-1})^* g(U \diffActing{1} \psi_s) 
= g(V \diffActing{1} \psi_s).
\end{align}
\end{lemma}

\begin{proof}
Using the chain rule, for $\phi, \psi \in \DiffC$, we have 
\begin{align*}
\dersq{\phi \psi}(z)
= \big| (\psi^{-1} \phi^{-1})'(z) \big|^2
= \big| (\psi^{-1})'(\phi^{-1}(z)) \big|^2 \, \big|(\psi^{-1})'(z)\big|^2
= \dersq{\psi}(\phi^{-1}(z)) \cdot \dersq{\phi}(z) ,
\end{align*}
and thus, using the fact that $\intera(x) = 1$ on $V$, 
we see that the left-hand side of~\eqref{eq:pullbacks} equals
\begin{align*}
\begin{aligned}
(\phi_t^{-1})^* g(U \diffActing{1} \psi_s) 
= \; & \big( \dersq{\psi_s}(\phi_t^{-1}(z)) \,  \interb(\phi_t^{-1}(z)) + 1 - \interb(\phi_t^{-1}(z)) \big) \dersq{\phi_t}(z) \, \dz \\
= \; & \big(\dersq{\phi_t \psi_s}(z) \, \interb(\phi_t^{-1}(z) ) + \dersq{\phi_t}(z) \big(1 - \interb(\phi_t^{-1}(z) )\big) \big) \, \dz .
\end{aligned}
\end{align*}
In turn, the right-hand side of~\eqref{eq:pullbacks} equals
\begin{align*}
g(V \diffActing{1} \psi_s)
= \big(\dersq{\phi_t \psi_s}(z) \, \interab(x) + \dersq{\phi_t}(z) \big(1 - \interab(x)\big) \big) \, \dz ,
\end{align*}
which agrees with the left-hand side of~\eqref{eq:pullbacks} by Equation~\eqref{eq:interbab}.
\end{proof}

\subsection{Differentiation of the cocycle and the proof of Theorem~\ref{thm:main}}
\label{subsec:mainthmproof}

Having the system~(\ref{eq:metrics_list},~\ref{eq:metrics_restrictions})  of metrics at hand, we now compute the sought cocycle in Equation~\eqref{eq:wanted_isomorphism}.

\begin{proposition}
Let $(\phi_t)_{t \in \R}$ and $(\psi_s)_{s \in \R}$ be analytic one-parameter families of complex deformations in $\DiffC$ such that $\phi_0 = \psi_0 = \id_{S^1}$. 
Then, for any $t, s \in (-\varepsilon, \varepsilon)$, 
the isomorphism
\begin{align}
\begin{split}
\natisophi{\psi_s}{\A}{\A \diffActingInline{1} \phi_t} \colon \Detrc(\psi_s, \A) & \longrightarrow \Detrc(\psi_s, \A \diffActingInline{1} \phi_t) , \\
\globalsection(\A \diffActingInline{1} \psi_s) \otimes \big(\globalsection(\A)\big)^\vee & \longmapsto \calpha(\phi_t, \psi_s) \:
\globalsection(\A \diffActingInline{1} \phi_t \psi_s) \otimes
\big(
\globalsection(\A \diffActingInline{1} \phi_t)
\big)^\vee 
\end{split}
\label{eq:change_ann_iso}
\end{align}
is given in terms of the metrics~\eqref{eq:metrics_list} with
\begin{align} \label{eq:alpha_cocycle}
\begin{split}
\calpha(\phi_t, \psi_s) = \exp\Big(
&\charge \liou{\unifsq{\phi_t\psi_s} \dz}{g(\A \diffActing{1} \phi_t \psi_s)} \\
&-\charge \liou{\unifsq{\phi_t} \dz}{g(\A \diffActing{1} \phi_t)}
-\charge \liou{\unifsq{\psi_s} \dz}{g(\A \diffActing{1} \psi_s)}
\Big) .
\end{split}
\end{align}
\end{proposition}

\begin{proof}
By~\eqref{eq:diffglobalsection} and Proposition~\ref{prop:global_section}, $\globalsectionphi(\psi_s)$ is given by
\begin{align}
\label{eq:global_section_s}
\globalsectionphi(\psi_s)
= e^{-\charge \liou{\unifsq{\psi_s} \dz}{g(\A \diffActingInline{1} \psi_s)}} \; 
[g(\A \diffActing{1} \psi_s)] \otimes [g(\A)]^\vee \; \in \; \Detrc(\psi_s, \A) .
\end{align}
By the choice of the metrics as in~\eqref{eq:metrics_list}, we see that (see also Figure~\ref{fig:changecomp})
\begin{align} \label{eq:natisophi1}
\natisophi{\psi_s}{A}{A \diffActingInline{1} \phi_t}
\big(
[g(\A \diffActing{1} \psi_s)] \otimes [g(\A)]^\vee
\big)
= [g(\A \diffActing{1} \phi_t \psi_s)] \otimes [g(\A \diffActing{1} \phi_t)]^\vee .
\end{align}
To obtain the cocycle $\calpha$, we compare this to~\eqref{eq:change_ann_iso} using analogues of~\eqref{eq:global_section_s} for the other cylinders:
\begin{align} \label{eq:natisophi2}
\; & e^{-\charge \liou{\unifsq{\psi_s} \dz}{g(\A \diffActingInline{1} \psi_s)}} \; [g(\A \diffActing{1} \phi_t \psi_s)] \otimes [g(\A \diffActing{1} \phi_t)]^\vee \\
= \; & \calpha(\phi_t, \psi_s) \; 
e^{-\charge \liou{\unifsq{\phi_t \psi_s} \dz}{g(\A \diffActingInline{1} \phi_t \psi_s)} + \charge \liou{\unifsq{\phi_t} \dz}{g(\A \diffActingInline{1} \phi_t)}} \; [g(\A \diffActing{1} \phi_t \psi_s)] \otimes [g(\A \diffActing{1} \phi_t)]^\vee .
\nonumber
\end{align}
Combining~(\ref{eq:natisophi1},~\ref{eq:natisophi2}) yields the asserted identity~\eqref{eq:alpha_cocycle}.
\end{proof}

\begin{figure}[t]
\centering
\begin{equation*}
\label{diag:changecomp}
\begin{adjustbox}{center}
    \begin{tikzcd}[row sep = 2em, ampersand replacement=\&]
	{\Detrc(\A \diffActing{1} \psi_s) \otimes (\Detrc(\A))^\vee}
    \& {
        \big[g\big(\A \diffActing{1} \psi_s\big)\big] \otimes \big[g(\A)\big]^\vee
    } \\
	{\Detrc\big(\Uac\big) \otimes \Detrc(\Ua \diffActing{1} \psi_s) \otimes (\Detrc(\Uac))^\vee \otimes (\Detrc(\Ua))^\vee} 
    \& {
    \big[g(\Uac)\big] \otimes \big[g(\Ua \diffActing{1} \psi_s)\big] \otimes \big[g(\Uac)\big]^\vee \otimes \big[g(\Ua)\big]^\vee} \\
	{\Detrc(\Ua \diffActing{1} \psi_s) \otimes (\Detrc(\Ua))^\vee}
    \& {
    \big[g(\Ua \diffActing{1} \psi_s)\big] \otimes \big[g(\Ua)\big]^\vee} \\
	{\Detrc(\Ub \diffActing{1} \psi_s) \otimes (\Detrc(\Ub))^\vee}
    \& {
            \big[g(\Ub \diffActing{1} \psi_s)\big] \otimes \big[g(\Ub)\big]^\vee
    } \\
	{\Detrc(\Ubc) \otimes \Detrc(\Ub \diffActing{1} \psi_s) \otimes (\Detrc(\Ubc))^\vee \otimes (\Detrc(\Ub))^\vee}
    \& {
            \big[g(\Ubc)\big] \otimes  \big[g(\Ub \diffActing{1} \psi_s)\big]
            \otimes \big[g(\Ubc)\big]^\vee \otimes \big[g(\Ub)\big]^\vee
    } \\
	{\Detrc(\A \diffActing{1} \phi_t \psi_s) \otimes (\Detrc(\A \diffActing{1} \phi_t))^\vee}
    \& {
        \big[g(\A \diffActing{1} \phi_t \psi_s)\big] \otimes \big[g(\A \diffActing{1} \phi_t)\big]^\vee
    }
    \arrow["{\text{$\sewiso{\Uac}{\Ua \diffActingInline{1} \psi_s}^{-1} \otimes \sewiso{\Uac}{\Ua}^{-1}$ (sewing)}}", from=1-1, to=2-1]
	\arrow[maps to, from=1-2, to=2-2]
    \arrow["{\text{$\applyat\evaluation{1}{3}$ (evaluation)}}", from=2-1, to=3-1]
	\arrow[maps to, from=2-2, to=3-2]
    \arrow["{\natisophi{\psi_s}{U}{V}}", from=3-1, to=4-1]
	\arrow[maps to, from=3-2, to=4-2]
    \arrow["{\text{$\applyat\evaluation{1}{3}$ (evaluation)}}"', from=5-1, to=4-1]
	\arrow[maps to, from=4-2, to=5-2]
	\arrow["\text{$\sewiso{\Ubc}{\Ub \diffActingInline{1} \psi_s}^{-1} \otimes \sewiso{\Ubc}{\Ub}^{-1}$ (sewing)}"', from=6-1, to=5-1]
	\arrow[maps to, from=5-2, to=6-2]
\end{tikzcd}
\end{adjustbox}
\end{equation*}
\caption{If a set of metrics is compatible with the decompositions as in Figures~\ref{fig:decompositions}~\&~\ref{fig:decompositions2}, 
then the multiplication isomorphism~\eqref{eq:change_ann_iso}  sends the vectors in the determinant lines induced by these metrics to each other without any additional factors. For the middle isomorphism $\natisophi{\psi_s}{U}{V}$, this is due to Lemma~\ref{lemma:nofactor_metrics}.}
\label{fig:changecomp}
\end{figure}

Recall that $\Detrpc(\DiffC)$ is an extension of $\DiffC$ by the multiplicative group $\Rp$, because Equation~\eqref{eq:alpha_cocycle} implies that 
$\calpha(\phi, \psi) > 0$ for all $\phi, \psi \in \DiffC$.
Since $\log \colon \Rp \to \R$ is an isomorphism from the multiplicative Lie group $\Rp$ to the additive Lie group $\R$, the Lie algebra cocycle $\algcocycle$ can be computed by differentiating the logarithm of $\calpha$ according to Equation~\eqref{eq:diff_cocycle}, that is, 
\begin{align} \label{eq:full_cocycle}
\algcocycle(v, w) =  
\frac{1}{2} \pdv{}{t}{s} 
\Big(
\log \calpha(\phi_t, \psi_s) - \log \calpha(\psi_s, \phi_t)\Big) \Big|_{t = s = 0},
\end{align}
where $\phi_t$ and $\psi_s$ are now the flows of complex vector fields $v, w \in \Witt$ as in~\eqref{eq:floweqs}.
This brings us to the proof of the main result of the present work.

\begin{proof}[Proof of Theorem~\ref{thm:main}]
We begin by computing the derivative of the cocycle~\eqref{eq:alpha_cocycle},
\begin{align} \label{eq:double derivative cocycle}
\pdv{}{t}{s} \log \calpha(\phi_t, \psi_s) \Big|_{t = s = 0} ,
\end{align}
up to symmetric terms, which will cancel out in the Lie algebra  cocycle~\eqref{eq:full_cocycle}. 

Note that $\liou{\unifsq{\phi_t} \dz}{g(\A \diffActingInline{1} \phi_t)}$ in Equation~\eqref{eq:alpha_cocycle} does not depend on $s$,
so it does not contribute to the derivative in~\eqref{eq:double derivative cocycle}. 
However, the term
\begin{align} \label{eq:hiddentdependence}
\begin{split}
\; & \liou{\unifsq{\psi_s} \dz}{g(\A \diffActingInline{1} \psi_s)} 
\\
= \; &
\frac{1}{48\pi \ii} \iint_\A
\Big(\big(\log \unifsq{\psi_s}
- \log\big(
\dersq{\psi_s}(z) \, 
\interb(z) + 1 - \interb(z)
\big) \big)
\\
&\phantom{\frac{1}{48\pi \ii} \iint_\A \Big(}
\times \partial \bar \partial
\big(
\log \unifsq{\psi_s}
+ \log\big(
\dersq{\psi_s}(z) \, 
\interb(z) + 1 - \interb(z)
\big) 
\big)
\Big)
\end{split}
\end{align}
does depend on $t$ via the cut-off function $\interb$. 
Nevertheless, note that setting $s = 0$ in  
$\log\big( \dersq{\psi_s}(z) \, \interb(z) + 1 - \interb(z) \big)$
and $\log\unifsq{\psi_s}$ each
yields $\log 1 = 0$, since
\begin{align} \label{eq:derivatives_one}
\dersq{\phi_0}(z) 
= \dersq{\psi_0}(z) 
= \dersq{\id_{S^1}}(z) 
= 1
= \unifsq{\id_{S^1}}(z) 
= \unifsq{\psi_0}(z) 
= \unifsq{\phi_0}(z) 
\end{align}
by~(\ref{eq:unifsq},~\ref{eq:def_F_phi}). 
Hence, applying $\pdv{s}$  to~\eqref{eq:hiddentdependence} and using the product rule under the integral, and evaluating at $s = 0$, shows that the derivative in~\eqref{eq:double derivative cocycle} of this term vanishes as well.
We thus proceed to compute the derivative of the remaining term
\begin{align}
\label{eq:computation_1}
&\liou{\unifsq{\phi_t \psi_s} \dz}{g(\A \diffActing{1} \phi_t \phi_s)} 
\\ =\: &
\nonumber
\frac{1}{48\pi \ii} \iint_\A \Big(
\big(
\log \unifsq{\phi_t \psi_s}
-
\log\big(
\dersq{\phi_t\psi_s}(z) \, \interab(x)
+ \dersq{\phi_t}(z) \, (\intera(x) - \interab(x))
+ 1 - \intera(x)
\big)\big) \\
\nonumber
&\phantom{\frac{1}{48\pi \ii} } 
\times \partial \bar \partial
\big(
\log \unifsq{\phi_t \psi_s}
+
\log\big(
\dersq{\phi_t\psi_s}(z) \, \interab(x)
+ \dersq{\phi_t}(z) \, (\intera(x) - \interab(x))
+ 1 - \intera(x)
\big)
\big)
\Big) .
\end{align}
Since functions of the form $\log(F_{\phi}(z))$ with $\phi \in \DiffC$ are harmonic, 
the integral vanishes in regions of $\A$ where 
both $\intera(x)$ and $\interab(x)$ are locally constant (with values $1$ or $0$).
Since $\interaPrime(x)$ and $\interabPrime(x)$ have disjoint supports, we can split the integral further into two parts.
We will take into account that $\interab(x) = 0$ in regions of $\A$ where $\interaPrime(x) \neq 0$, 
and $\intera(x) = 1$ in regions of $\A$ where $\interabPrime(x) \neq 0$.
We also remove the harmonic term $\log \unifsq{\phi_t \psi_s}$.
We thus obtain
\begin{align}
\textnormal{\eqref{eq:computation_1}} 
\label{eq:first}
= \; &
\phantom{{+}} 
\frac{1}{48\pi \ii} \underset{\interaPrime(z) \neq 0}{\iint}
\big( \log \unifsq{\phi_t \psi_s} \big)
\big( \partial \bar \partial
\log\big(
\dersq{\phi_t}(z) \, \intera(x) + 1 - \intera(x)
\big) \big)
\\
\label{eq:second}
&{-}
\frac{1}{48\pi \ii} \underset{\interaPrime(z) \neq 0}{\iint}
\big( 
\log\big(
\dersq{\phi_t}(z) \, \intera(x) + 1 - \intera(x)
\big) \big) 
\\
\nonumber
&\phantom{\frac{1}{48\pi \ii} \underset{\interabPrime(z) \neq 0}{\iint}} \quad
\times 
\big( \partial \bar \partial
\log\big(
\dersq{\phi_t}(z) \, \intera(x) + 1 - \intera(x)
\big) \big)
\\
\label{eq:third}
&{+}
\frac{1}{48\pi \ii} 
\underset{\interabPrime(z) \neq 0}{\iint}
\big( \log \unifsq{\phi_t \psi_s} \big) 
\big( \partial \bar \partial
\log\big(
\dersq{\phi_t\psi_s}(z) \, \interab(x) + \dersq{\phi_t}(z) (1 - \interab(x))
\big) \big) 
\\
\label{eq:fourth}
&{-}
\frac{1}{48\pi \ii} 
\underset{\interabPrime(z) \neq 0}{\iint}
\big( \log\big(
\dersq{\phi_t\psi_s}(z) \, \interab(x) + \dersq{\phi_t}(z) (1 - \interab(x))
\big) \big)
\\
\nonumber
&\phantom{\frac{1}{48\pi \ii} \underset{\interabPrime(z) \neq 0}{\iint}} \quad
\times \big(
\partial \bar \partial
\log\big(
\dersq{\phi_t\psi_s}(z) \, \interab(x) + \dersq{\phi_t}(z) (1 - \interab(x))
\big) \big)  .
\end{align}
Note that the second term~\eqref{eq:second} only depends on $\phi_t$, so its $s$-derivative vanishes. 
We proceed to take derivatives of the fourth term~\eqref{eq:fourth} by applying the product rule to 
\begin{align}
\label{eq:log1}
\log\big(
\dersq{\phi_t\psi_s}(z) \, \interab(x) + \dersq{\phi_t}(z) (1 - \interab(x))
\big).
\end{align}
Using~\eqref{eq:derivatives_one}, putting $t=s=0$ in the logarithm~\eqref{eq:log1} yields zero. Therefore, under the $\pdv{}{t}{s} \big|_{t=s=0}$-derivative of~\eqref{eq:computation_1}, only those terms with a single derivative with respect to $s$ or $t$ on each factor of the form~\eqref{eq:log1} contribute. 
A short computation shows that these first derivatives of the respective factors with $t=s=0$ equal
\begin{align}
\label{eq:dlog1}
\pdv{t}
\log\big(
\dersq{\phi_t\psi_0}(z) \, \interab(x) + \dersq{\phi_t}(z) (1 - \interab(x))
\big)
\bigg|_{t = 0}
&= {-}2 \Re(v'(z)) ,
\\ 
\label{eq:dlog2}
\pdv{s}
\log\big(
\dersq{\phi_0\psi_s}(z) \, \interab(x) + \dersq{\phi_0}(z) (1 - \interab(x))
\big)
\bigg|_{s = 0}
&= {-}2 \Re(w'(z)) \, \interab(x) ,
\end{align}
since using~(\ref{eq:def_F_phi},~\ref{eq:floweqs}), we have 
\begin{align*}
\pdv{t} \dersq{\phi_t}(z) \Big|_{t = 0} 
= \; & \pdv{t} |\phi_{-t}'(z)|^2 \Big|_{t = 0} 
= \pdv{t} \Big( \phi_{-t}'(z) \, \overline{\phi_{-t}'(z)}  \Big) \Big|_{t = 0} \\
= \; &{-}\Big( v'(\phi_{-t}(z)) \, \phi_{-t}'(z) \, \overline{\phi_{-t}'(z)}
+ \phi_{-t}'(z) \, \overline{v'(\phi_{-t}(z))} \, \overline{\phi_{-t}'(z)} \Big) \Big|_{t = 0} 
\\
= \; & {-} ( v'(z) + \overline{v'(z)}) 
={-}2 \Re(v'(z)) ,
\end{align*}
and similarly for $\smash{\pdv{s} \dersq{\psi_s}(z) \Big|_{s = 0} = 
{-}2 \Re(w'(z))}$.

We now insert equations~\eqref{eq:dlog1}~\&~\eqref{eq:dlog2} into the $\pdv{}{t}{s} \big|_{t=s=0}$-derivative of~\eqref{eq:computation_1}. 
Noting that $z \mapsto \Re(v'(z))$ is harmonic, we see that the contribution from the fourth term~\eqref{eq:fourth}~is
\begin{align} \label{eq:fourthagain}
\begin{split}
&{-}\frac{4}{48\pi \ii} 
\underset{\interabPrime(z) \neq 0}{\iint}
\Big( \Re(v'(z))
\partial \bar \partial ( \Re(w'(z)) \, \interab(x))
\; + \; \Re(w'(z)) \, \interab(x) 
\underbrace{\partial \bar \partial \Re(v'(z))}_{= \, 0}  \Big) \\
= \; & 
{-}\frac{1}{12\pi \ii} \underset{\interabPrime(z) \neq 0}{\iint}
\Re(v'(z))
\partial \bar \partial
(\Re(w'(z)) \, \interab(x)).
\end{split}
\end{align}

Next, we turn to the remaining terms~\eqref{eq:first} and~\eqref{eq:third}.  
Interestingly, involving $\unifsq{\phi_t \psi_s}$, these terms vanish for real vector fields by~\eqref{eq:unifsq}. 
For complex vector fields, their derivatives become
\begin{align}
\label{eq:int_u_t}
\begin{split}
\pdv{}{t}{s}
&\frac{1}{48\pi \ii} \underset{\interaPrime(z) \neq 0}{\iint}
\big(\log \unifsq{\phi_t \psi_s}\big)
\big(\partial \bar \partial
\log\big(
\dersq{\phi_t}(z) \, \intera(x) + 1 - \intera(x)
\big)\big)
\Big|_{t = s = 0}
\\
=
{-}&\frac{1}{24\pi \ii} \underset{\interaPrime(z) \neq 0}{\iint}
\Big(\pdv{s}
\unifsq{\psi_s}
\Big|_{s = 0}\Big)
\partial \bar \partial
\Re(v'(z)) \intera(x)
\end{split}
\end{align}
by~\eqref{eq:fourthagain}, and
\begin{align}
\label{eq:int_u_s}
\begin{split}
\pdv{}{t}{s}
&\frac{1}{48\pi \ii} 
\underset{\interabPrime(z) \neq 0}{\iint}
\big(\log \unifsq{\phi_t \psi_s}\big)
\big(\partial \bar \partial
\log\big(
\dersq{\phi_t\psi_s}(z) \, \interab(x) + \dersq{\phi_t}(z) (1 - \interab(x))
\big)\big)
\Big|_{t = s = 0}
\\
= 
{-}&\frac{1}{24\pi \ii} \underset{\interabPrime(z) \neq 0}{\iint}
\Big(\pdv{t}
\unifsq{\phi_t}
\Big|_{t = 0}\Big)
\partial \bar \partial
\Re(w'(z)) \interab(x) 
\end{split}
\end{align}
by~(\ref{eq:dlog1},~\ref{eq:dlog2}), 
where we again used the harmonicity of $z \mapsto \Re(v'(z))$.
Observe now that since the integrals are independent of the precise definitions of $\interab$ and $\intera$, the sum of~\eqref{eq:int_u_t} and~\eqref{eq:int_u_s}
is symmetric under exchange of $v$ and $w$, that is, under exchange of $\phi_t$ and $\psi_s$. 
Thus, from~\eqref{eq:fourthagain} 
the derivative~\eqref{eq:double derivative cocycle}, up to known symmetric terms, equals
\begin{align}
\label{eq:asym_terms}
{-}\frac{\charge}{12\pi \ii} \underset{\interabPrime(z) \neq 0}{\iint}
\Re(v'(z))
\partial \bar \partial
\Re(w'(z)) \, \interab(x).
\end{align}
In summary, the full cocycle $\algcocycle$ is obtained by inserting the expression~\eqref{eq:asym_terms} into~\eqref{eq:full_cocycle}, where the symmetric terms cancel out.

To finish, we shall compute the differential
$
\partial \bar \partial
(\Re(w'(z)) \, \interab(x))
$
piece by piece using the decomposition
\begin{align*}
\partial \bar \partial (fg)
= \partial((\bar \partial f) g + f (\bar \partial g))
= \partial(g (\bar \partial f) + f (\bar \partial g))
= g \partial \bar \partial f + \partial g \bar \partial f + \partial f \bar \partial g + f \partial \bar \partial g ,
\end{align*}
with $f=\Re(w'(z))$ and 
$g=\interab(x)$. We obtain the following terms (only one contributes).
\begin{enumerate}
\item
$
(
\partial
\bar \partial
\Re(w'(z))
)
\interab(x)
= 0,
$
 since $z \mapsto \Re(w'(z))$ is harmonic.
\item
Inserting
$
\Re(w'(z))
(
\partial
\bar \partial
\interab(x)
)
$
into the integral~\eqref{eq:asym_terms} yields an expression which is symmetric with respect to exchanging $v \leftrightarrow w$. Therefore, this term does not appear in the full cocycle~\eqref{eq:full_cocycle}.
\item Using the coordinate $z = \theta + \ii x$ from~\eqref{eq:zcoordinate}, we find that the cross-terms are
\begin{align*}
(
\partial
\interab(x)
) \; &
(
\bar \partial
\Re(w'(z))
)
+
(
\partial
\Re(w'(z))
)
(
\bar \partial
\interab(x)
)
\\
=\; &
\frac{\overline{w''(z)}}{2}
\frac{(- \ii) \interabPrime(x)}{2}
\dd z \dd \bar z
+
\frac{w''(z)}{2}
\frac{\ii \interabPrime(x)}{2}
\dz
\\
=\; &
\frac{\ii}{2}\frac{w''(z) - \overline{w''(z)}
}{2}
\interabPrime(x)
\dz
\\
=\; &
\ii
\Im(w''(z))
\, \interabPrime(x)
\, \dd \theta \dd x.
\end{align*}
\end{enumerate}
We conclude that~\eqref{eq:asym_terms} equals 
\begin{align*}
\textnormal{\eqref{eq:asym_terms}}
= 
{-}\frac{\charge}{12\pi}
\int_0^{2\pi}
\int_0^1
\Re(v'(\theta + \ii x))
\Im(w''(\theta + \ii x))
\, \interabPrime(x)
\, \dd x \dd \theta.
\end{align*}
Therefore, in summary, the full cocycle~\eqref{eq:full_cocycle} equals
\begin{align*}
\algcocycle(v ,w)
=\; &
{-}\frac{\charge}{24\pi}
\int_0^{2\pi}
\int_0^1
\Big(
\Re(v'(\theta + \ii x))
\Im(w''(\theta + \ii x))
\\
&\phantom{{-}\frac{\charge}{24\pi}
\int_0^{2\pi}
\int_0^1
\Big(}
- \Im(v''(\theta + \ii x))
\Re(w'(\theta + \ii x))
\Big)
\, \interabPrime(x)
\, \dd x \dd \theta
\\
=\; &
{-}\frac{\charge}{24\pi}
\int_0^1
\int_0^{2\pi}
\Big(
\Re(v'(\theta + \ii x))
\Im(w''(\theta + \ii x))
\\
&\phantom{{-}\frac{\charge}{24\pi}
\int_0^{2\pi}
\int_0^1
\Big(}
+ \Im(v'(\theta + \ii x))
\Re(w''(\theta + \ii x))
\Big)
\, \dd \theta
\, \interabPrime(x)
\, \dd x
\\
=\; &
{-}\frac{\charge}{24\pi}
\int_0^1
\left(
\int_0^{2\pi}
\Im(v'(\theta + \ii x)
w''(\theta + \ii x)
)
\, \dd \theta
\right)
\, \interabPrime(x)
\, \dd x
\\
=\; &
{-}\frac{\charge}{24\pi}
\int_0^1
\Im \left(
\int_0^{2\pi}
v'(\theta)
w''(\theta)
\, \dd \theta
\right)
\, \interabPrime(x)
\, \dd x
\\
=\; &
\phantom{{+}}\frac{\charge}{24\pi}
\Im
\int_0^{2\pi}
v'(\theta)
w''(\theta)
\, \dd \theta ,
\end{align*}
using an integration by parts with respect to $\theta$ in the second equality, and deformation of the contour integral over $v' w''$ to $x = 0$
in the fourth equality (thanks to analyticity).
\end{proof}

\appendix

\section{The conformal anomaly with boundary term}
\label{appendix:cocycle_property_boundary}

The following cocycle property (analogous to item~\ref{item:liou_cocycle} of Proposition~\ref{prop:liou_properties}) 
will be used in Appendix~\ref{appendix:diffeomorphisms_triviality} to prove triviality of the group cocycle $\calpha$ (Definition~\ref{def:cocycle_deformation}).

\begin{proposition} \label{prop:cocycle_property_boundary}
For $\sigma_1, \sigma_2 \in C^\infty(\Sigma, \R)$, we have
\begin{align}
\label{eq:cocycle_property_boundary}
\lioub{\sigma_1}{g} + \lioub{\sigma_2}{e^{2\sigma_1}g}
= \lioub{\sigma_1 + \sigma_2}{g} .
\end{align}
\end{proposition}

\begin{proof}
Similarly as in  Equations~(\ref{eq:pairingincoords},~\ref{eq:applygreens},~\ref{eq:greens}), we see that
\begin{align}
    \label{eq:stopairing}
\liou{g}{e^{2\sigma}g} =
- \frac{1}{12 \pi} \iint_\Sigma \bigg(
\frac{1}{2} |\nabla_g \sigma|_g^2 + R_g \sigma
\bigg) \vol_g 
+ \frac{1}{24 \pi} \int_{\partial \Sigma} \sigma N_{g} \sigma \, \tilde \vol_{g} ,
\end{align}
for each $\sigma \in C^\infty(\Sigma, \R)$, 
and thus
\begin{align}
\label{eq:anomalies_relation}
\lfunct(\sigma, g)
=
- \liou{g}{e^{2\sigma} g}
+ \frac{1}{12\pi} \int_{\partial \Sigma} k_g \sigma \, \tilde \vol_{g}
+ \frac{1}{24 \pi} \int_{\partial \Sigma} \sigma N_{g} \sigma \, \tilde \vol_{g}.
\end{align}
We compute~\eqref{eq:cocycle_property_boundary} for each term in~\eqref{eq:anomalies_relation} individually.
For $\liou{g}{e^{2\sigma}g}$, it follows from the computation in the proof of Proposition~\ref{prop:liou_properties} that
\begin{align}
\nonumber
& \; - \liou{g}{e^{2\sigma_1} g}
- \liou{e^{2\sigma_1} g}{e^{2\sigma_1 + 2\sigma_2} g}
+ \liou{g}{e^{2\sigma_1 + 2\sigma_2} g}
\\ = \; & 
\nonumber
-\frac{1}{24 \pi} \int_{\partial \Sigma} \big(
\sigma_1 N_{g} (\sigma_1 + \sigma_2) - (\sigma_1 + \sigma_2) N_{g} \sigma_1
\big) \, \tilde \vol_{g}
\\ = \; &
\label{eq:co_an}
\phantom{\,+\;} \frac{1}{24 \pi} \int_{\partial \Sigma} \big(
\sigma_1 N_{g} \sigma_2 - \sigma_2 N_{g} \sigma_1
\big) \, \tilde \vol_{g}.
\end{align}
For the term including the boundary curvature, 
we use the identities $\tilde \vol_{e^{2\sigma} g} = e^{\sigma} \tilde \vol_g$ and $k_{e^{2\sigma} g} = e^{- \sigma} (k_g + N_g \sigma)$ 
(for the latter, see, e.g.,~\cite[Appendix A]{Wang:Equivalent_descriptions_of_the_Loewner_energy}), 
to obtain
\begin{align}
\nonumber
& \; \frac{1}{12\pi} \int_{\partial \Sigma} k_g \sigma_1 \, \tilde \vol_{g}
+\frac{1}{12\pi} \int_{\partial \Sigma} k_{e^{2\sigma_1} g} \sigma_2 \, \tilde \vol_{e^{2\sigma_1} g}
-\frac{1}{12\pi} \int_{\partial \Sigma} k_g (\sigma_1 + \sigma_2) \, \tilde \vol_{g}
\\ = \; & 
\nonumber
\frac{1}{12\pi} \int_{\partial \Sigma} \sigma_2 (k_{g} + N_g \sigma_1) \, \tilde \vol_{g}
-\frac{1}{12\pi} \int_{\partial \Sigma} k_g \sigma_2 \, \tilde \vol_{g}
\\ = \; & 
\label{eq:co_k}
\frac{1}{12\pi} \int_{\partial \Sigma} \sigma_2 N_g \sigma_1 \, \tilde \vol_{g}.
\end{align}
Since a conformal change of the metric does not change angles, the unit normal vector fields are related by
\begin{align*}
N_{e^{2\sigma} g} = \frac{N_g}{|N_g|_{e^{2\sigma} g}} = \frac{N_g}{e^{\sigma} |N_g|_g} = e^{-\sigma} N_g.
\end{align*}
Thus, for the last term in~\eqref{eq:anomalies_relation} we obtain
\begin{align}
\nonumber
\; & 
\frac{1}{24 \pi} \int_{\partial \Sigma} ( \sigma_1 N_{g} \sigma_1 ) \, \tilde \vol_{g} 
+ \frac{1}{24 \pi} \int_{\partial \Sigma} ( \sigma_2 N_{e^{2\sigma_1} g} \sigma_2 ) \, \tilde \vol_{e^{2\sigma_1} g} 
- \frac{1}{24 \pi} \int_{\partial \Sigma} (\sigma_1 + \sigma_2) N_{g} (\sigma_1 + \sigma_2) \, \tilde \vol_{g}
\\ = \; & 
\label{eq:co_n}
- \frac{1}{24 \pi} \int_{\partial \Sigma} (
\sigma_1 N_{g} \sigma_2
+ \sigma_2 N_{g} \sigma_1
) \, \tilde \vol_{g}.
\end{align}
The asserted cocycle property~\eqref{eq:cocycle_property_boundary} now follows 
by adding~\eqref{eq:co_an},~\eqref{eq:co_k}, and~\eqref{eq:co_n}.
\end{proof}

The cocycle property in Proposition~\ref{prop:cocycle_property_boundary} implies that for a Riemann surface $\Sigma$ (with boundary), we may change Definition~\ref{def:detrc} of the real determinant line $\Detrc(\Sigma)$ to
\begin{align*}
\Detrc(\Sigma) \coloneqq (\R \times \conf(\Sigma))/_\sim,
\end{align*}
with the equivalence relation
\begin{align*}
(\lambda_1, g) \sim (\lambda_2, e^{2\sigma} g) 
\qquad \iff \qquad 
\lambda_1 = e^{- \charge \lfunct(\sigma, g)} \lambda_2.
\end{align*}
With this definition, 
the global trivialization $\globalsection(A)$ for a cylinder $A$
from Proposition~\ref{prop:global_section} reads 
\begin{align}
\label{eq:globalsectionboundary}
\globalsection(A)
= e^{-\charge \liou{g_0}{g}}[g] 
= \exp \bigg( - \frac{\charge}{24\pi} \int_{\partial \Sigma} \sigma N_{g_0} \sigma \, \tilde \vol_{g_0} \bigg) [g_0]
, \qquad A \in \cylinders ,
\end{align}
using Equation~\eqref{eq:stopairing} and $g = e^{2\sigma} g_0$.
However, with this definition of $\Detrc$ one then would have to make a conformal change to admissible metrics before applying the sewing isomorphisms in Definition~\ref{def:sewing_iso}.

\section{Triviality of the cocycle on diffeomorphisms of the circle}
\label{appendix:diffeomorphisms_triviality}

In this appendix, we give an explicit 
proof the that cocycle $\algcocycle$ vanishes on real vector fields, based on the ideas summarized in Remark~\ref{rmk:diffeomorphisms_triviality}.
In fact, we directly prove triviality of the Lie group cocycle $\calpha(\phi, \psi)$, where $\phi, \psi \in \Diffpan$  (Definition~\ref{def:cocycle_deformation}), 
overcoming the obstruction in the integration of the Lie algebra coycle pointed out in Remark~\ref{remark:obstruction}.

\begin{proposition}
\label{prop:diffeomorphisms_triviality}
For $\phi, \psi \in \Diffpan$, we have
\textnormal{(}with $\dersq{\psi}(z)$ as in~\eqref{eq:def_F_phi}\textnormal{)}
\begin{align}
\begin{split}
    \label{eq:calphaondiff}
\calpha(\phi, \psi) = \; & \bdryint{\phi \circ \psi} - \bdryint{\phi} - \bdryint{\psi},
\\
\bdryint{\phi} \coloneqq \; & \frac{1}{96\pi} \int_0^{2\pi} 
\log \dersq{\phi}(z)
\, \partial_x 
\log \dersq{\phi}(z)
\, \dd \theta
,
\end{split}
\end{align}
which is a coboundary in $H^2(\Diffpan, \R)$.
\end{proposition}

\begin{proof}
By Equation~\eqref{eq:wanted_isomorphism}, the cocycle $\calpha(\phi, \psi)$ is the factor of the isomorphism 
\begin{align} 
\natisophi{\psi}{\A}{\A \diffActing{1} \phi} \colon \Detrc(\A \diffActing{1} \psi) \otimes (\Detrc(\A))^\vee \longrightarrow \Detrc(\A \diffActing{1} \phi \circ \psi) \otimes (\Detrc(\A \diffActing{1} \phi))^\vee
\end{align}
with respect to the global trivialization $\globalsection$ in Equation~\eqref{eq:globalsectionboundary}.
Because $\psi \in \Diffpan$, the cylinder $\A \diffActing{1} \psi$ is still the uniformized representative in the sense of Proposition~\ref{prop:neretin_form}, 
and the canonical element of $\Detrc(\phi, \A)$ defined in~\eqref{eq:diffglobalsection} reads
\begin{align*}
\globalsectionphi(\psi) = e^{\charge \bdryint{\psi}} [\restrict{\dz}{\A}] \otimes [\restrict{\dz}{\A}]^\vee.
\end{align*}
Considering decompositions as in~\eqref{eq:decompa} and~\eqref{eq:decompa1} (where $\phi_t = \phi$) and following the left-hand side of the diagram in Figure~\ref{fig:changecomp}, 
we find that $\natisophi{\psi}{\A}{\A \diffActing{1} \phi}$ equals the following composition:
\begin{displaymath}
\begin{tikzcd}[ampersand replacement=\&, row sep=1em]
\Detrc(\A \diffActing{1} \psi) \otimes (\Detrc(\A))^\vee
\& e^{\charge \bdryint{\psi}} {[\restrict{\dz}{\A}] \otimes [\restrict{\dz}{\A}]^\vee}
\\
\Detrc(\Ua \diffActing{1} \psi) \otimes (\Detrc(\Ua))^\vee
\& e^{\charge \bdryint{\psi}} {[\restrict{\dz}{\Ua}] \otimes [\restrict{\dz}{\Ua}]^\vee}
\\
\Detrc(\Ub \diffActing{1} \psi) \otimes (\Detrc(\Ub))^\vee
\& e^{\charge \bdryint{\psi}} {
[\restrict{\phi_* (\dz)}{\Ub}] \otimes [ \restrict{\phi_*(\dz)}{\Ub}]^\vee
}
\\
\Detrc(\A \diffActing{1} \phi \circ \psi) \otimes (\Detrc(\A \diffActing{1} \phi))^\vee
\& e^{\charge \bdryint{\psi}} {[\restrict{\dz}{\A}] \otimes [\restrict{\dz}{\A}]^\vee}
\arrow[from=1-1, to=2-1]
\arrow[from=2-1, to=3-1]
\arrow[from=3-1, to=4-1]
\arrow[from=1-2, to=2-2, maps to]
\arrow[from=2-2, to=3-2, maps to]
\arrow[from=3-2, to=4-2, maps to]
\end{tikzcd}
\end{displaymath}
where the cylinders $U$ and $V$ (resp.~$U \diffActing{1} \psi$ and $V \diffActing{1} \psi$) are isomorphic and the isomorphism is given by $\phi$ in both cases.
Note that the conformal anomalies from the transformation of $\restrict{\phi_*(\dz)}{V}$ to $\restrict{\dz}{V}$ cancel out.
Now, since
\begin{align}
    \globalsection(A \diffActing{1} \phi \circ \psi) \otimes \globalsection(\A \diffActing{1} \phi)^\vee = e^{\charge \bdryint{\phi \circ \psi}} [\restrict{\dz}{\A}] \otimes e^{- \charge \bdryint{\phi}} [\restrict{\dz}{\A}]^\vee,
\end{align}
we conclude that the factor introduced by the
isomorphism $\natisophi{\psi}{\A}{\A \diffActing{1} \phi}$ with respect to the global trivialization $\globalsection$ equals $\exp ( \calpha(\phi, \psi) )$, 
where $\calpha(\phi, \psi)$ is given by Equation~\eqref{eq:calphaondiff}.
\end{proof}

\bigskip

\bibliographystyle{annotate}

\begin{thebibliography}{BGKRV24}

\bibitem[AHS23]{AHS:SLE_loop_via_conformal_welding_of_quantum_disks}
Morris Ang, Nina Holden, and Xin Sun.
\newblock The {SLE} loop via conformal welding of quantum disks.
\newblock {\em Electron. J.~Probab.}, 28:1--20, 2023.

\bibitem[Alv83]{Alvarez:Theory_of_strings_with_boundaries}
Orlando Alvarez.
\newblock Theory of strings with boundaries: {F}luctuations, topology and
  quantum geometry.
\newblock {\em Nucl. Phys. B}, 216(1):125--184, 1983.

\bibitem[AM24]{Alekseev-Meinrenken:Symplectic_geometry_of_Teichmuller_spaces_for_surfaces_with_ideal_boundary}
Anton Alekseev and Eckhard Meinrenken.
\newblock Symplectic geometry of {T}eichm\"uller spaces for surfaces with ideal
  boundary.
\newblock {\em Comm. Math. Phys.}, 405(10):1--46 (Article 229), 2024.

\bibitem[ARS25]{ARS:The_moduli_of_annuli_in_random_conformal_geometry}
Morris Ang, Guillaume R\'emy, and Xin Sun.
\newblock The moduli of annuli in random conformal geometry.
\newblock Ann. Sci. {\'E}c. Norm. Sup{\'e}r, to appear, 2025. 

\bibitem[AS60]{Ahlfors-Sario:Riemann_surfaces}
Lars~Valerian Ahlfors and Leo Sario.
\newblock {\em Riemann surfaces}.
\newblock Princeton University Press, 1960.

\bibitem[AST24]{AST:Berezin_quantization_conformal_welding_and_the_Bott-Virasoro_group}
Anton Alekseev, Samson~L. Shatashvili, and Leon~A. Takhtajan.
\newblock Berezin quantization, conformal welding and the {B}ott-{V}irasoro
  group. 
\newblock {\em Ann. Henri Poincar\'e}, 25(1):35--64, 2024.

\bibitem[BD16]{Benoist-Dubedat:SLE2_loop_measure}
St{\'e}phane Benoist and Julien Dub{\'e}dat.
\newblock An $\mathrm{SLE}_{2}$ loop measure.
\newblock {\em Ann. Henri Poincar{\'e}, Probab. Statist.}, 52(3):1406--1436,
  2016.

\bibitem[BFKM94]{BFKW:On_the_functional_logdet_and_related_flows}
Dan Burghelea, Lennie Friedlander, Thomas Kappeler, and Patrick McDonald.
\newblock On the functional logdet and related flows on the space of closed
  embedded curves on {$S^2$}.
\newblock {\em J.~Funct. Anal}, 120(2):440--466, 1994.

\bibitem[BGKRV24]{BGKRV:Virasoro_structure_and_the_scattering_matrix_for_Liouville_CFT}
Guillaume Baverez, Colin Guillarmou, Antti Kupiainen, R\'emi Rhodes, and
  Vincent Vargas.
\newblock The {V}irasoro structure and the scattering matrix for {L}iouville
  conformal field theory.
\newblock {\em Prob. Math. Phys.}, To appear, 2024.

\bibitem[BGKR24]{BGKR:Semigroup_of_annuli_in_Liouville_CFT}
Guillaume Baverez, Colin Guillarmou, Antti Kupiainen, and R\'emi Rhodes.
\newblock Semigroup of annuli in {L}iouville {CFT}.
\newblock Preprint in arXiv:2403.10914, 2024.

\bibitem[BJ24]{Baverez-Jego:The_CFT_of_SLE_loop_measures_and_the_Kontsevich-Suhov_conjecture}
Guillaume Baverez and Antoine Jego.
\newblock The {CFT} of {SLE} loop measures and the {K}ontsevich-{S}uhov conjecture. 
\newblock Preprint in arXiv:2407.09080, 2024.

\bibitem[Bot77]{Bott:On_the_characteristic_classes_of_groups_of_diffeomorphisms}
Raoul Bott.
\newblock On the characteristic classes of groups of diffeomorphisms.
\newblock {\em Enseign. Math.}, 23:209--220, 1977.

\bibitem[BR87a]{Bowick-Rajeev:The_holomorphic_geometry_of_closed_bosonic_string_theory_and_DiffS1modS1}
Mark~J. Bowick and Sarada~G. Rajeev.
\newblock The holomorphic geometry of closed bosonic string theory and
  {$\mathrm{Diff}(S^1)/S^1$}.
\newblock {\em Nucl. Phys. B}, 293:348--384, 1987.

\bibitem[BR87b]{Bowick-Rajeev:String_theory_as_the_Kahler_geometry_of_loop_space}
Mark~J. Bowick and Sarada~G. Rajeev.
\newblock String theory as the {K}{\"a}hler geometry of loop space.
\newblock {\em Phys. Rev. Lett.}, 58(6):535--538, 1987.

\bibitem[CP14]{Chavez-Pickrell:Werners_measure_on_self-avoiding_loops_and_welding}
Angel Chavez and Doug Pickrell.
\newblock Werner's measure on self-avoiding loops and welding.
\newblock {\em SIGMA Symmetry Integrability Geom. Methods Appl.}, 10(81):1--42,
  2014.

\bibitem[CW23]{Carfagnini-Wang:OM_functional_for_SLE_loop_measures}
Marco Carfagnini and Yilin Wang.
\newblock Onsager {M}achlup functional for {$\mathrm{SLE}_\kappa$} loop
  measures.
\newblock {\em Comm. Math. Phys.}, 405(11):1--13 (Article 258), 2024.

\bibitem[dAI98]{Azcarraga-Izquierdo:Lie_groups_Lie_algebras_cohomology_and_some_applications_in_physics}
Jos{\'e}~A. de~Azc{\'a}rraga and Jos{\'e}~M. Izquierdo.
\newblock {\em Lie groups, {L}ie algebras, cohomology, and some applications in
  physics}.
\newblock Cambridge monographs on mathematical physics. Cambridge University
  Press, 1998.

\bibitem[DS11]{Duplantier-Sheffield:LQG_and_KPZ}
Bertrand Duplantier and Scott Sheffield.
\newblock Liouville quantum gravity and {KPZ}.
\newblock {\em Invent. Math.}, 185(2):333--393, 2011.

\bibitem[Dub15]{Dubedat:SLE_and_Virasoro_representations_localization}
Julien Dub{\'e}dat.
\newblock $\mathrm{SLE}$ and {V}irasoro representations: localization.
\newblock {\em Comm. Math. Phys.}, 336(2):695--760, 2015.

\bibitem[FBZ04]{Frenkel-Ben-Zvi:Vertex_Algebras_and_Algebraic_Curves}
Igor~B. Frenkel and Devid Ben-Zvi.
\newblock {\em Vertex {A}lgebras and {A}lgebraic {C}urves}, volume~88 of {\em
  Mathematical Surveys and Monographs}.
\newblock American Mathematical Society, 2 edition, 2004.

\bibitem[Fri04]{Friedrich:On_connections_of_CFT_and_SLE}
Roland Friedrich.
\newblock On connections of conformal field theory and stochastic {L}oewner
  evolution.
\newblock Preprint in arXiv:math-ph/0410029, 2004.

\bibitem[FS87]{Friedan-Shenker:The_analytic_geometry_of_two-dimensional_conformal_field_theory}
Daniel Friedan and Stephen Shenker.
\newblock The analytic geometry of two-dimensional conformal field theory.
\newblock {\em Nucl. Phys. B}, 281:509--545, 1987.

\bibitem[Gaw99]{Gawedzki:CFT_lectures}
Krzysztof Gaw{\c e}dzki.
\newblock Lectures on conformal field theory.
\newblock In {\em Quantum fields and strings: {A} course for mathematicians},
  pages 727--805. American Mathematical Society, 1999.

\bibitem[GKRV21]{GKRV:Segals_axioms_for_Liouville_theory}
Colin Guillarmou, Antti Kupiainen, R{\'e}mi Rhodes, and Vincent Vargas.
\newblock Segal's axioms and bootstrap for {L}iouville {T}heory.
\newblock Preprint in arXiv:2112.14859, 2021.

\bibitem[GL06]{Gordina-Lescot:Riemannian_geometry_of_DiffS1}
Maria Gordina and Paul Lescot.
\newblock Riemannian geometry of {$\mathrm{Diff}(S^1)/S^1$}.
\newblock {\em J.~Func. Anal.}, 239:611--630, 2006.

\bibitem[Gor08]{Gordina:Riemannian_geometry_of_DiffS1_revisited}
Maria Gordina.
\newblock Riemannian geometry of {$\mathrm{Diff}(S^1)/S^1$} revisited.
\newblock In {\em Stochastic analysis in mathematical physics}, pages 19--29.
  World Sci. Publ., Hackensack, NJ, 2008.

\bibitem[GR07]{Guieu-Roger:Virasoro_book}
Laurent Guieu and Claude Roger.
\newblock {\em L'alg{\`e}bre et le groupe de {V}irasoro: aspects
  g{\'e}om{\'e}triques et alg{\'e}briques, g{\'e}n{\'e}ralisations}.
\newblock Publications CRM, 2007.

\bibitem[GRV19]{GRV:Polyakovs_formulation_of_2D_bosonic_string_theory}
Colin Guillarmou, R{\'e}mi Rhodes, and Vincent Vargas.
\newblock Polyakov's formulation of {2D} bosonic string theory.
\newblock {\em Publ. Math. Inst. Hautes {\'E}tudes Sci.}, 130(1):111--185,
  2019.

\bibitem[Hen14]{Henriques:Conformal_field_theory_lectures}
Andr\'e Henriques.
\newblock Conformal field theory, 2014. \\
\newblock Lecture notes,
  \url{http://andreghenriques.com/Teaching/CFT-2014.pdf}.

\bibitem[Hua97]{Huang:2D_Conformal_geometry_and_VOAs}
Yi-Zhi Huang.
\newblock {\em Two-{D}imensional {C}onformal {G}eometry and {V}ertex {O}perator
  {A}lgebras}, volume 148 of {\em Progress in Mathematics}.
\newblock Birkh\"auser Basel, 1997.

\bibitem[Hua98]{Huang:Genus-zero_modular_functors_and_intertwining_operator_algebras}
Yi-Zhi Huang.
\newblock Genus-zero modular functors and intertwining operator algebras.
\newblock {\em Int. J.~Math.}, 9(7):845--863, 1998.

\bibitem[Kon87]{Kontsevich:Virasoro_and_Teichmuller_spaces}
Maxim Kontsevich.
\newblock The {V}irasoro algebra and {T}eichm{\"u}ller spaces.
\newblock {\em Funct. Anal. Appl.}, 21(2):156--157, 1987.

\bibitem[Kon03]{Kontsevich:CFT_SLE_and_phase_boundaries}
Maxim Kontsevich.
\newblock {CFT}, $\mathrm{SLE}$, and phase boundaries.
\newblock In {\em Oberwolfach Arbeitstagung}, 2003.

\bibitem[KS07]{Kontsevich-Suhov:On_Malliavin_measures_SLE_and_CFT}
Maxim Kontsevich and Yuri Suhov.
\newblock On {M}alliavin measures, $\mathrm{SLE}$, and {CFT}.
\newblock {\em P. Steklov I. Math.}, 258(1):100--146, 2007.

\bibitem[KW09]{Khesin-Wendt:The_geometry_of_infinite-dimensional_groups}
Boris Khesin and Robert Wendt.
\newblock {\em The geometry of infinite-dimensional groups}.
\newblock Series of Modern Surveys in Mathematics. Springer-Verlag, 2009.

\bibitem[Law09]{Lawler:Partition_functions_loop_measure_and_versions_of_SLE}
Gregory~F. Lawler.
\newblock Partition functions, loop measure, and versions of {$\mathrm{SLE}$}.
\newblock {\em J.~Stat. Phys.}, 134(5-6):813--837, 2009.

\bibitem[LSW03]{LSW:Conformal_restriction_the_chordal_case}
Gregory~F. Lawler, Oded Schramm, and Wendelin Werner.
\newblock Conformal restriction: the chordal case.
\newblock {\em J.~Amer. Math. Soc.}, 16(4):917--955, 2003.

\bibitem[Mai21]{Maibach:Master_thesis}
Sid Maibach.
\newblock Real and {Complex} {Valued} {One}-{Dimensional} {Genus}-{Zero}
  {Modular} {Functors} in {Conformal} {Field} {Theory}.
\newblock Master's thesis, University of Bonn, 2021.

\bibitem[Mal99]{Malliavin:Canonic_diffusion_above_the_diffeomorphism_group_of_the_circle}
Paul Malliavin.
\newblock The canonic diffusion above the diffeomorphism group of the circle.
\newblock {\em C. R. Acad. Sci. Paris S{\'e}r. I Math.}, 329(4):325--329, 1999.

\bibitem[Mil85]{Milnor:Remarks_on_infinite_dimensional_Lie_groups}
John~W. Milnor.
\newblock Remarks on infinite dimensional {L}ie groups.
\newblock In {\em 40th Les Houches Summer School on Theoretical Physics:
  Relativity, Groups and Topology}, volume~2, pages 1007--1057, 1985.

\bibitem[Nee05]{Neeb:Infinite-Dimensional_Lie_Groups}
Karl-Hermann Neeb.
\newblock {\em Infinite-dimensional {L}ie groups}.
\newblock 3rd cycle. Monastir (Tunisie), 2005.
\newblock \url{https://cel.hal.science/cel-00391789}.

\bibitem[Ner90]{Neretin:Holomorphic_extensions_of_representations_of_the_group_of_diffeomorphisms_of_the_circle}
Yu~A. Neretin.
\newblock Holomorphic extensions of representations of the group of
  diffeomorphisms of the circle.
\newblock {\em Math. USSR Sbornik}, 67(1):75--97, 1990.

\bibitem[Ner96]{Neretin:book}
Yu~A. Neretin.
\newblock {\em Categories of symmetries and infinite-dimensional groups}.
\newblock Number new ser., 16 in Oxford Science Publications. Clarendon Press;
  Oxford University Press, 1996.

\bibitem[NS95]{Nag-Sullivan:Teichmuller_theory_and_the_universal_period_mapping_via_quantum_calculus_and_H_half_space_on_the_circle}
Subhashis Nag and Dennis Sullivan.
\newblock Teichm\"uller theory and the universal period mapping via quantum
  calculus and the {$H^{1/2}$} space on the circle.
\newblock {\em Osaka J.~Math.}, 32(1):1--34, 1995.

\bibitem[NV90]{Nag-Verjovsky:DiffS1_and_Teichmuller_spaces}
Subhashis Nag and Alberto Verjovsky.
\newblock {$\mathrm{Diff}(S^1)$} and the {T}eichm\"uller spaces.
\newblock {\em Comm. Math. Phys.}, 130(1):123--138, 1990.

\bibitem[OPS88]{OPS:Extremals_of_determinants_of_Laplacians}
Brad Osgood, Ralph Phillips, and Peter Sarnak.
\newblock Extremals of determinants of {L}aplacians.
\newblock {\em J.~Funct. Anal.}, 80(1):148--211, 1988.

\bibitem[Pol81]{Polyakov:Quantum_geometry_of_bosonic_strings}
Alexander~M. Polyakov.
\newblock Quantum geometry of bosonic strings.
\newblock {\em Phys. Lett. B}, 103(3):207--210, 1981.

\bibitem[PW24]{Peltola-Wang:LDP}
Eveliina Peltola and Yilin Wang.
\newblock Large deviations of multichordal {$\mathrm{SLE}_{0+}$}, real rational
  functions, and zeta-regularized determinants of {L}aplacians.
\newblock {\em J.~Eur. Math. Soc.}, 26(2):469--535, 2024.

\bibitem[Rad03]{Radnell:PhD}
David Radnell.
\newblock {\em Schiffer variation in {T}eichm\"ueller space, determinant line
  bundles, and modular functors}.
\newblock PhD thesis, Rutgers University, New Brunswick, NJ, USA, 2003.

\bibitem[RS71]{Ray-Singer:R-Torsion_and_the_Laplacian_on_Riemannian_manifolds}
Daniel~B. Ray and Isadore~M. Singer.
\newblock R-{T}orsion and the {L}aplacian on {R}iemannian manifolds.
\newblock {\em Adv. Math.}, 7(2):145--210, 1971.

\bibitem[RS12]{Radnell-Schippers:The_semigroup_of_rigged_annuli_and_the_Teichmueller_space_of_the_annulus}
David Radnell and Eric Schippers.
\newblock The semigroup of rigged annuli and the {T}eichm\"uller space of the
  annulus.
\newblock {\em J.~London Math. Soc.}, 86(2):312--342, 2012.

\bibitem[RSS17]{RSS:Quasiconformal_Teichmuller_theory_as_analytical_foundation_for_CFT}
David Radnell, Eric Schippers, and Wolfgang Staubach.
\newblock Quasiconformal {T}eichm{\"u}ller theory as an analytical foundation
  for two-dimensional conformal field theory.
\newblock In {\em Proceedings of the Conference on Lie Algebras, Vertex
  Operator Algebras, and Related Topics (Univ. Notre Dame, Indiana)}, volume
  695 of {\em Contemp. Math.}, pages 205--238. American Mathematical Society,
  2017.

\bibitem[RSSS21]{RSSS:Schiffer_operators_and_calculation_of_a_determinant_line_in_conformal_field_theory}
David Radnell, Eric Schippers, Mohammad Shirazi, and Wolfgang Staubach.
\newblock Schiffer operators and calculation of a determinant line in conformal
  field theory.
\newblock {\em New York J.~Math.}, 27(1):253--271, 2021.

\bibitem[Sch08]{Schottenloher:Mathematical_introduction_to_CFT}
Martin Schottenloher.
\newblock {\em A mathematical introduction to conformal field theory}, volume
  759 of {\em Lecture Notes in Physics}.
\newblock Springer-Verlag, Berlin Heidelberg, 2nd edition, 2008.

\bibitem[Seg88]{Segal:Definition_of_CFT}
Graeme Segal.
\newblock The definition of conformal field theory.
\newblock In {\em Differential geometrical methods in theoretical physics
  (Como, 1987)}, volume 250 of {\em NATO Adv. Sci. Inst. Ser. C Math. Phys.
  Sci.}, pages 165--171. Kluwer Acad. Publ., Dordrecht, 1988.

\bibitem[SH62]{Schiffer-Hawley:Connections_and_conformal_mapping}
Menahem~M. Schiffer and Newton~S. Hawley.
\newblock Connections and conformal mapping.
\newblock {\em Acta Math.}, 107(3-4):175--274, 1962.

\bibitem[TT06]{Takhtajan-Teo:Weil-Petersson_metric_on_the_universal_Teichmuller_space}
Leon~A. Takhtajan and Lee-Peng Teo.
\newblock Weil-{P}etersson metric on the universal {T}eichm\"uller space.
\newblock {\em Mem. Amer. Math. Soc.}, 183:1--119, 2006.

\bibitem[TV15]{Teschner-Vartanov:SUSY_gauge_theories_quantization_of_M_flat_and_CFT}
J{\"o}rg Teschner and Grigorii~S. Vartanov.
\newblock Supersymmetric gauge theories, quantization of
  {$\mathcal{M}_{\textrm{flat}}$}, and conformal field theory.
\newblock {\em Adv. Theor. Math. Phys.}, 19(1):1--135, 2015.

\bibitem[Wan19]{Wang:Equivalent_descriptions_of_the_Loewner_energy}
Yilin Wang.
\newblock Equivalent descriptions of the {L}oewner energy.
\newblock {\em Invent. Math.}, 218:573--621, 2019.

\bibitem[Wer08]{Werner:The_conformally_invariant_measure_on_self-avoiding_loops}
Wendelin Werner.
\newblock The conformally invariant measure on self-avoiding loops.
\newblock {\em J.~Amer. Math. Soc.}, 21(1):137--169, 2008.

\bibitem[Zha21]{Zhan:SLE_loop_measures}
Dapeng Zhan.
\newblock {SLE} loop measures.
\newblock {\em Probab. Theory Related Fields}, 179(1-2):345--406, 2021.

\end{thebibliography}

\end{document}